\documentclass[10pt]{article}

\usepackage[T1]{fontenc}
\usepackage[left=2.8cm, right=2.8cm, top=3.5cm, bottom=3.5cm]{geometry}

\usepackage[style=alphabetic]{biblatex}
\AtBeginBibliography{\small}
\usepackage[normalem]{ulem}
\usepackage{amsmath}
\usepackage{mathtools}
\usepackage[makeroom]{cancel}
\usepackage{amssymb}
\usepackage{amsthm}
\usepackage{tensor}
\usepackage{enumerate}
\usepackage{bbm}
\usepackage{float}
\usepackage{varwidth}
\usepackage{pictexwd, dcpic}
\usepackage{xcolor}
\usepackage{graphicx}
\usepackage{fancyhdr}
\usepackage{mathrsfs}
\usepackage{tikz}
\usetikzlibrary{decorations.pathmorphing}
\usepackage{scalerel,stackengine}
\usepackage{slashed}
\usepackage{caption}
\usepackage{setspace}
\usepackage{units}

\usepackage{hyperref}
\hypersetup{
    colorlinks,
    linkcolor={blue!60!black},
    citecolor={red!60!black},
    urlcolor={blue!80!black}
}
\usepackage{cleveref}

\renewcommand{\d}{\mathrm{d}}
\renewcommand{\leq}{\leqslant}
\renewcommand{\geq}{\geqslant}
\renewcommand{\epsilon}{\varepsilon}
\renewcommand{\Re}{\operatorname{Re}}

\renewcommand{\eth}{\text{\rm{\dh}}}
\newcommand{\thorn}{\text{\rm{\th}}}

\newcommand{\R}{\mathbb{R}}
\newcommand{\la}{\lesssim}

\newcommand{\scri}{\mathscr{I}}

\theoremstyle{plain}
\newtheorem{theorem}{Theorem}
\newtheorem*{theorem*}{Theorem}
\newtheorem{lemma}[theorem]{Lemma}
\newtheorem{proposition}[theorem]{Proposition}
\newtheorem{corollary}[theorem]{Corollary}
\newtheorem*{corollary*}{Corollary}

\theoremstyle{definition}
\newtheorem{definition}[theorem]{Definition}

\theoremstyle{remark}
\newtheorem{remark}[theorem]{Remark}

\numberwithin{theorem}{section}
\numberwithin{equation}{section}

\newlength{\LETTERheight}
\AtBeginDocument{\settoheight{\LETTERheight}{I}}
\newcommand*{\longleadsto}[1]{\ \raisebox{0.24\LETTERheight}{\tikz \draw [->,
line join=round,
decorate, decoration={
    zigzag,
    segment length=4,
    amplitude=.9,
    post=lineto,
    post length=2pt
}] (0,0) -- (#1,0);}\ }
\DeclareFontFamily{U}{mathx}{\hyphenchar\font45}
\DeclareFontShape{U}{mathx}{m}{n}{
      <5> <6> <7> <8> <9> <10>
      <10.95> <12> <14.4> <17.28> <20.74> <24.88>
      mathx10
      }{}
\DeclareSymbolFont{mathx}{U}{mathx}{m}{n}
\DeclareFontSubstitution{U}{mathx}{m}{n}
\DeclareMathAccent{\widecheck}{0}{mathx}{"71}
\DeclareMathAccent{\wideparen}{0}{mathx}{"75}

\DeclareFontFamily{U}{MnSymbolC}{}
\DeclareSymbolFont{MnSyC}{U}{MnSymbolC}{m}{n}
\DeclareFontShape{U}{MnSymbolC}{m}{n}{
    <-6>  MnSymbolC5
   <6-7>  MnSymbolC6
   <7-8>  MnSymbolC7
   <8-9>  MnSymbolC8
   <9-10> MnSymbolC9
  <10-12> MnSymbolC10
  <12->   MnSymbolC12}{}
\DeclareMathSymbol{\intprod}{\mathbin}{MnSyC}{'270}
\newcommand*{\defeq}{\mathrel{\vcenter{\baselineskip0.5ex \lineskiplimit0pt
                     \hbox{\scriptsize.}\hbox{\scriptsize.}}}%
                     =}

\newcommand*{\eqdef}{=\mathrel{\vcenter{\baselineskip0.5ex \lineskiplimit0pt
                     \hbox{\scriptsize.}\hbox{\scriptsize.}}}%
                     }
\DeclareFontFamily{U}{BOONDOX-calo}{\skewchar\font=45 }
\DeclareFontShape{U}{BOONDOX-calo}{m}{n}{
  <-> s*[1.05] BOONDOX-r-calo}{}
\DeclareFontShape{U}{BOONDOX-calo}{b}{n}{
  <-> s*[1.05] BOONDOX-b-calo}{}
\DeclareMathAlphabet{\mathcalboondox}{U}{BOONDOX-calo}{m}{n}
\SetMathAlphabet{\mathcalboondox}{bold}{U}{BOONDOX-calo}{b}{n}
\DeclareMathAlphabet{\mathbcalboondox}{U}{BOONDOX-calo}{b}{n}
\Crefformat{equation}{(#2#1#3)}
\Crefrangeformat{equation}{(#3#1#4)--(#5#2#6)}
\Crefmultiformat{equation}{(#2#1#3)}%
{ and~(#2#1#3)}{, (#2#1#3)}{ and~(#2#1#3)}

\DeclareMathOperator{\dvol}{dv}

\newcounter{mnotecount}%[section]

\newcommand{\mnotex}[1]%{}
{\protect{\stepcounter{mnotecount}}$^{\mbox{\footnotesize $\bullet$\themnotecount}}$ 
\marginpar{%\color{red}%
\raggedright\tiny\em
$\!\!\!\!\!\!\,\bullet$\themnotecount: #1} }

\newcommand{\hl}{\hat{l}}
\newcommand{\hn}{\hat{n}}
\newcommand{\hm}{\hat{m}}
\newcommand{\bhm}{\bar{\hat{m}}}
\title{Conformal Scattering of Maxwell Potentials}
\author{Jean-Philippe Nicolas\footnote{Electronic address: \texttt{jean-philippe.nicolas@univ-brest.fr}}~ and Grigalius Taujanskas\footnote{Electronic address: \texttt{taujanskas@dpmms.cam.ac.uk}}}

\begin{document}

\maketitle

\begin{abstract} We construct a complete conformal scattering theory for finite energy Maxwell potentials on a class of curved, asymptotically flat spacetimes with prescribed smoothness of null infinity and a non-zero ADM mass. In order to define the full set of scattering data, we construct a Lorenz-like gauge which makes the field equations hyperbolic and non-singular up to null infinity, and reduces to an intrinsically solvable ODE on null infinity. We develop a method to solve the characteristic Cauchy problem from this scattering data based on a theorem of H\"ormander. In the case of Minkowski space, we further investigate an alternative formulation of the scattering theory by using the Morawetz vector field instead of the usual timelike Killing vector field.
\end{abstract}

%\tableofcontents

\section{Introduction}

The study of scattering is crucial to the understanding of both non-perturbative aspects of $S$-matrices arising in quantum field theory, and the asymptotic behaviour of classical fields and spacetimes in general relativity \cite{Strocchi2013}. For instance, from the pioneering works of Dimock and Kay \cite{Dimock1985,DimockKay1986,DimockKay1987} and Bachelot \cite{Bachelot1991,Bachelot1994,Bachelot1997,Bachelot1999} on the Schwarzschild metric using spectral theory, and the more recent studies on rotating black hole backgrounds \cite{GeorgescuGerardHafner2014,GeorgescuGerardHafner2013,GeorgescuGerardHafner2015,HafnerBesset2021}, to the work of Dafermos, Rodnianski, Shlapentokth-Rothman and others using the vector field method \cite{DafermosRodnianskiShlapentokhRothman2014,Besset2021,Masaood2020,Alford2020,AngelopoulosAretakisGajic2020}, scattering theory has been instrumental in studying questions of decay rates, the stability of spacetimes, and the Hawking effect. \emph{Conformal} scattering emerges as a combination of Penrose's ideas to apply the tools of conformal geometry in the setting of general relativity \cite{Penrose1963,Penrose1965}, the classical scattering theory of Lax and Phillips \cite{LaxPhillips1964,LaxPhillips1967}, and Friedlander's work on radiation fields \cite{Friedlander1962,Friedlander1964,Friedlander1967}. Here, \emph{null infinity}, $\scri$---a null hypersurface composed of all endpoints of inextendible null geodesics in the spacetime---is brought to a finite location using a conformal rescaling of the metric. Asymptotically, this scaling coincides with the scaling which returns Friedlander's radiation field, and scattering theory is reinterpreted as the characteristic Cauchy, or Goursat, problem from $\scri$. A key ingredient is that massless fields enjoy good conformal covariance properties, and so one is able to work with field equations both in physical and rescaled spacetimes, as suited. The construction of a scattering operator of this kind was first performed by Friedlander \cite{Friedlander1980} for the wave equation on a class of static, asymptotically flat, but not necessarily Einstein, Lorentzian manifolds admitting a smooth conformal compactification, including at $i^0$. Friedlander observed that the Lax--Phillips scattering theory could be reinterpreted as the resolution of a Goursat problem in the compactified spacetime, which---on the curved backgrounds mentioned above---enabled him to perform a scattering construction in the conformal picture and show its equivalence to the analytic ingredients of the Lax--Phillips theory. Baez, Segal and Zhou \cite{BaezSegalZhou1990} subsequently extended the construction to a nonlinear wave equation on flat spacetime. The later work of Mason and Nicolas \cite{MasonNicolas2004,MasonNicolas2007} reformulated the conformal scattering construction in terms of Hörmander's approach to the resolution of the Goursat problem \cite{Hormander1990}, which used energy estimates and compactness arguments. As a result of the flexibility of the method, Mason and Nicolas were able to extend the construction to fields of spin $0$, $\nicefrac{1}{2}$, and $1$ evolving in the background of a large class of curved, non-stationary spacetimes. Since then, linear scattering processes have been studied conformally on exteriors of black hole spacetimes \cite{Pham2022,Mokdad2017,NicolasPham2018,Pham2020a,Pham2020b,Nicolas2013} as well as in the interior of black holes \cite{Mokdad2022,KehleShlapentokhRothman2018}. Further work has also been done on nonlinear fields \cite{Joudioux2019,Joudioux2010,Taujanskas2018}.

An important distinction between the constructions of \cite{Friedlander1980,BaezSegalZhou1990} and the series of works spurred by \cite{MasonNicolas2004} is the treatment of spatial infinity $i^0$, the endpoint of all inextendible spacelike geodesics. It is by now well-known that a point compactification of $i^0$ must generically result in a singularity in the Weyl tensor, the only exception being the case of Minkowski space. As a result, Friedlander's decay assumptions \cite{Friedlander1980} for a smooth compactification at $i^0$ excluded non-trivial solutions to the Einstein equations. The work of Mason and Nicolas therefore introduces a \emph{partial} compactification of the spacetime in which $\scri$ is brought to a finite distance but $i^0$ is left at infinity, and treats the region near $i^0$ separately. It is worth noting that near $i^0$, this conformal scale is \emph{the same}, at least in the case of Minkowski space, as that which had previously been used by Friedrich to construct the so-called cylinder at spatial infinity \cite{Friedrich1998a}; a key point there is a judicious choice of coordinates (which we do not adopt in this paper) which blows up $i^0$ to a $(2+1)$-dimensional submanifold in order to allow a more detailed analysis of the asymptotics at $i^0$. It has recently been observed that Friedrich's formalism is closely related to Ashtekar and Hansen's earlier notion of the so-called hyperboloid at spatial infinity \cite{AshtekarHansen1978,MohamedKroon2021}. One further expects that both of these frameworks are also closely related to the more recent celebrated work of Hintz and Vasy \cite{HintzVasy2020} on the stability of Minkowski space.

In the present paper we construct a complete conformal scattering theory for Maxwell \emph{potentials} on a class of non-stationary curved spacetimes which may contain matter. Combined with \cite{MasonNicolas2004}, our construction settles, on a large class of spacetimes, a conjecture made by Geroch \cite{EspositoWitten1977} in 1976. We further obtain precise decay rates for all components of finite-energy potentials towards and along $\scri$, and towards $i^0$. There are several reasons it is of interest to study the scattering of Maxwell potentials (in contrast to fields). For instance, if one is interested in scattering from the perspective of asymptotic symmetries, the electromagnetic memory effect may be expressed at the level of the potential, so one is led to understanding the scattering matrix for the potential. Moreover, it is essential to understand potential scattering in order to have any hope of developing a scattering theory for nonlinear Yang--Mills fields, where the field is no longer gauge-invariant and the potential becomes fundamental. Indeed, even in the case of nonlinear abelian fields (such as the Maxwell--Klein--Gordon system), the potential plays a fundamental role and must be handled. In fact, it may be argued that even in the abelian Maxwell case the potential, rather than the field, ought to be treated as fundamental, as there exist physical situations in which the potential plays a role despite the field being zero, such as the Aharonov--Bohm effect. 

The difficulties in extending the constructions of \cite{MasonNicolas2004} to Maxwell potentials are twofold. First, the question of gauge choice must be addressed, and second, the Goursat problem for the resulting equations must be solved. For the latter, we make use of Bär--Wafo's extension \cite{BarWafo2015} to spatially non-compact spacetimes of a theorem due to H\"ormander \cite{Hormander1990}, which ensures that a solution to the characteristic initial value problem for linear wave equations can be obtained with no loss of regularity. We solve the Goursat problem in two stages, first solving in a neighbourhood of timelike infinity, and then near the rest of null infinity. For our spaces of scattering data the solution near $i^\pm$ is pure gauge, but non-zero at $i^\pm$. For the former problem, in the general case we construct a Lorenz-like gauge which involves an additional residual gauge fixing condition on null infinity, and allows one to define a complete set of scattering data for the potential. Roughly, the residual gauge fixing condition corresponds to the vanishing of the transverse derivative of the component of the potential parallel to the generators of $\scri$. This Lorenz-like gauge reduces to a first-order ODE on $\scri$, which may be integrated and yields an integration constant. We believe this constant to be related to the memory effect. 

We work on a class of background spacetimes which we refer to as CSCD spacetimes. These spacetimes are constructed using the initial data gluing theorems of Corvino, Schoen, Chrus\'ciel and Delay \cite{Corvino2000,CorvinoSchoen2006,ChruscielDelay2002,ChruscielDelay2003} and Friedrich's theorem for the semi-global stability of Minkowski space \cite{Friedrich1986}, possess regular (but not $\mathcal{C}^\infty$) null and timelike infinities, and are diffeomorphic to the Schwarzschild spacetime in a neighbourhood of $i^0$.

This paper is divided into two parts. In the first part (\Cref{sec:potentialscatteringpartiallyflat}) we construct a conformal scattering theory for Maxwell potentials on Minkowski space. Even though a complete compactification is available, here we use a partial compactification in which $i^0$ remains at an infinite distance. In \Cref{sec:flatMaxwellequationsofmotion} we fix the gauge and derive an implied condition on $\scri$, which is necessary to recover the full set of scattering data. Due to the triviality of the background, here one is able to choose the gauge "greedily" and impose the temporal, Lorenz and Coulomb gauges simultaneously. In \Cref{sec:Maxwellpotentialsflatscattering} we construct function spaces of initial and scattering data on Minkowski space and prove the existence and invertibility of a scattering operator. In \Cref{sec:alternativeformulations} we also consider an alternative formulation of the scattering theory using the Morawetz vector field $K_0 = (t^2 + r^2) \partial_t + 2 tr \partial_r$ as the multiplier in place of the Killing vector field $\partial_t$. We show that this gives a scattering theory which is in some sense strictly stronger (that is, the space of scattering data is strictly smaller). In the second part (\Cref{sec:curvedspacetimes}) we extend the constructions to CSCD spacetimes. A first step is to prove two-way energy estimates between initial data and scattering data, which are given in \Cref{sec:Maxwell_energy_estimates_proof}. Then, in order to define the required gauge, in \Cref{sec:structure_of_scri} we first construct a conformal scale in which $\scri$ is as flat as possible, and define the gauge condition near $\scri$ with respect to this scale. We also analyse the gauge near the initial surface, and interpolate between the two conditions in the bulk of the spacetime. Finally, we construct spaces of scattering data and prove the existence and invertibility of the scattering operator in \Cref{sec:scattering_operator_curved_spacetimes} using a similar approach to that on Minkowski space. In the general case our space of scattering data for the potential turns out to be isomorphic to $\dot{H}^1(\mathbb{R}; L^2(\mathbb{S}^2))$. The space of initial data for the potential is slightly more complicated to describe and is given in \Cref{sec:space_of_initial_data_curved}. Here, in order to recover sufficient regularity for the potential, we make an assumption on the Ricci curvature of a Cauchy surface in the spacetime, precisely that its $L^\infty$ norm is not too large. Roughly, the main results of the paper are the following.
\begin{theorem*}[Scattering theory on Minkowski space]
    Let $\Sigma$ be the standard initial Cauchy surface in Minkowski space. A finite energy solution to Maxwell's equations on Minkowski space admits the Coulomb, temporal, and Lorenz gauges simultaneously, and there exist bounded, invertible linear operators
    \[ \mathfrak{T}^\pm : \dot{H}^1_C(\Sigma) \oplus L^2_C(\Sigma) \longrightarrow \dot{H}^1(\mathbb{R};L^2(\mathbb{S}^2)), \]
    corresponding to the future/past Cauchy development of Maxwell's equations, which map finite energy Maxwell potential initial data on $
    \Sigma$ to finite energy Maxwell potential characteristic data on null infinity. The resulting scattering operator $\mathscr{S} = \mathfrak{T}^+ \circ (\mathfrak{T}^-)^{-1}$ is an isomorphism of Hilbert spaces.
\end{theorem*}
\begin{theorem*}[Scattering theory on CSCD spacetimes] Let $\Sigma$ be an initial Cauchy surface in a Corvino--Schoen--Chru\'sciel--Delay spacetime $\mathcalboondox{M}$ which is sufficiently close to Minkowski space. Then a finite energy solution to Maxwell's equations on $\mathcalboondox{M}$ admits Lorenz-like gauges near $\Sigma$ and null infinity, and there exist bounded, invertible linear operators
    \[ \mathfrak{T}^\pm : \dot{H}^1_C(\Sigma)^\text{curl} \oplus L^2(\Sigma) \longrightarrow \dot{H}^1(\mathbb{R}; L^2(\mathbb{S}^2)), \]
corresponding to the future/past Cauchy development of Maxwell's equations on $\mathcalboondox{M}$, which map finite energy Maxwell potential initial data on $\Sigma$ to finite energy Maxwell potential characteristic data on null infinity. The resulting scattering operator $\mathscr{S} = \mathfrak{T}^+ \circ (\mathfrak{T}^-)^{-1}$ is an isomorphism of Hilbert spaces.
\end{theorem*}
In the above theorems, the subscript ${}_C$ denotes the spaces of gauge-fixed initial data. Full descriptions are given in \Cref{sec:spaces_of_data_flat,sec:space_of_initial_data_curved}.

\section{Setup}

\subsection{Conventions and Notation}

Our conventions are consistent with Penrose \& Rindler \cite{spinorsandspacetime1,spinorsandspacetime2}. In particular, we work on 4-dimensional spacetimes $\mathcalboondox{M}$ with metric signature $(+,-,-,-)$, and the Riemann curvature tensor $\mathrm{R}^c_{\phantom{c}dab}$ is defined by $2 \nabla_{[a}\nabla_{b]} X^c = - \mathrm{R}^c_{\phantom{c}dab} X^d$. We denote the Weyl tensor by $\mathrm{C}_{abcd}$, and the trace-free part of the Ricci tensor by $\Phi_{ab}$. For a connection $\nabla_a$ (e.g. the Levi--Civita connection of a Lorentzian metric $g_{ab}$), we denote by $\Box = \nabla^a \nabla_a$ the associated wave operator. We will work with conformally related metrics such as $\hat{g}_{ab} = \Omega^2 g_{ab}$, and for the metric $\hat{g}_{ab}$ will denote the associated Levi--Civita connection by $\hat{\nabla}_a$, and the corresponding wave operator by $\widehat{\Box} = \hat{\nabla}^a \hat{\nabla}_a$. Given a conformal factor $\Omega$ relating $g_{ab}$ and $\hat{g}_{ab}$, we will frequently employ the notation $\Upsilon_a = \partial_a \log \Omega$, and will use the symbol $\approx$ to denote equality on \emph{null infinity}, i.e. where $\Omega = 0$ for an appropriate choice of $\Omega$. For a spacelike hypersurface $(\Sigma, h_{ab})$ of $\mathcalboondox{M}$, we will denote by $\mathcal{C}^k(\Sigma)$ and $H^k(\Sigma)$ the standard spaces of functions on $\Sigma$ which have $k$ continuous derivatives and $k$ derivatives in $L^2(\Sigma)$, respectively. We will use the same notation, e.g. $L^2(\Sigma)$, to refer to the space $L^2(\Sigma; S)$ of sections of a vector bundle $S \to \Sigma$ over $\Sigma$, where in each case the vector bundle will be clear from context. We will denote by $\dvol$ the 4-volume form associated to the spacetime metric $g_{ab}$, by $\widehat{\dvol}$ the 4-volume form of $\hat{g}_{ab}$, and by $\dvol_\Sigma$ the volume form of a hypersurface $(\Sigma, h_{ab})$.

We will make use of the Newman--Penrose and Geroch--Held--Penrose (GHP) formalisms throughout the paper; the reader entirely unfamiliar with the notation might like to consult \cite{spinorsandspacetime1,spinorsandspacetime2}. On Minkowski space we choose a \emph{Newman--Penrose (NP) tetrad} $(l^a, m^a, \bar{m}^a, n^a)$ of null vectors, given in standard radial coordinates $(t, r, \theta, \phi)$ by
\begin{equation} \label{firsttetrad} n^a = \frac{1}{2}\left( \partial_t - \partial_r \right), \quad l^a = \partial_t + \partial_r, \quad m^a = \frac{1}{\sqrt{2}r} \left( \partial_\theta + \frac{i}{\sin \theta} \partial_\phi \right), \quad \bar{m}^a = \frac{1}{\sqrt{2}r} \left( \partial_\theta - \frac{i}{\sin \theta} \partial_\phi \right).
\end{equation}
These satisfy $l^a n_a = 1 = - m^a \bar{m}_a$ and $l^a m_a = l^a \bar{m}_a = n^a m_a = n^a \bar{m}_a = 0$. The integral curves of $l^a$ and $n^a$ trace out, respectively, the outgoing and incoming radial null geodesics, and $m^a$ and $\bar{m}^a$ span the tangent space of spacelike spheres at each point. The directions of $n^a$ and $l^a$ are shown on the Penrose diagram of Minkowski space below (\Cref{fig:Minkowski_Penrose_diagram}). Here, the points $i^\pm$ denote future and past timelike infinities---the endpoints of all inextendible future- and past-directed timelike geodesics, $i^0$ denotes spatial infinity, the endpoint of all inextendible spacelike geodesics, and $\scri^\pm$---future and past null infinities---are surfaces consisting of all endpoints of future- and past-directed null geodesics.

\begin{figure}[h]
\centering
	\begin{tikzpicture}
	\centering
	\node[inner sep=0pt] (conformalfactorMinkowski) at (3.4,0)
    	{\includegraphics[width=.18\textwidth]{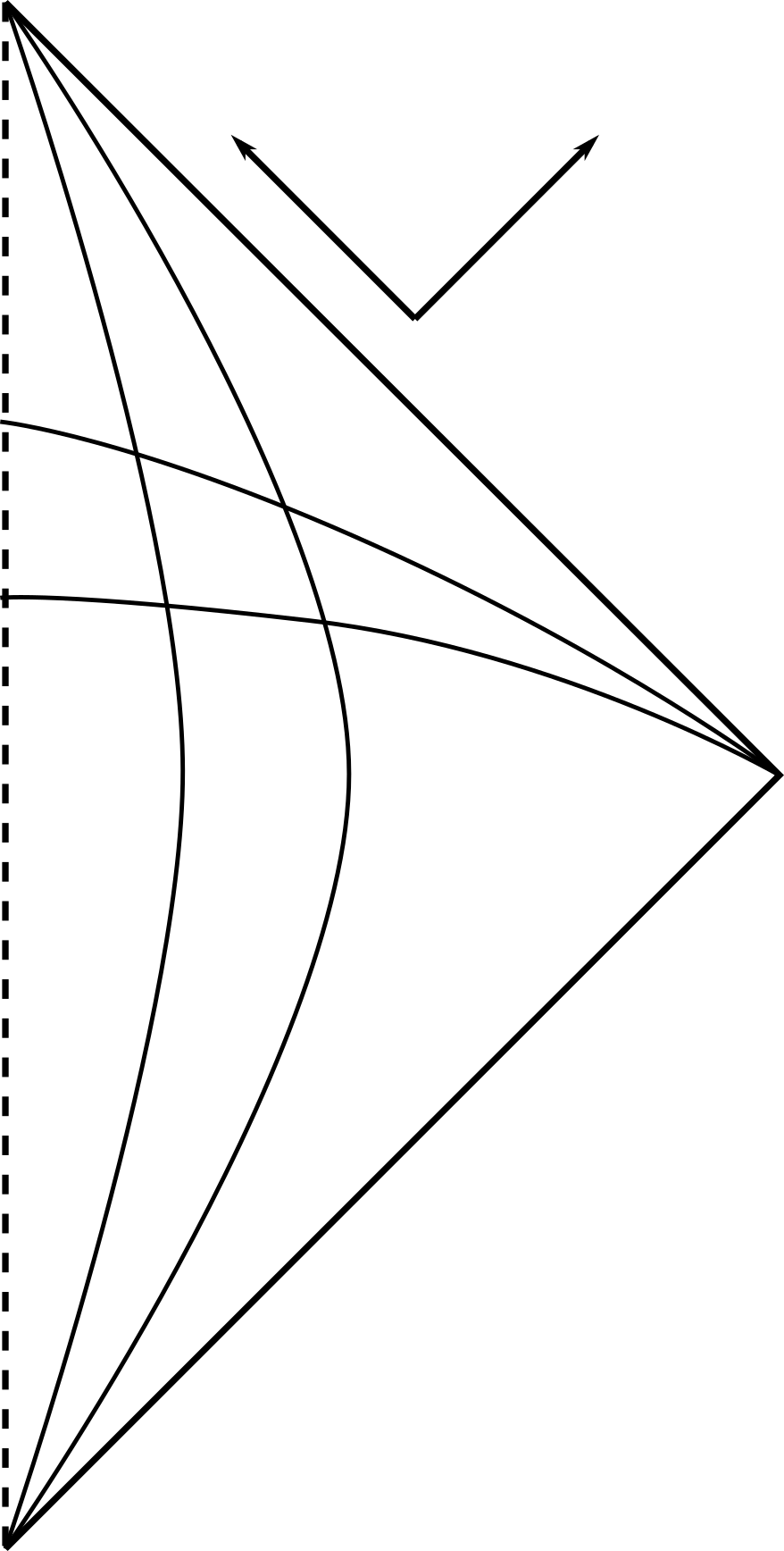}};

	\node[label={[shift={(2,2.95)}]$i^+$}] {};
 	\node[label={[shift={(2,-3.6)}]$i^-$}] {};
	\node[label={[shift={(4.4,-1.4)}]$\scri^-$}] {};
	\node[label={[shift={(4.4,0.7)}]$\scri^+$}] {};
	\node[label={[shift={(1.7,0)}, rotate={90}]$r = 0$}] {};
	\node[label={[shift={(5.2,-0.3)}]$i^0$}] {};

 	\node[label={[shift={(4.23,1.78)}]\footnotesize$l^a$}] {};
 	\node[label={[shift={(3.15,2.1)}]\footnotesize$n^a$}] {};

	\end{tikzpicture}
	\captionsetup{width=0.85\textwidth}
	\caption{The Penrose diagram of Minkowski space showing surfaces of constant $r$ and surfaces of constant $t$.} \label{fig:Minkowski_Penrose_diagram}	
\end{figure}
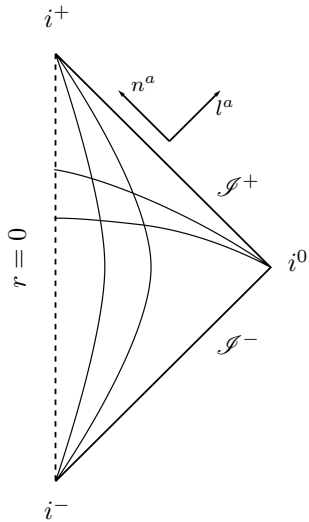

In curved spacetimes, $l^a$ and $n^a$ will be similarly aligned with outgoing and incoming radial null geodesics. Note that $n^a$ is tangent to $\scri^+$, and in fact becomes a generator of $\scri^+$ at the conformal boundary. The directional derivatives along an NP tetrad $(l^a, m^a, \bar{m}^a, n^a)$ will be denoted by $(D, \delta, \bar{\delta}, \Delta)$ respectively, with $(\thorn, \eth, \bar{\eth}, \thorn')$ the corresponding spin-weighted directional derivatives. For the benefit of the reader unfamiliar with the GHP formalism, we give expressions for spin-weighted directional derivatives below. The way in which $(\thorn,\eth,\bar{\eth},\thorn')$ act on scalars depends on their \emph{weight}. Precisely, an NP tetrad $(l^a, m^a, \bar{m}^a, n^a)$ may be rescaled according to
\[ l^a \mapsto \lambda \bar{\lambda} l^a, \qquad m^a \mapsto \lambda \bar{\lambda}^{-1} m^a, \qquad \bar{m}^a \mapsto \lambda^{-1} \bar{\lambda} \bar{m}^a, \qquad n^a \mapsto \lambda^{-1} \bar{\lambda}^{-1} n^a \]
for any nowhere vanishing complex scalar field $\lambda$, leaving the orthognality relations of the tetrad and the metric $g_{ab} = l_a n_b + n_a l_b -m_a \bar{m}_b -\bar{m}_a m_b$ unchanged. A scalar (or tensor), say $\eta$, formed by contracting a spacetime tensor with elements of the NP tetrad therefore acquires weights under the above rescaling, say $\eta \mapsto \lambda^p \bar{\lambda}^q\eta$. We say that $\eta$ is a $(p,q)$-scalar (or -tensor). The spin-weighted directional derivatives are then defined by
\begin{align*}
    \thorn \eta & = (D - p \epsilon -q\bar{\epsilon})\eta, \\
    \thorn' \eta & = (\Delta - p \gamma - q\bar{\gamma}) \eta, \\
    \eth \eta & = (\delta - p \beta - q \bar{\alpha}) \eta, \\
    \eth' \eta & = (\bar{\delta} - p \alpha - q\bar{\beta} )\eta,
\end{align*}
where the definitions of the spin coefficients $(\epsilon, \gamma, \alpha, \beta)$ may be found in \cite{spinorsandspacetime2}.

\subsection{Background Spacetimes}

Let $(\mathcalboondox{M}, g_{ab})$ be a smooth globally hyperbolic four-dimensional spacetime diffeomorphic to $\mathbb{R}^4$. We will consider conformal rescalings of $g_{ab}$ of the form $\hat{g}_{ab} = \Omega^2 g_{ab}$ for suitable functions $\Omega : \mathcalboondox{M} \to \mathbb{R}$, and in order to distinguish $g_{ab}$ from $\hat{g}_{ab}$ will refer to $g_{ab}$ as the \emph{physical} metric and $\hat{g}_{ab}$ as the \emph{rescaled}, or \emph{unphysical}, metric. We perform orthogonal $3+1$ decompositions of the physical and rescaled metrics as follows. Since $\mathcalboondox{M}$ is globally hyperbolic, there exists a smooth time function $t: \mathcalboondox{M} \to \mathbb{R}$ such that $\nabla^a t$ is uniformly timelike on $\mathcalboondox{M}$, where $\nabla$ is the Levi--Civita connection of $g_{ab}$. The level sets $\{ \Sigma_t \}_t$ of $t$ define a uniformly spacelike foliation of $\mathcalboondox{M}$. Since $\mathcalboondox{M}$ is diffeomorphic to $\mathbb{R}^4$, each $\Sigma_t$ is diffeomorphic to some $\Sigma \simeq \mathbb{R}^3$, and the flow of the vector field $\nabla^a t$ effects the identification $\mathcalboondox{M} \simeq \mathbb{R}_t \times \Sigma$. The metric $g_{ab}$ then decomposes as
\[ \d s^2 = g_{ab} \, \d x^a \, \d x^b = N^2 \d t^2 - h, \]
where $h$ is a smooth Riemannian metric on $\Sigma_t$ for each $t$, and $N$ is a smooth non-vanishing \emph{lapse} function. The unit normal to the hypersurfaces $\Sigma_t$ is
\[ T^a = \frac{1}{N} \frac{\partial}{\partial t}, \quad \text{i.e.} \quad T_a \d x^a = N \d t, \]
so the metric can be written as 
\[ g_{ab} = T_a T_b - h_{ab}. \]
The Levi--Civita connection of $g_{ab}$ decomposes as 
\[ \nabla_a = T_a \nabla_T + \nabla_a^\perp, \]
where $\nabla_a^\perp = - h_a^b \nabla_b$ is the part of $\nabla_a$ orthogonal to $T^a$, $T^a \nabla_a^\perp = 0$. It is the $4$-dimensional covariant derivative $\nabla$ projected onto $\Sigma_t$, and differs from the Levi--Civita connection $\boldsymbol{\nabla}$ of $(\Sigma_t, h_{ab}(t))$ by the extrinsic curvature $\kappa_{ab}$ of $\Sigma_t$. Indeed,
\[ \nabla_a^\perp T_b = - h_a^c \nabla_c T_b = \kappa_{ab} = \kappa_{(ab)}, \]
so that for any $X_a$ such that $T^a X_a = 0$
\[ \boldsymbol{\nabla}_a X_b - \nabla_a^\perp X_b = \kappa_a^{\phantom{a}c} T_b X_c. \]
We also define the trace of the extrinsic curvature by
\[ \operatorname{tr} \kappa = \kappa^a_{\phantom{a}a} = - h^{ab} \nabla_a T_b. \]

A similar decomposition may be performed for the rescaled metric $\hat{g}_{ab}$. Here we choose a smooth time function $\tau$ such that $\hat{\nabla}^a \tau$ is uniformly timelike and such that $\tau(i^\pm) = \pm \tau_{\mathrm{max}}$, $0 < \tau_{\mathrm{max}} < \infty$, where $\hat{\nabla}$ is the Levi--Civita connection of $\hat{g}_{ab}$. The level surfaces $\{ \hat{\Sigma}_\tau \}_\tau$ of $\tau$ define a uniformly spacelike foliation of $\hat{\mathcalboondox{M}}$ such that the leaves $\hat{\Sigma}_\tau$ are transverse to $\scri$, and, as $\tau \to \pm \tau_{\mathrm{max}}$, the leaves $\hat{\Sigma}_\tau$ shrink to the points $i^\pm$. With respect to this foliation the rescaled metric decomposes as
\[ \hat{g}_{ab} = \hat{T}_a \hat{T}_b - \hat{h}_{ab}, \qquad \hat{g}_{ab} \,  \d x^a \, \d x^b = \hat{N}^2 \d \tau^2 - \hat{h}, \]
where $\hat{T}^a$ is the unit normal to $\hat{\Sigma}_\tau$ with respect to $\hat{g}_{ab}$, and $\hat{h}_{ab}$ is a smooth Riemannian metric on $\hat{\Sigma}_\tau$ for each $\tau$. As before, the Levi--Civita connection $\hat{\nabla}$ of $\hat{g}_{ab}$ decomposes as
\[ \hat{\nabla}_a = \hat{T}_a \hat{\nabla}_{\hat{T}} + \hat{\nabla}^\perp_a. \]

We assume that the functions $t$ and $\tau$ are such that the initial leaf of the rescaled foliation $\{ \hat{\Sigma}_\tau \}_\tau$ agrees with the initial leaf of the physical foliation $\{ \Sigma_t \}_t$, $\hat{\Sigma}_0 = \Sigma_0$. The vector fields $\hat{T}^a$ and $T^a$ are therefore parallel on $\hat{\Sigma}_0 = \Sigma_0 \eqdef \Sigma$, and the above decomposition of the metric gives the relation
\[ \hat{T}^a|_\Sigma = \Omega^{-1} T^a |_\Sigma. \]
We will also assume that the time derivative of the conformal factor vanishes on $\Sigma$,
\[ \partial_t \Omega|_\Sigma \propto \partial_\tau \Omega |_\Sigma = 0. \]
The uniformly spacelike foliation $\{ \Sigma_t \}_t$ of the physical spacetime extends to an asymptotically null foliation of the rescaled spacetime (we say that the leaves accumulate at $\scri^\pm$ as $t\rightarrow \pm\infty$). Indeed, the unit normal $T^a$ with respect to $g_{ab}$ has norm $\Omega^2$ with respect to $\hat{g}_{ab}$, which tends to zero as $\Omega \to 0$. Conversely, the uniformly spacelike foliation $\{ \hat{\Sigma}_\tau \}_\tau$ of the rescaled spacetime corresponds to a foliation of the physical spacetime by hyperboloids which are asymptotically null.

\begin{figure}[H]
\centering
\begin{minipage}{0.42\textwidth}
	\begin{tikzpicture}
	\centering
	\node[inner sep=0pt] (asnullfoliation) at (3.5,1)
    	{\includegraphics[width=.75\textwidth]{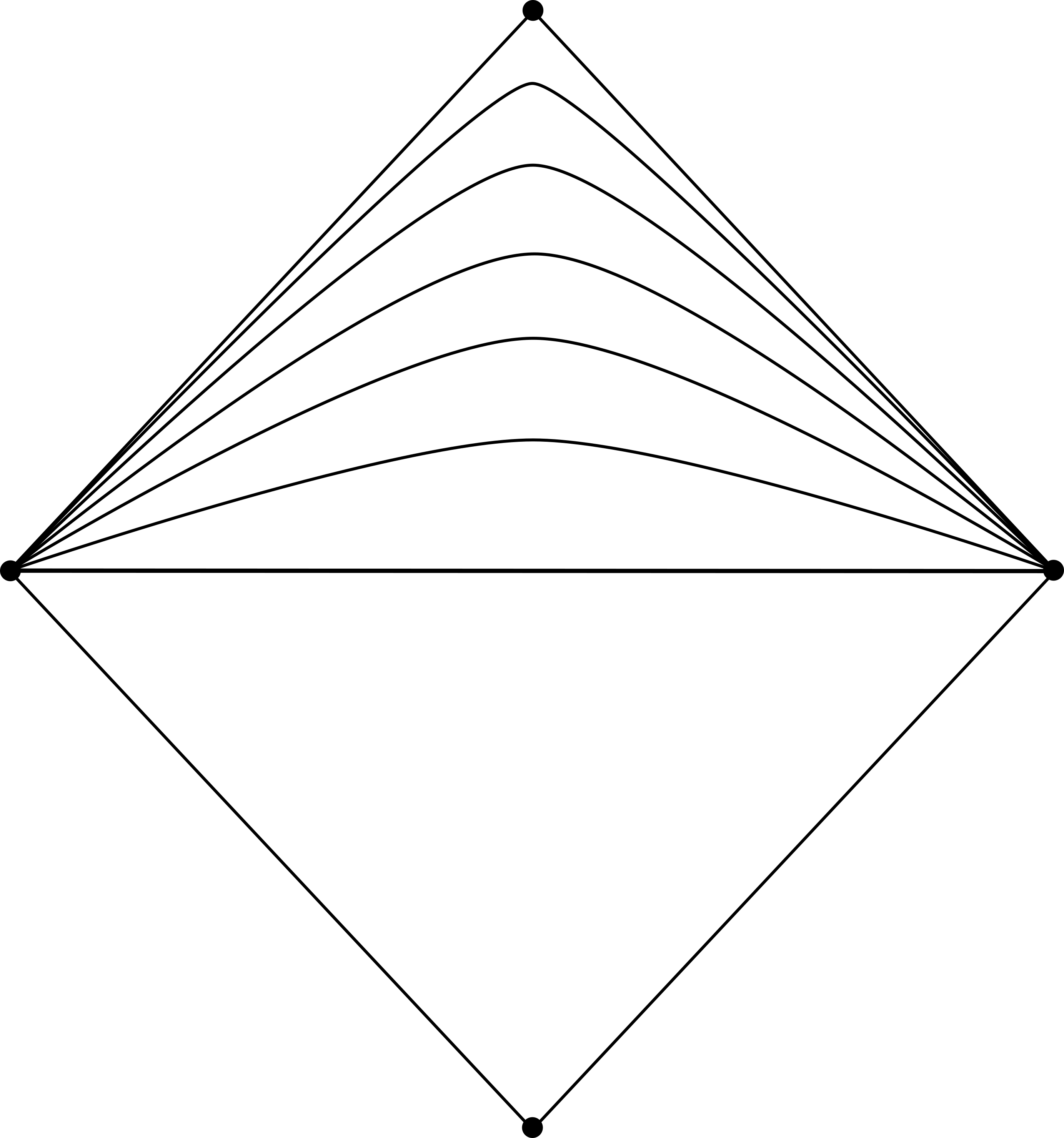}};

	\node[label={[shift={(3.62,3.7)}]$i^+$}] {};
	\node[label={[shift={(3.62,-2.5)}]$i^-$}] {};
    
	\node[label={[shift={(2, 2.4)}]$\scri^+$}] {};
	\node[label={[shift={(5.1, 2.4)}]$\scri^+$}] {};

	\node[label={[shift={(6.3,0.6)}]$i^0$}]{};
	\node[label={[shift={(0.7,0.6)}]$i^0$}]{};

	\node[label={[shift={(2, -1.3)}]$\scri^-$}] {};
	\node[label={[shift={(5.1, -1.3)}]$\scri^-$}] {};
	
	\node[label={[shift={(3.55, 0.2)}]$\Sigma$}] {};

    \draw[fill=white] (1.03,1) circle (1.5pt);
    \draw[fill=white] (5.96,1) circle (1.5pt);

	\end{tikzpicture}
	\caption{The asymptotically null foliation $\{ \Sigma_t \}_t$ of $\hat{\mathcalboondox{M}}$.} \label{fig:asnullfoliation}
\end{minipage}
\hspace{10pt}
\begin{minipage}{0.42\textwidth}
\centering
	\begin{tikzpicture}
	\centering
	\node[inner sep=0pt] (transversefoliation) at (3.5,1)
    	{\includegraphics[width=.75\textwidth]{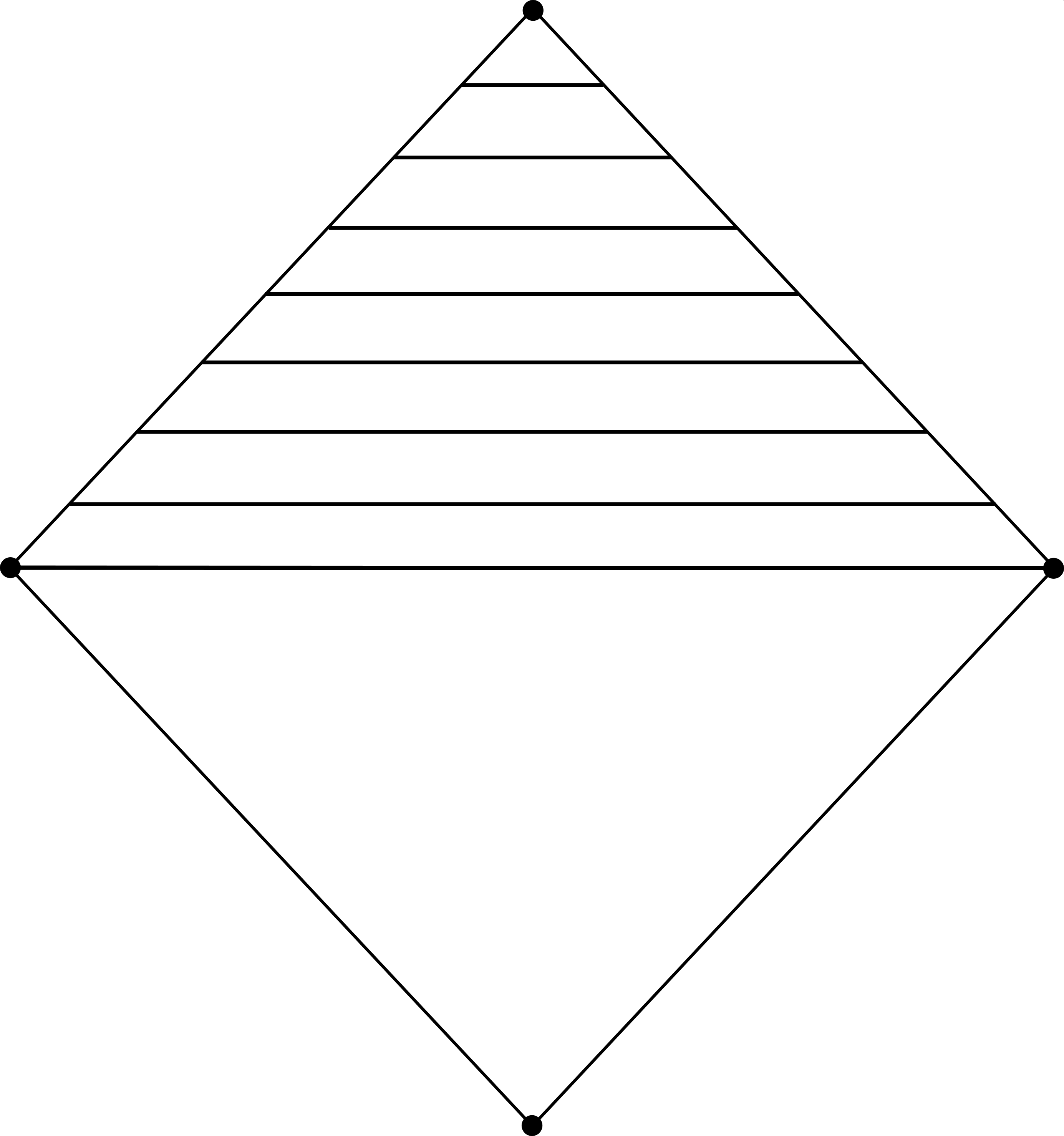}};

	\node[label={[shift={(3.62,3.65)}]$i^+$}] {};
	\node[label={[shift={(3.62,-2.45)}]$i^-$}] {};
    
	\node[label={[shift={(2, 2.4)}]$\scri^+$}] {};
	\node[label={[shift={(5.1, 2.4)}]$\scri^+$}] {};

	\node[label={[shift={(6.3,0.6)}]$i^0$}]{};
	\node[label={[shift={(0.7,0.6)}]$i^0$}]{};

	\node[label={[shift={(2, -1.3)}]$\scri^-$}] {};
	\node[label={[shift={(5.1, -1.3)}]$\scri^-$}] {};
	
	\node[label={[shift={(3.55, 0.2)}]$\Sigma$}] {};

    \draw[fill=white] (1.03,1) circle (1.5pt);
    \draw[fill=white] (5.96,1) circle (1.5pt);
	
	\end{tikzpicture}
	\caption{The foliation $\{\hat{\Sigma}_\tau \}_\tau$ of $\hat{\mathcalboondox{M}}$ whose leaves are transverse to $\scri^+$.} \label{fig:transversefoliation}	
\end{minipage}
\end{figure}

We define the projection onto $(\Sigma_t, h_{ab})$ of a $1$-form $A_a$ on $\mathcalboondox{M}$ by
\[ \mathbf{A}_\alpha \defeq - h^a_\alpha A_a. \]
The $\Sigma_t$-covariant derivative $\boldsymbol{\nabla}_\alpha$ applied to $\mathbf{A}_\beta$ is then given by
\[ \boldsymbol{\nabla}_\alpha \mathbf{A}_\beta =  h^b_\beta h^a_\alpha \nabla_a \mathbf{A}_b = - h^b_\beta h^a_\alpha \nabla_a ( h^c_b A_c ), \]
and more generally the $\Sigma_t$-covariant derivative of a tensor field $T^{a_1 \dots a_n}_{\phantom{a_1 \dots a_n}b_1 \dots b_m}$ is 
\[ \boldsymbol{\nabla}_\gamma T^{\alpha_1 \dots \alpha_n}_{\phantom{\alpha_1 \dots \alpha_n}\beta_1 \dots \beta_m} = (-1)^{n+m+1} h^c_\gamma h^{\alpha_1}_{a_1} \dots h^{\alpha_n}_{a_n} h^{b_1}_{\beta_1} \dots h^{b_m}_{\beta_m} \nabla_c T^{a_1 \dots a_n}_{\phantom{a_1 \dots a_n}b_1 \dots b_m}. \]
The factor of $(-1)^{n+m+1}$ is included to account for the successive changes of sign each time the projector $h^a_b$ is applied: note that $h^{ac} h_{cb} = - h^a_b = \delta^a_b - T^a T_b$. On tensors on $\Sigma_t$, $-h^a_b$ acts as $\delta^a_b$.

\subsection{Maxwell Fields and Potentials}

For a real 2-form $F = F_{ab} \, \d x^a \wedge \d x^b$ on $\mathcalboondox{M}$, the extremizers of the Lagrangian $\mathcal{L} = - \frac{1}{4} F_{ab} F^{ab}$ are called \emph{Maxwell fields}. The Euler--Lagrange equations satisfied by Maxwell fields---\emph{Maxwell's equations}---are given by
\begin{equation} \label{Maxwellsequations} \nabla^a F_{ab} = 0,
\end{equation}
together with the so-called \emph{Bianchi identity}
\begin{equation} \label{Bianchiidentity} \nabla_{[a} F_{bc]} = 0.
\end{equation}
The identity \eqref{Bianchiidentity} states that the 2-form $F$ is closed (similarly, the equation \eqref{Maxwellsequations} states that $F$ is co-closed), so by the Poincar\'e lemma $F$ is exact: there exists a real 1-form $A = A_a \, \d x^a$ such that $F = \d A$, or
\[ F_{ab} = 2 \partial_{[a} A_{b]} = \nabla_a A_b - \nabla_b A_a. \]
The 1-form $A$ is called the \emph{Maxwell potential}. Since $\d^2 = 0$, the Maxwell potential is only determined by the Maxwell field $F$ up to exact 1-forms $\d \chi$, so that the potentials $A$ and $A + \d \chi$ give rise to the same Maxwell field. This is the \emph{gauge freedom} in the Maxwell potential, so that a \emph{gauge transformation} is given by
\[ A_a \rightsquigarrow A_a + \nabla_a \chi. \]
The equations of motion \eqref{Maxwellsequations}, written in terms of the potential, become
\begin{equation} \label{Maxwellsequationspotential} \Box A_a - \nabla_b (\nabla_a A^a) + \mathrm{R}_{ab} A^a = 0.
\end{equation}
The canonical Maxwell stress-energy tensor is given by
\begin{equation} \label{stresstensor} \mathbf{T}_{ab} = - F_a^{\phantom{a}c} F_{bc} + \frac{1}{4} g_{ab} F_{cd} F^{cd},
\end{equation}
and is conserved on-shell.

A key feature of Maxwell's equations is that they are conformally invariant. That is, under the conformal transformation $g_{ab} \rightsquigarrow \hat{g}_{ab} = \Omega^2 g_{ab}$, if $A_a$ is chosen to have conformal weight zero, 
\[ \hat{A}_a = A_a, \]
then $\hat{F}_{ab} = F_{ab}$ and the physical and rescaled field equations are equivalent:
\[ \nabla^a F_{ab} = 0 = \nabla_{[a} F_{bc]}  \iff  \hat{\nabla}^a \hat{F}_{ab} = 0 = \hat{\nabla}_{[a} \hat{F}_{bc]}.  \]
This is clear, since the action $S = \int_{\mathcalboondox{M}} \mathcal{L} \, \mathrm{dv}$ is invariant:
\[ S = \int_{\mathcalboondox{M}} \mathcal{L} \, \mathrm{dv} = \int_{\hat{\mathcalboondox{M}}} - \frac{1}{4} \Omega^4 \hat{F}_{ab} \hat{F}^{ab} \Omega^{-4} \widehat{\mathrm{dv}} = \hat{S}. \]

\subsection{Maxwell Components}

With respect to an NP tetrad $(l^a, m^a, \bar{m}^a, n^a)$ (on any spacetime), we denote the two real and one complex component of the physical Maxwell potential $A_a$ and the three complex components of the physical Maxwell field $F_{ab}$ by
\[ \renewcommand{\arraystretch}{1.3} \left( \begin{array}{ccc} A_0 & A_1 & A_2 \\ F_0 & F_1 & F_2 \end{array} \right) = \left( \begin{array}{ccc} A_a l^a & A_a n^a & A_a m^a \\ F_{ab} l^a m^b & \frac{1}{2} F_{ab} (l^a n^b + \bar{m}^a m^b ) & F_{ab} \bar{m}^a n^b \end{array} \right). \]
We also denote by
\[ \mathfrak{a} \defeq T^a A_a, \]
and define the electric and magnetic fields with respect to the foliation $(\Sigma_t, h_{ab})$ by 
\[ \mathbf{E}_a \defeq T^b F_{ba} = - h^\alpha_a \mathbf{E}_\alpha \]
and
\[ \mathbf{B}_a \defeq \frac{1}{2} \epsilon_a^{\phantom{a}bc} F_{bc} = - \frac{1}{2} \epsilon_{\alpha \beta \gamma} h^\alpha_a h^{\beta b} h^{\gamma c} F_{bc} = - h^\alpha_a \mathbf{B}_\alpha , \] 
where $\epsilon_{abc}$ is the volume form of $h_{ab}$. The components of the rescaled Maxwell potential $\hat{A}_a$ and rescaled Maxwell field $\hat{F}_{ab}$ with respect to $(\hat{l}^a, \hat{m}^a, \bar{\hat{m}}^a, \hat{n}^a)$, as well as $\hat{\mathbf{E}}_a$ and $\hat{\mathbf{B}}_a$, are defined in the same way.

The components of $\mathbf{T}_{ab}$ with respect to the causal vectors of the tetrad $(l^a, m^a, \bar{m}^a, n^a )$ are given by
\[ \left( |F_0|^2, \, | F_1 |^2, \, |F_2|^2 \right) = \left( \textstyle{\frac{1}{2} \mathbf{T}_{ab} l^a l^b}, \, \textstyle{\frac{1}{2} \mathbf{T}_{ab} l^a n^b}, \, \textstyle{\frac{1}{2} \mathbf{T}_{ab} n^a n^b} \right), \]
and similarly for the rescaled stress-energy tensor $\hat{\mathbf{T}}_{ab}$ with respect to the rescaled tetrad $(\hat{l}^a,\hat{m}^a,\bar{\hat{m}}^a, \hat{n}^a)$.
The components of the Maxwell field $F_{ab}$ are given in terms of the components of the Maxwell potential $A_a$ by
\begin{align} 
\label{F0components} F_0 & = ( \thorn - \bar{\rho} ) A_2 + \kappa A_1 - ( \eth + \bar{\pi} ) A_0  - \sigma \bar{A}_2, \\
\label{F1components} F_1 & = \frac{1}{2} \left( - (\thorn' - \mu + \bar{\mu} ) A_0 + (\thorn + \rho - \bar{\rho})A_1 + (\bar{\eth} - \bar{\tau} - \pi ) A_2 - (\eth + \tau + \bar{\pi}) \bar{A}_2 \right)	, \\
\label{F2components} F_2 & = -( \thorn' + \bar{\mu} ) \bar{A}_2 - \nu A_0 + (\bar{\eth} + \bar{\tau} ) A_1 - \lambda A_2.
\end{align}

\section{Minkowski Space} \label{sec:potentialscatteringpartiallyflat}

\subsection{Partial Conformal Compactifications} \label{sec:partial_compactification_Minkowski}

In this section we work on Minkowski space $(\mathcalboondox{M} = \mathbb{R}^4, \eta)$,
\begin{equation} \label{Minkowskimetric} \eta = \d t^2 - \d r^2 - r^2 g_{\mathbb{S}^2}, \end{equation}
where $(t,r) \in \mathbb{R} \times [0, \infty)$.
We introduce two conformal scales. The first, which we refer to as the \emph{checked conformal scale}, allows us to define $\scri=\scri^+\cup \scri^-$ as the null boundary of the rescaled spacetime, but $i^\pm$ and $i^0$ remain at infinity. Although it provides an incomplete compactification, this scale is natural and useful because it preserves the symmetry associated to the timelike Killing vector field $\partial_t$. The second, the \emph{hatted conformal scale}, is obtained by modifying the checked conformal scale in a neighbourhood of timelike infinity in order to bring it to a finite distance. In this new scale, the boundary of the rescaled spacetime will be $\scri \cup  i^- \cup i^+$. It will be useful in situations when we will need to see $\scri^+$ as the regular backwards lightcone of $i^+$, e.g. when performing energy estimates and solving the Goursat problem. These compactifications are time-symmetric, so we shall mostly focus on future null infinity, although some details about past null infinity will also be given.

For the checked conformal scale, the conformal factor is a smooth positive function on $\mathcalboondox{M}$, depending only on $r$, such that $\Omega = 1/r$ for, say, $r>1$. In terms of the retarded Bondi or Eddington--Finkelstein coordinate $u=t-r$ and the inverted radial variable $R=1/r$, the Minkowski metric takes the form
\begin{equation} \label{MinkowskimetricRu} \eta = \d u^2 - \frac{2}{R^2} \, \d u \, \d R - \frac{1}{R^2} g_{\mathbb{S}^2}.	
\end{equation}
Applying the checked conformal rescaling gives, for $R<1$,
\begin{equation} \label{rescaledMinkowski} 
\check{\eta} \defeq R^2 \eta = R^2 \d u^2 - 2 \d u \, \d R - g_{\mathbb{S}^2}.	
\end{equation}
The rescaled metric is now regular at $R=0$, unlike the physical metric at $r=\infty$. In these coordinates, the set $\{R=0\}$ is the set of endpoints of outgoing radial null geodesics that are the $u$ coordinate lines, i.e. $\{ R = 0 \} = \scri^+$. Using the advanced coordinate $v=t+r$ instead of $u$, we have
\[ \check{\eta} = R^2 \d v^2 + 2 \, \d v \, \d R - g_{\mathbb{S}^2}, \]
and in these coordinates the set $\{R=0\}$ is now $\scri^-$. We define the compactified spacetime as
\[ \check{\mathcalboondox{M}} \defeq \mathcalboondox{M} \cup \scri^-\cup \scri^+ \, .\]
Note that timelike infinities $i^\pm$ and spacelike infinity $i^0$ are not brought to a finite distance in this scale. This can be seen by observing that the $u$ coordinate lines on $\scri^+$ are null geodesics for $\check{\eta}$, and that they admit $u$ as an affine parameter. Hence
\[ \scri^+ = \R_u \times \{R=0 \}\times \mathbb{S}^2 \quad \text{and} \quad \scri^- = \R_v \times \{R=0 \}\times \mathbb{S}^2\]
are infinite cylinders, diffeomorphic to $\R\times \mathbb{S}^2$. Future timelike infinity $i^+$ is the ``future end'' of $\scri^+$, given by $u=+\infty$, $R=0$, whereas spacelike infinity $i^0$ is its ``past end'', at $u=-\infty$, $R=0$. Similarly, past timelike infinity $i^-$ is the past end of $\scri^-$, while spacelike infinity is the future end of $\scri^-$.

The boundary of $\check{\mathcalboondox{M}}$ is $\scri$. Although it is a null hypersurface (the induced metric has signature $(0,-,-)$), one may integrate over $\scri^+$ with respect to the measure $\widecheck{\dvol}_{\scri^+} = \partial_u \intprod \widecheck{\dvol} $, where $\widecheck{\dvol}$ is the volume form of the rescaled spacetime, i.e. $\widecheck{\dvol}_{\scri^+} = \d u \wedge \dvol_{\mathbb{S}^2}$. Since $\partial_u$ is tangent to $\scri^+$ and $\check{\eta}(\partial_u, \partial_R) = - 1$, the vector field $\partial_R$ is transverse to $\scri^+$. The inverse metric to \eqref{rescaledMinkowski} is\footnote{Here $\odot$ denotes the symmetric tensor product $v \odot w = \frac{1}{2}(v \otimes w + w \otimes v)$.}
\[ \check{\eta}^{-1} = - 2 \, \partial_u \odot \partial_R - R^2 \partial_R \otimes \partial_R - g_{\mathbb{S}^2}^{-1}, \]
from which one sees in particular that $\partial_R$ is null for $R<1$. In fact, we may write on the whole of $\check{\mathcalboondox{M}}$
\begin{eqnarray*}
\check{\eta} &=& \Omega^2 \d u^2 -2 \frac{\Omega^2}{R^2} \d u \, \d R - \frac{\Omega^2}{R^2} g_{\mathbb{S}^2} \, ,\\
\check{\eta}^{-1} &=& - 2 \frac{R^2}{\Omega^2} \partial_u \odot \partial_R - \frac{R^4}{\Omega^2} \partial_R \otimes \partial_R - \frac{R^2}{\Omega^2} g_{\mathbb{S}^2}^{-1} \, .
\end{eqnarray*}
The vector field $\partial_R$ in coordinates $(u,R,\theta,\phi)$ is therefore null where it is defined, i.e. on the whole of $\check{\mathcalboondox{M}} \setminus \scri^-$. Similarly, in the $(v,R,\theta,\phi)$ coordinates, $\partial_R$ is null on $\check{\mathcalboondox{M}} \setminus \scri^+$. Note also that when working with $(u,R,\theta,\phi)$, $\partial_R$ is past-oriented, but it is future-oriented in the coordinates $(v,R,\theta,\phi)$.

For the hatted conformal scale, we consider the conformally rescaled unphysical metric $\hat{\eta}_{ab} \defeq \Omega^2 \eta_{ab}$, where $\Omega$ is a smooth positive radial function on $\mathcalboondox{M}$ chosen in the future of $\Sigma$ as shown in Figure \ref{fig:conformalfactorMinkowski}.
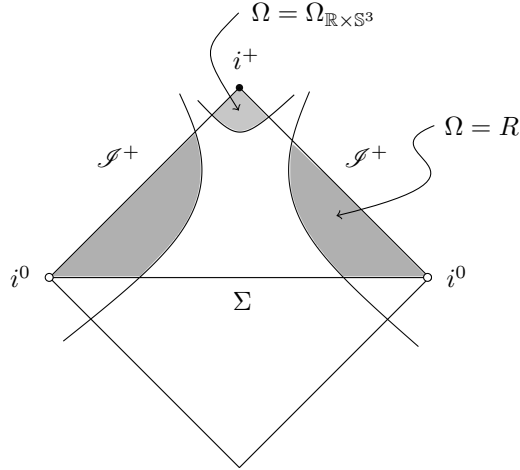
\begin{figure}[H]
\centering
	\begin{tikzpicture}
	\centering
	\node[inner sep=0pt] (conformalfactorMinkowski) at (3.4,0)
    	{\includegraphics[width=.32\textwidth]{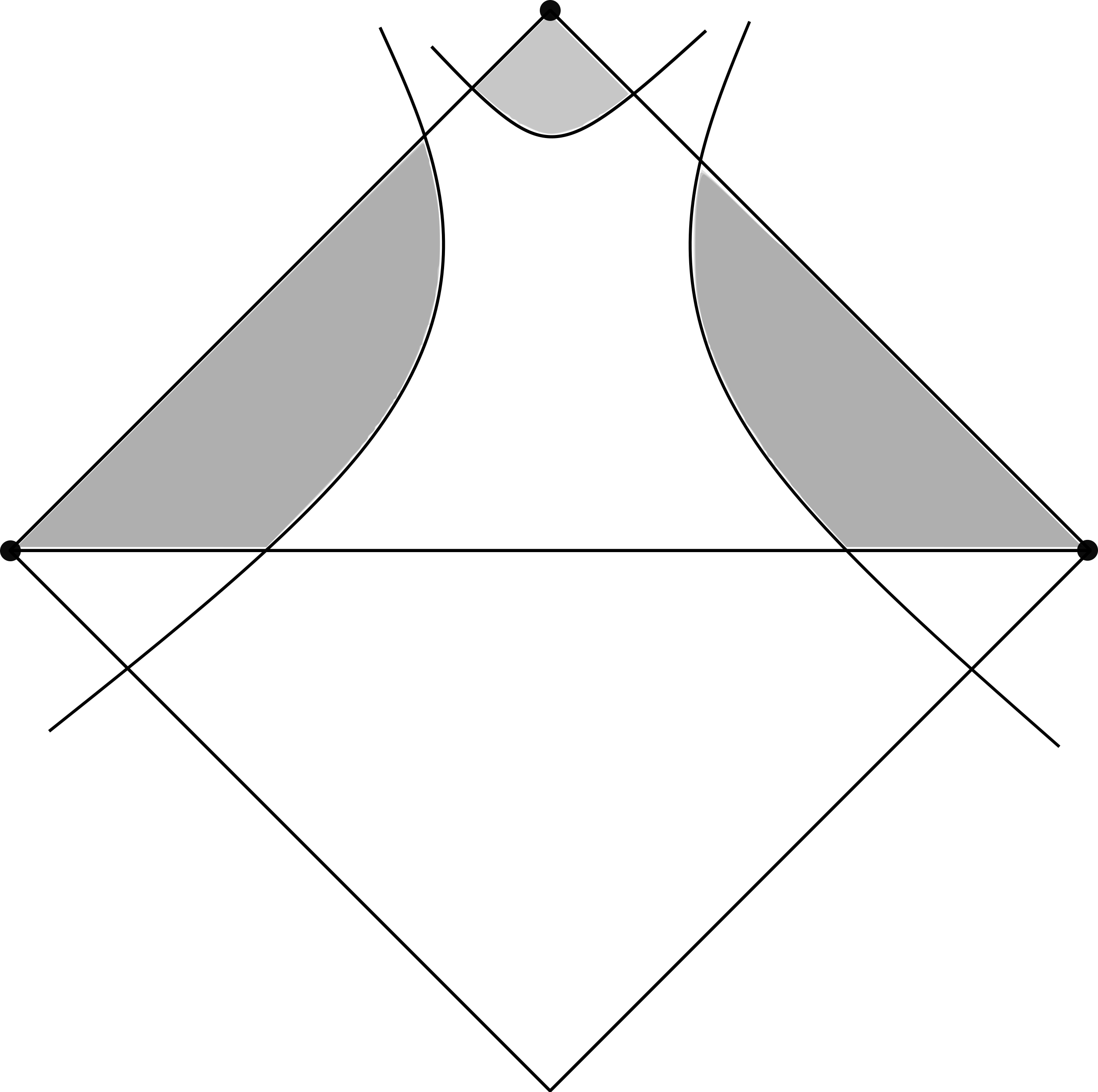}};

	\node[label={[shift={(3.5,2.5)}]$i^+$}] {};
	\node[label={[shift={(1.8,1.2)}]$\scri^+$}] {};
	\node[label={[shift={(5.1,1.2)}]$\scri^+$}] {};
	\node[label={[shift={(3.45,-0.7)}]$\Sigma$}] {};
	\node[label={[shift={(0.5,-0.4)}]$i^0$}] {};
	\node[label={[shift={(6.3,-0.4)}]$i^0$}] {};
	\node[label={[shift={(6.63,1.63)}]$ \Omega = R $}] {};	
	\node[label={[shift={(4.4,3.05)}]$\Omega = \Omega_{\mathbb{R} \times \mathbb{S}^3}$}] {};
	
	\draw[->] (6,2) .. controls (5,1.6) and (7.2,1.4) .. (4.7,0.8);
	
	\draw[->] (3.4,3.5) .. controls (2.9,3) .. (3.4,2.2);

    \draw[fill=white] (0.88,-0.02) circle (1.5pt);
    \draw[fill=white] (5.91,-0.02) circle (1.5pt);
 
	\end{tikzpicture}
	\captionsetup{width=0.85\textwidth}
	\caption{We choose a conformal factor $\Omega$ which is equal to $R$ near $\scri^+$ away from $i^+$, and smoothly brings $i^+$ to a finite distance.} \label{fig:conformalfactorMinkowski}	
\end{figure}
\noindent We assume in addition that $\Omega$ is time-symmetric, which implies in particular that $\partial_t \Omega \vert_{\Sigma}=0$. This conformal factor is such that $\Omega = R$ near $\scri^+$ and away from $i^+$ (i.e. in the region where $u \leq u_0$, and $r \gg 1$ for $u_0 \gg 1$ given), and
\[ \Omega = \frac{2R}{\sqrt{1+u^2}\sqrt{R^2 + (2+uR)^2}} \eqdef \Omega_{\mathbb{R}\times \mathbb{S}^3} \]
in a neighbourhood of $i^+$. In the white region near $\scri^+$ in \Cref{fig:conformalfactorMinkowski} the conformal factor $\Omega$ interpolates smoothly between $R$ and $\Omega_{\mathbb{R} \times \mathbb{S}^3}$ in such a manner that the function $\Omega /R$ is smooth near $\scri^+$ away from $i^+$, and does not vanish at null infinity. The function $\Omega_{\mathbb{R} \times \mathbb{S}^3}$ is precisely the conformal factor which embeds Minkowski space into the Einstein cylinder $\mathbb{R} \times \mathbb{S}^3$, and so the smoothness of the rescaled metric at $i^+$ is automatic.
The compactified spacetime is then defined as
\[ \hat{\mathcalboondox{M}} \defeq \mathcalboondox{M} \cup \scri^-\cup \scri^+ \cup i^- \cup i^+ \, .\]
The boundary of $\hat{\mathcalboondox{M}}$ is $\scri^-\cup \scri^+ \cup i^- \cup i^+ $, $i^\pm$ are finite regular points of $\hat{\eta}_{ab}$ and $\scri^\pm$ are semi-infinite cylinders that focus to $i^\pm$ in the future (resp. past). A natural measure on $\scri^+$ is now given by $\widehat{\dvol}_{\scri^+} = \partial_u \intprod \widehat{\mathrm{dv}}$, where $\widehat{\dvol}$ is the volume form of the rescaled spacetime.

Where $\Omega = \Omega_{\mathbb{R} \times \mathbb{S}^3}$, the rescaled metric $\hat{\eta}_{ab}$ can in fact be written as
\[ \hat{\eta}= \d \tau^2 - \d \zeta^2 - (\sin^2 \zeta) g_{\mathbb{S}^2} = \d \tau^2 - g_{\mathbb{S}^3}, \]
where $\tau = \arctan(u+2/R) + \arctan(u)$ and $\zeta = \arctan(u+2/R) - \arctan(u)$.
\begin{remark} \label{rmk:i+}
It will be useful to denote by
\[ \chi(u, R) \defeq \Omega R^{-1}, \]
where $\Omega$ is the conformal factor corresponding to the hatted conformal scale. The resulting function $\chi$ is then smooth near and on $\scri^+$, and equal to unity for $u \leq u_0$, $r\gg 1$. Near $i^+$, $\chi \approx (1+u^2)^{-1/2}$, and therefore \emph{on} $\scri^+$ depends only on $u$. The decay of $\chi$ as $u \to +\infty$ is responsible for the shrinking of the $2$-spheres at $i^+$; the rescaled metric in the hatted scale is given by (for $R<1$)
\[ \hat{\eta} = \chi^2 R^2 \d u^2 - 2 \chi^2 \, \d u \, \d R - \chi^2 g_{\mathbb{S}^2}. \]
\end{remark}

The NP tetrad \eqref{firsttetrad} in the coordinates $(u,r,\theta,\phi)$ becomes
\begin{align} 
\label{Minkowskitetradl}	 & n^a = \partial_u - \frac{1}{2} \partial_r, && n_a = \frac{1}{2} \, \d u + \d r, \\
\label{Minkowskitetradn}	 & l^a = \partial_r, && l_a = \d u, \\
\label{Minkowskitetradm}	 & m^a = \frac{1}{\sqrt{2}r} \left( \partial_\theta + \frac{i}{\sin \theta} \partial_\phi \right), && m_a = - \frac{r}{\sqrt{2}} ( \d \theta + i \sin \theta \, \d \phi ), \\
\label{Minkowskitetradmbar} & \bar{m}^a = \frac{1}{\sqrt{2}r} \left( \partial_\theta - \frac{i}{\sin \theta} \partial_\phi \right), && \bar{m}_a = -\frac{r}{\sqrt{2}}( \d \theta - i \sin \theta \, \d \phi ).
\end{align}
To obtain a related NP tetrad $(\hat{l}^a, \hat{m}^a, \bar{\hat{m}}^a, \hat{n}^a)$  on $\hat{\mathcalboondox{M}}$, we employ the conformal scaling
\begin{align}
	\label{tetradrescalingl} & \hat{n}^a = n^a, && \hat{n}_a = \Omega^{2} n_a, \\
	\label{tetradrescalingn} & \hat{l}^a = \Omega^{-2} l^a, && \hat{l}_a = l_a, \\
	\label{tetradrescalingm} & \hat{m}^a = \Omega^{-1} m^a, \phantom{\hspace{58pt}} && \hat{m}_a = \Omega m_a, \phantom{\hspace{64pt}} \\
	\label{tetradrescalingmbar} & \bar{\hat{m}}^a = \Omega^{-1} \bar{m}^a, && \bar{\hat{m}}_a = \Omega \bar{m}_a.
\end{align}
Explicitly, we obtain a tetrad on $\hat{\mathcalboondox{M}}$ which near $\scri^+$ takes the form
\begin{align}
\label{rescaledMinkowskitetradl} & \hat{n}^a = \partial_u + \frac{1}{2} R^2 \partial_R, && \hat{n}_a = \frac{1}{2} \chi^2 R^2 \d u - \chi^2 \d R, \\
\label{rescaledMinkowskitetradn} & \hat{l}^a = - \chi^{-2} \partial_R, && \hat{l}_a = \d u, \\
\label{rescaledMinkowskitetradm} & \hat{m}^a = \frac{1}{\sqrt{2} \chi} \left( \partial_\theta + \frac{i}{\sin \theta} \partial_\phi \right), && \hat{m}_a = - \frac{\chi}{\sqrt{2}} ( \d \theta + i \sin \theta \, \d \phi), \\
\label{rescaledMinkowskitetradmbar} & \bar{\hat{m}}^a = \frac{1}{\sqrt{2} \chi} \left( \partial_\theta - \frac{i}{\sin \theta} \partial_\phi \right), && \bar{\hat{m}}_a = - \frac{\chi}{\sqrt{2}} ( \d \theta - i \sin \theta \, \d \phi ).
\end{align}
The explicit tetrad on $\check{\mathcalboondox{M}}$---in the checked scale---is given, for $R<1$, by the expressions \Cref{rescaledMinkowskitetradl,rescaledMinkowskitetradm,rescaledMinkowskitetradmbar,rescaledMinkowskitetradn}, with $\chi \equiv 1$. In this setting the normal $T^a = \partial_t = \partial_u$ to the surfaces $\Sigma_t$ of constant $t$ reads
\[ T^a = n^a + \frac{1}{2} l^a = \hat{n}^a + \frac{\Omega^2}{2} \hat{l}^a = \check{n}^a + \frac{R^2}{2} \check{l}^a. \]
For reference, we also note that the full sets of physical and rescaled (in both the hatted and checked scales) spin coefficients on $\mathcalboondox{M}$ are given by (for $R<1$ in the case of the checked scale)
\begin{align} \renewcommand{\arraystretch}{1.3}
\label{Minkowski_physical_spin_coefficients}
&\left[ \begin{array}{ccc}
    \epsilon & \kappa & \pi  \\
    \alpha & \rho & \lambda \\
    \beta & \mu & \sigma \\
    \gamma & \nu & \tau 
\end{array} \right] 
= 
\left[ \begin{array}{ccc}
    0 & 0 & 0  \\
    \frac{1}{2 \sqrt{2} r} \cot \theta & -\frac{1}{r} & 0 \\
     \frac{-1}{2 \sqrt{2} r} \cot \theta & -\frac{1}{2r} & 0 \\
    0 & 0 & 0 
\end{array} \right], 
\\
\label{Minkowski_rescaled_spin_coefficients}
&\left[ \begin{array}{ccc}
    \hat{\epsilon} & \hat{\kappa} & \hat{\pi}  \\
    \hat{\alpha} & \hat{\rho} & \hat{\lambda} \\
    \hat{\beta} & \hat{\mu} & \hat{\sigma} \\
    \hat{\gamma} & \hat{\nu} & \hat{\tau} 
\end{array} \right] 
=
\left[ \begin{array}{ccc}
    0 & 0 & 0  \\
    \Omega^{-1} \alpha & -\Omega^{-2}(R + D \log \Omega) & 0 \\
    \Omega^{-1} \beta & -\frac{R}{2} + \Delta \log \Omega & 0 \\
    -\Delta \log \Omega & 0 & 0 
\end{array} \right],
\\
\label{Minkowski_rescaled_spin_coefficients_checked}
&\left[ \begin{array}{ccc}
    \check{\epsilon} & \check{\kappa} & \check{\pi}  \\
    \check{\alpha} & \check{\rho} & \check{\lambda} \\
    \check{\beta} & \check{\mu} & \check{\sigma} \\
    \check{\gamma} & \check{\nu} & \check{\tau} 
\end{array} \right] 
=
\left[ \begin{array}{ccc}
    0 & 0 & 0  \\
    R^{-1} \alpha & 0 & 0 \\
    R^{-1} \beta & 0 & 0 \\
    \frac{R}{2} & 0 & 0 
\end{array} \right].
\end{align}

\subsection{A Priori Energy Estimates} \label{sec:aprioriestimatesflatMaxwell}

The volume form on the rescaled spacetime $\hat{\mathcalboondox{M}}$ is given by
\[ \widehat{\dvol} = \hat{n}^\flat \wedge \hat{l}^\flat \wedge ( i \hat{m}^\flat \wedge \bar{\hat{m}}^\flat), \]
and, in the hatted scale, is explicitly given for $R<1$ by $\widehat{\dvol} = \chi^4 \, \d u \wedge \d R \wedge \dvol_{\mathbb{S}^2}$. Using $K^a = T^a$ as a multiplier vector field, we compute the energy density $3$-form
\begin{align*} K^a \hat{\mathbf{T}}_a^{\phantom{a}b} \partial_b \intprod \widehat{\dvol} & = \left( \hat{n}^a + \frac{1}{2} \Omega^2 \hat{l}^a \right) \hat{\mathbf{T}}_{ac} ( \hat{n}^b \hat{l}^c + \hat{n}^c \hat{l}^b - \hat{m}^c \bar{\hat{m}}^b - \bar{\hat{m}}^c \hat{m}^b ) \partial_b \intprod \widehat{\dvol} \\
& = - \left( 2 | \hat{F}_1 |^2 + \Omega^2 | \hat{F}_0 |^2 \right) \hat{n}^\flat \wedge(i \hat{m}^\flat \wedge \bar{\hat{m}}^\flat )  + \left( 2 |\hat{F}_2|^2 + \Omega^2 | \hat{F}_1 |^2 \right) \hat{l}^\flat \wedge (i \hat{m}^\flat \wedge \bar{\hat{m}}^\flat ) \\
& + \dots,
\end{align*}
where the ellipsis represents contractions of $\widehat{\dvol}$ with either $\hat{m}$ or $\bar{\hat{m}}$. One immediately reads off the energy on $\scri^+$,
\begin{equation} \label{Maxwellfieldscrienergyflat}	\mathcal{E}_{\scri^+}[\hat{F}] \simeq \int_{\scri^+} | \hat{F}_2 |^2 \,\widehat{\dvol}_{\scri^+} = \| \hat{F}_2 \|^2_{L^2(\scri^+)},
\end{equation}
where $\widehat{\dvol}_{\scri^+} = \hat{l}^\flat \wedge (i \hat{m}^\flat \wedge \bar{\hat{m}}^\flat )  = \chi^2 \, \d u \wedge \dvol_{\mathbb{S}^2}$. Similarly, on the initial surface $\Sigma = \{ t = 0 \}$ whose future-pointing unit normal with respect to $\hat{\eta}$ is
\[ \hat{T}^a \Big|_\Sigma = \Big( \Omega^{-1} \hat{n}^a + \frac{1}{2} \Omega \hat{l}^a \Big) \Big|_\Sigma, \]
the energy density $3$-form is
\begin{align*} K^a \hat{\mathbf{T}}_a^{\phantom{a}b} \partial_b \intprod \widehat{\dvol} \big|_\Sigma & = K^a \hat{T}^b \hat{\mathbf{T}}_{ab} ( \hat{T} \intprod \widehat{\dvol} ) \big|_{\Sigma} = \left( \frac{2}{\Omega} | \hat{F}_2 |^2 + 2 \Omega |\hat{F}_1 |^2 + \frac{\Omega^3}{2} | \hat{F}_0 |^2 \right) \widehat{\dvol}_{\Sigma},
\end{align*}
so that 
\begin{equation} \label{Maxwellfieldinitialenergyflat} \mathcal{E}_{\Sigma}[\hat{F}] \simeq \int_{\Sigma} \left( \Omega^{-1} | \hat{F}_2 |^2 +  \Omega |\hat{F}_1 |^2 + \Omega^3 | \hat{F}_0 |^2 \right) \widehat{\dvol}_{\Sigma}, 	
\end{equation}
where $\widehat{\dvol}_{\Sigma} = ( \hat{T} \intprod \widehat{\dvol} ) \big|_{\Sigma}$. The following theorem is then largely a triviality.

\begin{theorem} \label{thm:aprioriestimatesflatMaxwell} For smooth compactly supported Maxwell initial data on $\Sigma$ there exists a unique smooth rescaled solution $\hat{F}_{ab}$ which extends smoothly to $\scri^+$ and satisfies the energy estimate
\begin{equation} \label{Maxwellflatenergyestimate} \mathcal{E}_{\scri^+} = \mathcal{E}_{\Sigma}. \end{equation}
\end{theorem}

\begin{proof} The rescaled Maxwell field $\hat{F}_{ab}$ satisfies\footnote{By differentiating $\hat{\nabla}_{[a} \hat{F}_{bc]} = 0$ and using $\hat{\nabla}^a \hat{F}_{ab} = 0$, or, more geometrically, noting that $\d \hat{F} = 0$ and $\delta \hat{F} = 0$ imply $(\d \delta + \delta \d ) \hat{F} = 0$, where $\delta$ is the codifferential on $\hat{\mathcalboondox{M}}$, and $\d \delta + \delta \d $ is $\widehat{\Box}$ up to lower order terms.} the linear wave equation $\widehat{\Box} \hat{F}_{ab} + \hat{L}_0[\hat{F}]_{ab} = 0$ on $\hat{\mathcalboondox{M}}$, where $\hat{L}_0$ is a linear zeroth order differential operator involving the curvature of $\hat{\mathcalboondox{M}}$. It is classical \cite{Leray1953} that therefore $\hat{F}_{ab}$ propagates at finite speed, and for smooth data $\mathcalboondox{f} = (\hat{\mathbf{E}}_0, \hat{\mathbf{B}}_0, \hat{\mathbf{E}}_1, \hat{\mathbf{B}}_1) = (\hat{\mathbf{E}}, \hat{\mathbf{B}}, \partial_\tau \hat{\mathbf{E}}, \partial_\tau \hat{\mathbf{B}})|_{\Sigma} \in \mathcal{C}_c^\infty(\Sigma)^4$ there exists a unique smooth solution $\hat{F}_{ab}$ on $\hat{\mathcalboondox{M}}$, for example by using the foliation $\{ \hat{\Sigma}_\tau \}_\tau$ to solve the Cauchy problem on $\hat{\mathcalboondox{M}}$. We therefore have a unique smooth solution $\hat{F}_{ab}$ which extends smoothly to $\scri^+$, and has support as depicted in \Cref{fig:finitespeedofpropagation}.

\begin{figure}[H]
\centering
	\begin{tikzpicture}
	\centering
	\node[inner sep=0pt] (estimates) at (3.4,0)
    	{\includegraphics[width=.32\textwidth]{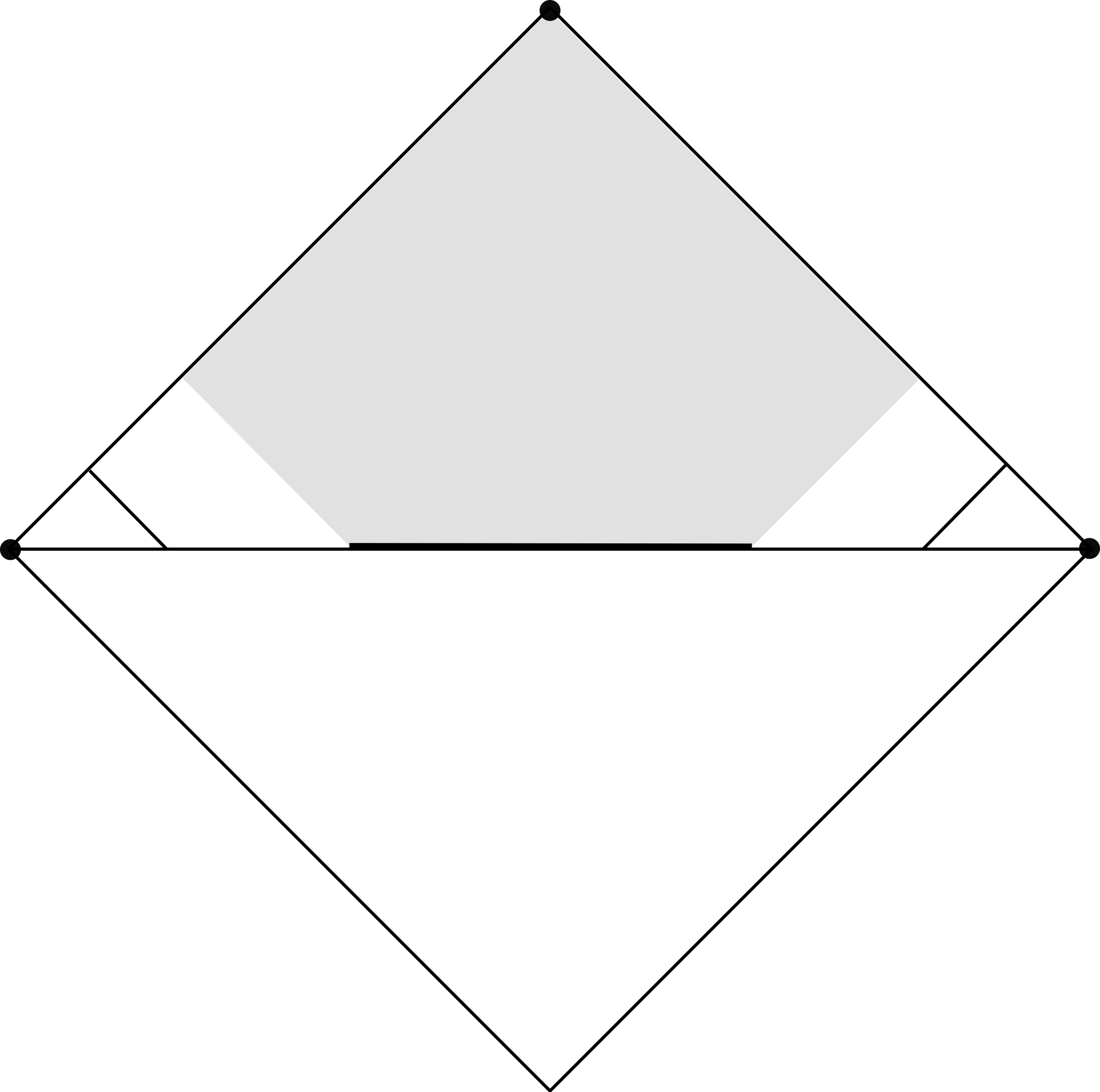}};

	\node[label={[shift={(3.5,2.5)}]$i^+$}] {};
	\node[label={[shift={(1.7,1.2)}]$\scri^+$}] {};
	\node[label={[shift={(5.1,1.2)}]$\scri^+$}] {};
	\node[label={[shift={(3.45,-0.7)}]$\operatorname{supp} \mathcalboondox{f}$}] {};
	\node[label={[shift={(0.5,-0.4)}]$i^0$}] {};
	\node[label={[shift={(6.3,-0.4)}]$i^0$}] {};
	\node[label={[shift={(6.8,2.1)}]$\mathscr{J}$}] {};
	
	\draw[->] (6.5,2.5) .. controls (5.5,1.6) and (4.5, 0.7)  .. (5.3,0.2);
	\draw[->] (6.5,2.5) .. controls (4.5,2) and (2.5, 0.7)  .. (1.5,0.2);

    \draw[fill=white] (0.88,-0.01) circle (1.5pt);
    \draw[fill=white] (5.91,-0.01) circle (1.5pt);
 
	\end{tikzpicture}
	\captionsetup{width=0.85\textwidth}
	\caption{For smooth compactly supported data the solution is smooth and compactly supported in $\hat{\mathcalboondox{M}}$.} \label{fig:finitespeedofpropagation}	
\end{figure}
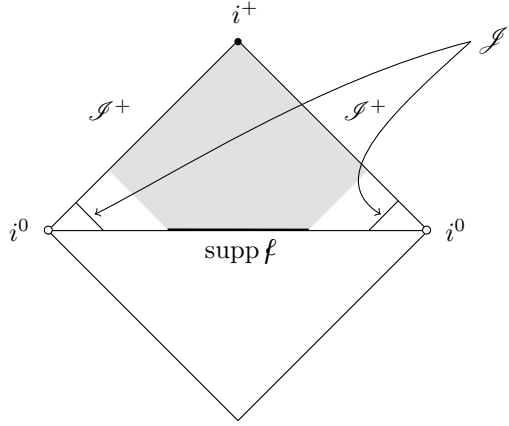

\noindent To prove the energy estimate, choose a compact subset $K \Subset \Sigma $ of the initial surface such that $\operatorname{supp} \mathcalboondox{f} \subset K$ and consider the null hypersurface $\mathscr{J} = \partial J^+(K)$. Then $J^+(K) \cap \hat{\mathcalboondox{M}}$ is a compact manifold with boundary $K \cup \mathscr{J} \cup (\scri^+ \cap J^+(K))$. We integrate the divergence
\[ \hat{\nabla}^b ( K^a \hat{\mathbf{T}}_{ab} ) = \hat{\nabla}^{(b} K^{a)} \hat{\mathbf{T}}_{ab} + K^a \hat{\nabla}^b \hat{\mathbf{T}}_{ab} \]
over $J^+(K) \cap \hat{\mathcalboondox{M}}$ and apply the divergence theorem. Since the stress-energy tensor $\hat{\mathbf{T}}_{ab}$ is conserved and $K^a = \partial_t$ is conformally\footnote{This follows from the fact that $K^a = \partial_t$ is exactly Killing on the physical spacetime $\mathcalboondox{M}$, and the identity $\pounds_K \hat{\eta}_{ab} = \pounds_K (\Omega^2 \eta_{ab} ) = ( \Omega^{-2} \pounds_K \Omega^2 ) \hat{\eta}_{ab} = 2 (\Omega^{-1} \partial_u \Omega ) \hat{\eta}_{ab}$, where $\pounds_K$ denotes the Lie derivative along $K$. In fact, $K^a$ is also exactly Killing with respect to $\hat{\eta}$ since the conformal factor $\Omega$ is purely a function of $r$ and therefore satisfies $\partial_u \Omega =0$.} Killing on the rescaled spacetime $(\hat{\mathcalboondox{M}}, \hat{\eta})$, one sees that the current $\hat{J}_b = K^a \hat{\mathbf{T}}_{ab}$ is exactly conserved, $\hat{\nabla}^b \hat{J}_b  = 0$. We therefore have, using the earlier expression for the energy density $3$-form,
\[ 0 = - \mathcal{E}_\Sigma[\hat{F}] + \mathcal{E}_{\scri^+}[\hat{F}] + \int_{\mathscr{J}} \hat{J}^b \partial_b \intprod \widehat{\dvol}. \]
Since the hypersurface $\mathscr{J}$ is outside of the support of $\hat{F}_{ab}$, the last integral vanishes identically, and we conclude the result.
\end{proof}

\subsubsection{Conformal invariance of energies}

On any hypersurface $\mathcal{H}$ of $\hat{\mathcalboondox{M}} \supset \mathcalboondox{M}$ the energies induced by the rescaled stress-energy tensor $\hat{\mathbf{T}}_{ab}$ and the physical stress-energy tensor $\mathbf{T}_{ab}$ are equal as a consequence of the conformal covariance $\hat{\mathbf{T}}_{ab} = \Omega^{-2} \mathbf{T}_{ab}$ of $\mathbf{T}_{ab}$. Indeed, for any multiplier vector field $K^a$
\begin{align} \begin{split} \label{Maxwellconformalinvarianceenergies} \int_{\mathcal{H}} K^a \mathbf{T}_a^{\phantom{a}b} \partial_b \intprod \dvol &= \int_{\mathcal{H}} K^a \mathbf{T}_{ac} g^{bc} \partial_b \intprod \dvol \\
& = \int_{\mathcal{H}} K^a \Omega^2 \hat{\mathbf{T}}_{ac} \Omega^2 \hat{g}^{bc} \partial_b \intprod \Omega^{-4} \widehat{\dvol} = \int_{\mathcal{H}} K^a \hat{\mathbf{T}}_a^{\phantom{a}b} \partial_b \intprod \widehat{\dvol}.
\end{split}	
\end{align}
In particular, the initial energy \eqref{Maxwellfieldinitialenergyflat} is
\begin{align*} \mathcal{E}_{\Sigma}[\hat{F}] & = \int_{\Sigma} K^a \hat{\mathbf{T}}_{ab} \hat{T}^b ( \hat{T} \intprod \widehat{\dvol} ) = \int_{\Sigma} K^a \mathbf{T}_{ab} T^b (T \intprod \dvol) = \frac{1}{2} \int_{\Sigma} ( |\mathbf{E}|^2 + |\mathbf{B}|^2) \dvol_{\Sigma} \eqdef \mathcal{E}_{\Sigma}[F],
\end{align*}
where $\dvol_{\Sigma} = (T \intprod \dvol)\big|_{\Sigma}$. We also note that $\widehat{\dvol}_{\Sigma} = \Omega^3 \dvol_\Sigma$ and the components of the Maxwell field scale according to
\begin{align} \label{Maxwellfieldcomponentsconformalscaling}
(F_0, F_1, F_2) = (\Omega^3 \hat{F}_0, \Omega^2 \hat{F}_1, \Omega \hat{F}_2), \quad \text{or} \quad F_i = \Omega^{3-i} \hat{F}_i, \quad i \in \{0, 1, 2\},
\end{align}
so the expression \eqref{Maxwellfieldinitialenergyflat} can be rewritten as
\[ \mathcal{E}_{\Sigma}[\hat{F}] \simeq \int_{\Sigma} \left( |F_0 |^2 + | F_1 |^2 + |F_2|^2 \right) \dvol_{\Sigma} = \| F_0 \|^2_{L^2(\Sigma)} + \| F_1 \|^2_{L^2(\Sigma)} + \| F_2 \|^2_{L^2(\Sigma)} \simeq \mathcal{E}_{\Sigma}[F]. \]

\subsection{Field Equations and Gauge Fixing} \label{sec:flatMaxwellequationsofmotion}

The equations \eqref{Maxwellsequationspotential} on $\mathcalboondox{M}$ read
\begin{equation} \label{MaxwellsequationspotentialMinkowski} \Box A_b - \nabla_b ( \nabla_a A^a ) = 0,
\end{equation}
and, by conformal invariance, are equivalent to
\begin{equation} \label{Maxwellsequationspotentialhatted} \widehat{\Box} \hat{A}_b - \hat{\nabla}_b ( \hat{\nabla}_a \hat{A}^a ) + \hat{\mathrm{R}}_{ab} \hat{A}^a = 0
\end{equation}
 on $\hat{\mathcalboondox{M}}$. The energies defined by \eqref{Maxwellfieldscrienergyflat} and \eqref{Maxwellfieldinitialenergyflat}, when written out in terms of the potential $A_a$, contain antisymmetrized derivatives of $A_a$ and do not define norms on the potential without a choice of gauge. To construct the trace and scattering operators as maps between Hilbert spaces, one thus aims to fix the gauge in such a way that the natural energies on the initial surface $\Sigma$ and $\scri^+$ become norms on the space of Maxwell potentials. To this end, we will impose a gauge on the physical field $A_a$, and show that it leads to an admissible gauge fixing condition on $\hat{A}_a$ throughout $\hat{\mathcalboondox{M}}$, all the way up to and on $\scri^+$. To reduce the natural energy in the physical spacetime to a norm in $A_a$ on $\Sigma \simeq \mathbb{R}^3$, the obvious choice of gauge is the Coulomb gauge $\boldsymbol{\nabla} \cdot \mathbf{A} = 0$ throughout $\mathcalboondox{M}$, since then for a smooth compactly supported potential on $\Sigma$
 \begin{align*} \mathcal{E}_{\Sigma}[F] & = \frac{1}{2} \int_{\Sigma} \left( | \mathbf{E} |^2 + | \mathbf{B} |^2 \right) \dvol_\Sigma \\ 
 & = \frac{1}{2} \int_{\Sigma} \left( | \dot{\mathbf{A}} |^2 - 2 \dot{\mathbf{A}} \cdot \boldsymbol{\nabla} \mathfrak{a} + | \boldsymbol{\nabla} \mathfrak{a} |^2 + | \boldsymbol{\nabla} \mathbf{A} |^2 - \boldsymbol{\nabla}_j \mathbf{A}_i \boldsymbol{\nabla}^i \mathbf{A}^j \right) \dvol_{\Sigma} \\
 & = \frac{1}{2} \int_{\Sigma} \left( | \dot{\mathbf{A}} |^2 + | \boldsymbol{\nabla} \mathfrak{a} |^2 + | \boldsymbol{\nabla} \mathbf{A} |^2 \right) \dvol_{\Sigma},
 \end{align*}
where $\boldsymbol{\nabla}$ is the Levi--Civita connection on $\Sigma$, $\mathbf{A}$ denotes the projection of $A_a$ onto $\Sigma$, $\dot{\mathbf{A}} = \partial_t \mathbf{A}$, and in the last line we integrated by parts and used the Coulomb gauge conditions $\boldsymbol{\nabla} \cdot \mathbf{A} = 0 = \boldsymbol{\nabla} \cdot \dot{\mathbf{A}}$ on $\Sigma$. Now if one contracts \eqref{MaxwellsequationspotentialMinkowski} with $T^a = \partial_t$, one ends up with the elliptic equation
\begin{equation} \label{MaxwellCoulombellipticMinkowski} \boldsymbol{\Delta} \mathfrak{a} = 0 \quad \text{on } \Sigma_t	
\end{equation}
for each $t \in \mathbb{R}$. We therefore have the following result.

\begin{proposition} \label{prop:simultaneousgaugesMinkowski} On Minkowski space $(\mathcalboondox{M}=\mathbb{R}^4, \eta)$ one may impose the gauges
\begin{equation} \label{simultaneousgaugesMinkowski} \boldsymbol{\nabla} \cdot \mathbf{A} = 0, \quad \mathfrak{a} = 0, \quad \text{and} \quad \nabla_a A^a = 0 \end{equation}
simultaneously. We call the gauge \eqref{simultaneousgaugesMinkowski} the \emph{temporal-Coulomb gauge}.
\end{proposition}

\begin{proof} Let $A_a = (\mathfrak{a}, \mathbf{A})$ be any smooth solution to Maxwell's equations on $\mathcalboondox{M}$. We impose the Coulomb gauge $\boldsymbol{\nabla} \cdot \mathbf{A} = 0$, which has the residual gauge freedom $\chi_{\mathrm{res.}}$, where $\boldsymbol{\Delta} \chi_{\mathrm{res.}} = 0$ on $\Sigma_t$ for all $t$. The solutions to $\boldsymbol{\Delta} \chi_{\mathrm{res.}} = 0$ are constants on $\Sigma_t \simeq \mathbb{R}^3$, so a residual gauge transformation effects
\[ \mathfrak{a} \rightsquigarrow \mathfrak{a} + \partial_t \chi_{\mathrm{res.}}, \]
where $\chi_{\mathrm{res.}}$ is a function only of $t$. From \eqref{MaxwellCoulombellipticMinkowski}, $\mathfrak{a}$ is also a function only of $t$, so we may choose $\chi_{\mathrm{res.}}$ so that the residually gauge-transformed component $\mathfrak{a}$ is identically zero (by choosing $\chi_{\mathrm{res.}}$ to be the negative of the antiderivative of $\mathfrak{a}$). Then $\nabla_a A^a = \partial_t \mathfrak{a}  - \boldsymbol{\nabla} \cdot \mathbf{A} = 0$ follows automatically.
\end{proof}

\noindent To study $\hat{A}_a$ on the rescaled spacetime $\hat{\mathcalboondox{M}}$ (and in particular to ensure that we can solve for $\hat{A}_a$ up to $\scri$), we must convert the gauge condition \eqref{simultaneousgaugesMinkowski} into a gauge condition on the rescaled Maxwell potential $\hat{A}_a$ and solve the system \eqref{Maxwellsequationspotentialhatted}. Under a conformal transformation $g_{ab} \rightsquigarrow \hat{g}_{ab} = \Omega^2 g_{ab}$ the spacetime divergence $\nabla_a A^a$ transforms as
\[ \nabla_a A^a = \Omega^2 (\hat{\nabla}_a \hat{A}^a - 2 \Upsilon_a \hat{A}^a ), \]
so Lorenz gauge in the physical spacetime $\mathcalboondox{M}$ is equivalent to the gauge condition
\begin{equation} \label{physicalLorenzgaugeMaxwell} \hat{\nabla}_a \hat{A}^a = 2 \Upsilon_a \hat{A}^a = 2 \Omega^{-1} (\hat{\nabla}^a \Omega) \hat{A}_a
\end{equation}
on $\hat{\mathcalboondox{M}}$. The equation \eqref{Maxwellsequationspotentialhatted} then reads
\begin{equation} \label{rescaledMaxwellequationspotentialphysicalLorenzgauge} \widehat{\Box} \hat{A}_b - 2 \hat{\nabla}_b (\Omega^{-1} (\hat{\nabla}^a \Omega) \hat{A}_a ) + \hat{\mathrm{R}}_{ab} \hat{A}^a = 0. \end{equation}
Here the appearance of the $\Omega^{-1}$ factor in \eqref{rescaledMaxwellequationspotentialphysicalLorenzgauge} is problematic; as it stands, solutions to \eqref{rescaledMaxwellequationspotentialphysicalLorenzgauge} may develop singularities on $\scri = \{ \Omega = 0 \}$. The extra temporal gauge condition $\mathfrak{a} = 0$ on $\mathcalboondox{M}$ ensures that this cannot happen: recalling that $T^a = n^a + \frac{1}{2} l^a = \hat{n}^a + \frac{1}{2} \Omega^2 \hat{l}^a$, the temporal gauge transported to $\hat{\mathcalboondox{M}}$ reads
\begin{equation} \label{temporalgaugerescaled} 0 = \hat{A}_1 + \frac{1}{2} \Omega^2 \hat{A}_0. \end{equation}
Since $\Omega$ is radial, we must have $\hat{\nabla}^a \Omega = f \hat{n}^a + \tilde{g} \hat{l}^a$ for some functions $f$ and $\tilde{g}$. Since $\hat{\nabla}^a \Omega $ becomes proportional to $\hat{n}^a$ on $\scri^+$, we must also have $\tilde{g} \to 0$ as $\Omega \to 0$, i.e. $\tilde{g} = \Omega g$ for some function $g$. Therefore, using \eqref{temporalgaugerescaled},
\[ \Omega^{-1} (\hat{\nabla}^a \Omega) \hat{A}_a = \left(- \frac{1}{2} \Omega f + g \right) \hat{A}_0 \eqdef \varsigma \hat{A}_0, \]
showing that the coefficients of the equation \eqref{rescaledMaxwellequationspotentialphysicalLorenzgauge} are in fact non-singular up to $\scri^+$. In the hatted conformal scale, $\Omega = R \chi(u,R)$ for $R<1$, so one may compute $f = - \chi^{-1} (1 + R \chi^{-1} \chi_R )$, $g=\chi^{-1} \chi_u + \frac{1}{2}R (1+R \chi^{-1} \chi_R) $, and $\varsigma = R ( 1+ R \chi^{-1} \chi_R ) + \chi^{-1} \chi_u$. Equation \Cref{rescaledMaxwellequationspotentialphysicalLorenzgauge} then becomes
\begin{equation} \label{rescaledMaxwellflatnonsingular} \widehat{\Box} \hat{A}_b - 2 \hat{\nabla}_b (\varsigma \hat{A}_0 ) + \hat{\mathrm{R}}_{ab} \hat{A}^a = 0.	
\end{equation}
This will ensure that the solution to \eqref{rescaledMaxwellequationspotentialphysicalLorenzgauge} is in fact smooth throughout the partially compactified spacetime $\hat{\mathcalboondox{M}}$, including on $\scri^+$ and at $i^+$. A little care is needed to solve the Cauchy problem for equation \eqref{rescaledMaxwellequationspotentialphysicalLorenzgauge}, as the background spacetime $\hat{\mathcalboondox{M}}$ is singular at $i^0$. To resolve this, we follow the strategy of Mason and Nicolas \cite{MasonNicolas2004} (see Lemma 2.4 therein for helpful illustrations of the following procedure). For smooth compactly supported initial data $(\hat{A}_a, \hat{\nabla}_{\hat{T}}\hat{A}_a)|_\Sigma$ the putative solution will have support bounded away from $i^0$. This allows us to deform the initial surface $\Sigma$ away from the support of the initial data in such a way that the new deformed initial surface $\tilde{\Sigma}$ remains uniformly spacelike, intersects $\scri^+$ in the future of $i^0$, and the support of the solution in the future of $\Sigma$ remains to the future of the deformed surface $\tilde{\Sigma}$. We then cut off and discard the past of the deformed surface $\tilde{\Sigma}$. The part of $\hat{\mathcalboondox{M}}$ lying in the future of $\tilde{\Sigma}$ is now a completely regular compact globally hyperbolic Lorentzian manifold with boundary, and from the point of view of the solution to \eqref{rescaledMaxwellflatnonsingular} in $J^+(\Sigma)$ is indistinguishable from $\hat{\mathcalboondox{M}}$. It may be extended to a smooth globally hyperbolic Lorentzian manifold without boundary, say $(\hat{\mathcalboondox{M}}^e = \mathbb{R} \times \mathbb{S}^3, \hat{h})$, where $\hat{h}$ agrees with $\hat{\eta}$ on $J^+(\tilde{\Sigma}) \cap \hat{\mathcalboondox{M}}$. By standard theory (e.g. Leray's theory for symmetric hyperbolic systems \cite{Leray1953}), the original smooth compactly supported data $(\hat{A}_a, \hat{\nabla}_{\hat{T}}\hat{A}_a)|_\Sigma$ gives rise to a unique smooth solution $\hat{A}_a$ on $\hat{\mathcalboondox{M}}^e$, which solves \eqref{rescaledMaxwellflatnonsingular} in $J^+(\Sigma) \cap \hat{\mathcalboondox{M}}$ and whose support remains away from $i^0$. In particular, the components $\hat{A}_0$ and $\hat{A}_1$ of this solution are smooth up to $\scri^+$, and so the temporal gauge condition \eqref{temporalgaugerescaled} can be extended smoothly \emph{onto} $\scri^+$, where it becomes
\begin{equation} \label{temporalgaugescri} \hat{A}_1 \approx 0.	
\end{equation}
With the gauge condition \eqref{temporalgaugerescaled} now satisfied throughout $\hat{\mathcalboondox{M}}$, the equation \eqref{rescaledMaxwellflatnonsingular} in fact consists of three, not four independent equations, since the component $\hat{A}_1$ can be determined from $\hat{A}_0$. We are thus in a position to prove the following.

\begin{theorem} \label{thm:MaxwellpotentialflatCauchyproblem} To smooth compactly supported initial data $(A_a, \nabla_T A_a)|_{\Sigma}$ for Maxwell's equations in the temporal-Coulomb gauge $\mathfrak{a}|_\Sigma = \boldsymbol{\nabla} \cdot \mathbf{A}|_\Sigma = \boldsymbol{\nabla} \cdot \dot{\mathbf{A}}|_\Sigma = 0$ one can associate a unique smooth \emph{rescaled} solution $\hat{A}_a$ on $\hat{\mathcalboondox{M}}$. The support of $\hat{A}_a$ remains away from $i^0$, and $\hat{A}_a$ satisfies the gauge conditions \eqref{physicalLorenzgaugeMaxwell} and \eqref{temporalgaugerescaled} throughout $\hat{\mathcalboondox{M}}$. In particular, $\hat{A}_1 \approx 0$ on $\scri^+$.
\end{theorem}

\begin{proof} First, it is clear that for smooth compactly supported initial data \emph{for the field} there exists a unique smooth solution $F_{ab}$ on $\mathcalboondox{M}$, for example by \cite{Leray1953}, and that the initial gauge constraints
\[ \mathfrak{a}|_\Sigma = \boldsymbol{\nabla} \cdot \mathbf{A}|_\Sigma = \boldsymbol{\nabla} \cdot \dot{\mathbf{A}}|_\Sigma = 0 \]
are propagated throughout $\mathcalboondox{M}$. By \Cref{prop:simultaneousgaugesMinkowski}, we may impose the temporal-Coulomb gauge on this solution throughout $\mathcalboondox{M}$. Once rescaled initial data $(\hat{A}_a, \hat{\nabla}_{\hat{T}} \hat{A}_a)|_\Sigma$ is obtained from the physical initial data $(A_a, \nabla_T A_a)|_\Sigma$, the above construction goes through to extend the solution $\hat{A}_a = A_a$ to $\scri^+$, and ensure that the relevant gauge conditions are satisfied. The rescaled initial data is easily constructed from the physical initial data; one has
\[ \hat{A}_a = A_a \]
and
\[ \hat{\nabla}_{\hat{T}} \hat{A}_a = \hat{T}^b \hat{\nabla}_b \hat{A}_a = \Omega^{-1} T^b ( \nabla_b A_a - \Upsilon_b A_a - \Upsilon_a A_b + \eta_{ab} \eta^{cd} \Upsilon_c A_d ). \]
The first of these immediately gives the data for $\hat{A}_a$ in terms of the data for $A_a$, while for the time derivative, restriction to $\Sigma$ gives
\begin{equation} \label{physicaltorescaleddata} \hat{\nabla}_{\hat{T}} \hat{A}_a \big|_\Sigma = \Omega^{-1} ( \nabla_T A_a + T_a \eta^{cd} \Upsilon_c A_d) \big|_\Sigma,	
\end{equation}
since $T^b \Upsilon_b |_\Sigma = 0$ and $\mathfrak{a}|_\Sigma = 0$. Smoothness and compact support of $(A_a, \nabla_T A_a)|_\Sigma$ then imply the smoothness and compact support of $(\hat{A}_a, \hat{\nabla}_{\hat{T}} \hat{A}_a)|_\Sigma$.
\end{proof}

\subsubsection{Gauge reduction on \texorpdfstring{$\scri^+$}{scri}}

In addition to the gauge $\hat{A}_1 \approx 0$ on $\scri^+$, the triple gauge fixing condition \eqref{simultaneousgaugesMinkowski} in fact also gives rise to a kind of second-order gauge reduction on $\scri^+$, as we shall see now. Noting that $\mathfrak{a} = \hat{A}_1 + \frac{1}{2} \Omega^2 \hat{A}_0 = 0$, $\hat{\nabla}^a \Omega = f \hat{n}^a + \Omega g \hat{l}^a$, and the values of the rescaled spin coefficients \eqref{Minkowski_rescaled_spin_coefficients}, we have
\begin{align*}
- \boldsymbol{\nabla} \cdot \mathbf{A} & = \nabla_a A^a \\
& = \Omega^2 ( \hat{\nabla}_a \hat{A}^a - 2 \Upsilon_a \hat{A}^a ) \\
& = \Omega^2 \left( \hat{\thorn}' \hat{A}_0 - \hat{\eth} \bar{\hat{A}}_2 - \bar{\hat{\eth}} \hat{A}_2 + 2 \hat{\mu} \hat{A}_0 - 2 g \hat{A}_0 \right) + \mathcal{O}(\Omega^3).
\end{align*}
In our hatted conformal scale $g \approx \chi^{-1} \chi_u \approx \hat{\mu}$, so dividing by $\Omega^2$ and taking the limit $\Omega \to 0$ we find that the Coulomb gauge $\boldsymbol{\nabla} \cdot \mathbf{A} = 0$ implies the condition
\begin{equation} \label{temporalCoulombgaugescri_lightcone} 
\hat{\thorn}' \hat{A}_0 \approx 2 \operatorname{Re} \hat{\eth} \bar{\hat{A}}_2
\end{equation}
on $\scri^+$. We shall use this relation to construct a complete set of characteristic data on $\scri^+$, and hence define suitable spaces of scattering data.

\begin{remark} \label{rmk:flat_frame_scri_gauge_reduction} For $R<1$, the checked conformal scale is related to the hatted scale by $\hat{\eth} = \chi^{-1} \check{\eth}$, $\hat{A}_2 = \chi^{-1} \check{A}_2$, where the norms of the derivatives $\check{\eth}$ are now independent of $u$. In particular, $2 \operatorname{Re} \check{\eth} \bar{\check{A}}_2 = \nabla_{\mathbb{S}^2} \cdot \check{A}_{\mathbb{S}^2}$, where $\nabla_{\mathbb{S}^2}$ is the Levi--Civita connection on the round sphere and $\check{A}_{\mathbb{S}^2}$ are the (real) components of $A_a$ with respect to a frame on $\mathbb{S}^2$ that does not depend on $u$. The condition \eqref{temporalCoulombgaugescri_lightcone} then reads
\begin{equation} 
\label{temporalCoulombgaugescri}
\partial_u\check{A}_0 \approx 2 \operatorname{Re} \check{\eth} \bar{\check{A}}_2,
\end{equation}
where $\check{A}_0 = \chi^2 \hat{A}_0$.
\end{remark}

\subsection{The Scattering Construction} \label{sec:Maxwellpotentialsflatscattering}

\subsubsection{Spaces of initial and scattering data} \label{sec:spaces_of_data_flat}

Since $\hat{A}_1 \approx 0$ and the angular derivatives are tangential to $\scri^+$, we have $\eth \hat{A}_1 \approx 0$. Noting\footnote{In fact, the conditions $\hat{\nu} \approx 0$ and $\hat{\lambda} \approx 0$ are not special to the physical spacetime being Minkowski, but hold more generally and encode the fact that $\scri^+$ is shear-free and geodetic, respectively. This will be important to us when we work on curved spacetimes in \Cref{sec:curvedspacetimes}. The condition $\hat{\mu} \approx \bar{\hat{\mu}}$ partly encodes the fact that $\scri^+$ is a null hypersurface, and will also hold more generally.} that $\hat{\nu} \approx 0 \approx \hat{\lambda}$ and $\hat{\mu} \approx \bar{\hat{\mu}}$, the expansion \eqref{F2components} for $\hat{F}_2$ reduces to
\[ \hat{F}_2 \approx - (\hat{\thorn}' + \hat{\mu} ) \bar{\hat{A}}_2. \]
On $\scri^+$ the relevant spin coefficients are $- \hat{\gamma} \approx \chi^{-1} \chi_u \approx \hat{\mu}$, so in fact 
\[ \hat{F}_2 \approx - \partial_u \bar{\hat{A}}_2 - \chi^{-1} \chi_u \bar{\hat{A}}_2 = - \partial_u( \chi \bar{\hat{A}}_2 ) \chi^{-1}. \]
Therefore
\[ \mathcal{E}_{\scri^+}[\hat{A}] \simeq \int_{\scri^+} |\hat{F}_2|^2 \, \widehat{\dvol}_{\scri^+} = \int_{\scri^+} |\partial_u(\chi \hat{A}_2)|^2 \chi^{-2} \, \widehat{\dvol}_{\scri^+}. \]

\begin{definition} For the component $\hat{A}_2$ of the Maxwell potential we define the semi-norm $\| \cdot \|_{\dot{\mathcal{H}}^1(\scri^+)}$ by
\begin{equation}
\label{characteristicdatanormMinkowski}
\| \hat{A}_2 \|^2_{\dot{\mathcal{H}}^1(\scri^+)} \defeq \int_{\scri^+} |\partial_u (\chi \hat{A}_2 )|^2 \chi^{-2} \, \widehat{\dvol}_{\scri^+}.
\end{equation}
\end{definition}

\begin{remark} \label{rmk:decompactification_scri_Minkowski} One can also rewrite the energy \eqref{characteristicdatanormMinkowski} in terms of the checked conformal scale, which is perhaps more natural if $\scri^+ \simeq \mathbb{R} \times \mathbb{S}^2$ is to be thought of as an abstract manifold detached from the interior of the spacetime, with the degenerate metric $0 \cdot \d u^2 - g_{\mathbb{S}^2}$. In this scale one simply has
\begin{equation}
\label{characteristicdatannormMinkowski_flat_frame}
\| \check{A}_2 \|^2_{\dot{\mathcal{H}}^1(\scri^+)} = \int_{\scri^+} |\partial_u \check{A}_2  |^2 \, \d u \wedge \dvol_{\mathbb{S}^2}, 
\end{equation}
where $\check{A}_2$ is the conformally transformed $\hat{A}_2$ on $\scri^+$, $\check{A}_2 = \chi \hat{A}_2$. The finiteness of the energy \eqref{characteristicdatannormMinkowski_flat_frame} then puts $\check{A}_2 \in \dot{H}^1(\mathbb{R}_u; L^2(\mathbb{S}^2))$.
\end{remark}

Given the component $\check{A}_2$ on $\scri^+$, one can recover the component $\check{A}_0$ on $\scri^+$ using the relation \eqref{temporalCoulombgaugescri}. That is, we may define
\[ \check{A}_0 \approx \int_{-\infty}^u 2 \operatorname{Re} \check{\eth} \bar{\check{A}}_2 \, \d u \in \dot{H}^2(\mathbb{R}_u; H^{-1}(\mathbb{S}^2)) \]
on $\scri^+$, where $\dot{H}^2(\mathbb{R})$ is the space of functions whose second derivative is in $L^2(\mathbb{R})$. We are now prepared to define the Hilbert space of characteristic data on $\scri^+$.

\begin{definition} \label{def:characteristicdatapotential_Minkowski} The Hilbert space $\dot{\mathcal{H}}^1(\scri^+)$ is the completion of the space (canonically isomorphic to $\mathcal{C}_c^\infty(\scri^+)$) of triplets
\[ (\check{A}^+_0, \check{A}^+_1, \check{A}^+_2) \in \mathcal{C}^\infty(\scri^+) \times \mathcal{C}_c^\infty(\scri^+) \times \mathcal{C}^\infty_c(\scri^+) \]
such that $\check{A}^+_1 \equiv 0$, and
\[ \check{A}^+_0 = \int_{-\infty}^u 2 \operatorname{Re} \check{\eth} \bar{\check{A}}^+_2 \, \d u, \]
in the norm \eqref{characteristicdatannormMinkowski_flat_frame}. This Hilbert space is the space of equivalence classes of functions (see \Cref{BeppoLeviR} below) in which two triplets are said to be equivalent if their difference has norm \eqref{characteristicdatannormMinkowski_flat_frame} equal to zero. The equality of two instances of $\check{A}_2^+$ in this norm identifies them up to the addition of constant-in-$u$ functions on $\mathbb{S}^2$, and the identification of two instances of the $\check{A}_0^+$ component of the triplet requires this function on the sphere to be $\check{\eth}$-constant. Therefore two triplets are equivalent if the $\check{A}_2^+$ components differ by a constant on $\scri^+$. Note that, as per \Cref{rmk:decompactification_scri_Minkowski}, $\dot{\mathcal{H}}^1(\scri^+) \simeq \dot{H}^1(\mathbb{R}_u; L^2(\mathbb{S}^2))$.
\end{definition}

\begin{remark}\label{BeppoLeviR} Because $\dot{\mathcal{H}}^1(\scri^+)$ consists of equivalence classes of functions, it is a Hilbert space but not a space of distributions. This is due to the fact that $\dot{H}^1 (\R^n)$ is defined as the completion of the space $\mathcal{C}^\infty_c(\mathbb{R}^n) \ni f$ in the norm $\| \nabla f \|_{L^2(\mathbb{R}^n)}$. Of course, constants have zero $\dot{H}^1 (\R^n)$ norm, but in dimension $1$ they can be approached in this norm by smooth compactly supported functions. In fact, this happens in dimensions $1$ and $2$. In dimensions strictly greater than $2$, we have Hardy's inequality on $\R^n$,
\[ \left \Vert \frac{f}{\vert x \vert} \right\Vert_{L^2 (\R^n)} \leq 2 \Vert \nabla f \Vert_{L^2 (\R^n)} \, ,\]
and analogous Sobolev-type inequalities for other topologies such as $\R \times \mathbb{S}^2$, that rule out constants in the completion. This is the reason why in definition \ref{DefHSpaceInitData} below, the space of Coulomb gauge initial data is a genuine distribution space. In the present situation, if one wished to make the space of scattering data a space of distributions, one could consider instead $\partial_u \check{A}_2^+$, not $\check{A}_2^+$, as the fundamental piece of data on $\scri^+$. Both descriptions contain the same information and are physically equivalent, and for us the fundamental characteristic data will be the component $\check{A}_2^+$ itself.
\end{remark}

\vspace{1em}
On the initial surface $\Sigma$, we have already seen that in the temporal-Coulomb gauge, the energy $\mathcal{E}_\Sigma$, as given in \eqref{Maxwellfieldinitialenergyflat}, is neatly expressed in terms of the physical potential $A_a$ as 
\begin{equation} \label{Maxwellpotentialinitialenergyflat} \mathcal{E}_\Sigma[F] = \frac{1}{2} \int_{\Sigma} \left( |\mathbf{E}|^2 + |\mathbf{B}|^2 \right) \dvol_\Sigma = \frac{1}{2} \int_\Sigma \left( | \dot{\mathbf{A}} |^2 + |\boldsymbol{\nabla} \mathbf{A} |^2 \right) \dvol_\Sigma \eqdef \mathcal{E}_\Sigma[\mathbf{A}].
\end{equation}

\begin{definition} \label{DefHSpaceInitData} For initial data $(\mathbf{A}, \dot{\mathbf{A}})|_\Sigma$ for the free Maxwell's equations, we define the Hilbert space $\dot{H}^1_C(\Sigma) \oplus L^2_C(\Sigma)$ of Coulomb gauge initial data by completion of smooth compactly supported Coulomb gauge initial data in the semi-norm
\[ \| ( \mathbf{A}, \dot{\mathbf{A}} ) \|^2_{\dot{H}^1 \oplus L^2} = \int_{\Sigma} \left( | \boldsymbol{\nabla} \mathbf{A} |^2 + | \dot{\mathbf{A}} |^2 \right) \dvol_\Sigma. \]
More precisely,
\[ \dot{H}^1_C(\Sigma) \oplus L^2_C(\Sigma) = \overline{ \{ (\mathbf{A}, \dot{\mathbf{A}}) \in \mathcal{C}^\infty_c(\Sigma) \oplus \mathcal{C}^\infty_c(\Sigma) \, : \, \boldsymbol{\nabla} \cdot \mathbf{A} = 0 = \boldsymbol{\nabla} \cdot \dot{\mathbf{A}} \} }^{\dot{H}^1 \oplus L^2} . \]

\end{definition}

\subsubsection{Construction of Trace Operators} \label{sec:constructionoftraceoperators}

Let 
\[ \mathcal{D}^\infty_c(\Sigma) \defeq \{ (\mathbf{A}, \dot{\mathbf{A}}) \in \mathcal{C}^\infty_c(\Sigma) \oplus \mathcal{C}^\infty_c(\Sigma) \, : \, \boldsymbol{\nabla} \cdot \mathbf{A} = 0 = \boldsymbol{\nabla} \cdot \dot{\mathbf{A}} \} \]
be the space of smooth compactly supported Coulomb gauge initial data for the physical Maxwell's equations in the temporal-Coulomb gauge \eqref{simultaneousgaugesMinkowski}. An element $\mathcalboondox{a} = (\mathbf{A}, \dot{\mathbf{A}})$ of $\mathcal{D}^\infty_c(\Sigma)$ defines smooth compactly supported initial data for the rescaled Maxwell's equations \eqref{rescaledMaxwellequationspotentialphysicalLorenzgauge} in the temporal-Coulomb gauge as follows. First,
\[ \hat{\mathfrak{a}}\big|_\Sigma = 0, \qquad \hat{\mathbf{A}}\big|_\Sigma = \mathbf{A}\big|_\Sigma. \]
For the time derivative part of the initial data, one computes the inverse relation to \eqref{physicaltorescaleddata},
\[ \nabla_T A_b \big|_\Sigma = \Omega \left( \hat{\nabla}_{\hat{T}} \hat{A}_b + \hat{T}^a \Upsilon_a \hat{A}_b + \Upsilon_b \hat{\mathfrak{a}} - \hat{T}_b \Upsilon_a \hat{A}^a \right)\Big|_\Sigma,  \]
so since $\partial_t \Omega|_\Sigma = 0$,
\[ - \hat{h}_a^b \hat{\nabla}_{\hat{T}} \hat{A}_b \big|_\Sigma = \Omega^{-1} \dot{\mathbf{A}}_a \big|_\Sigma. \]
Note that $\dot{\mathbf{A}}$ is supported away from $i^0$, so that $\Omega^{-1}$ is smooth on its support. Therefore we can solve \eqref{rescaledMaxwellequationspotentialphysicalLorenzgauge} as described in \Cref{sec:flatMaxwellequationsofmotion} (\Cref{thm:MaxwellpotentialflatCauchyproblem}) to get a unique smooth solution $\hat{A}_a$ satisfying the gauge conditions \eqref{physicalLorenzgaugeMaxwell} and \eqref{temporalgaugerescaled}. Using the smoothness of $\hat{A}_a$, one may take the trace of this solution on $\scri^+$ to get a smooth restriction $\hat{A}_a|_{\scri^+}$. Switching to the checked conformal scale, one obtains $(\check{A}_0^+, \check{A}_1^+, \check{A}_2^+) = (\check{A}_0, \check{A}_1, \check{A}_2)|_{\scri^+}$, where $\check{A}_0^+$ satisfies \eqref{temporalCoulombgaugescri} and $\check{A}_1^+ \equiv 0$.  One therefore has the linear map
\begin{align} \label{Maxwellflattraceoperator} \begin{split} \mathfrak{T}^+ : \mathcal{D}^\infty_c(\Sigma) & \longrightarrow \mathcal{C}^\infty(\scri^+) \times \mathcal{C}^\infty(\scri^+) \times \mathcal{C}^\infty(\scri^+), \\
(\mathbf{A}, \dot{\mathbf{A}}) & \longmapsto (\check{A}_0^+, \check{A}_1^+, \check{A}_2^+),
\end{split}	
\end{align}
where $\check{A}_{0,1,2}$ are supported away from $i^0$. The energy estimate \eqref{Maxwellflatenergyestimate} implies that there exists a constant $C>0$ such that for all $\mathcalboondox{a} \in \mathcal{D}^\infty_c(\Sigma)$
\begin{equation} \label{Maxwellflatforwardtraceestimate} \| \mathfrak{T}^+ \mathcalboondox{a} \|_{\dot{\mathcal{H}}^1(\scri^+)} \leq C \| \mathcalboondox{a} \|	_{\dot{H}^1 \oplus L^2},
\end{equation}
and
\begin{equation} \label{Maxwellflatbackwardtraceestimate} \| \mathcalboondox{a} \|_{	\dot{H}^1 \oplus L^2} \leq C \| \mathfrak{T}^+ \mathcalboondox{a} \|_{\dot{\mathcal{H}}^1(\scri^+)}.
\end{equation}
By \eqref{Maxwellflatforwardtraceestimate} and the density of $\mathcal{D}^\infty_c(\Sigma)$ in $\dot{H}^1_C(\Sigma) \oplus L^2_C(\Sigma)$, the bounded linear operator $\mathfrak{T}^+$ extends uniquely to a bounded linear operator from $\dot{H}^1_C(\Sigma) \oplus L^2_C(\Sigma)$ into $\dot{\mathcal{H}}^1(\scri^+)$. Moreover, the reverse estimate \eqref{Maxwellflatbackwardtraceestimate} ensures that $\mathfrak{T}^+$ is an isomorphism from $\dot{H}^1_C(\Sigma) \oplus L^2_C(\Sigma)$ to its image, and that the image is a closed subspace of $\dot{\mathcal{H}}^1(\scri^+)$.

\begin{definition} The bounded linear operator
\[ \mathfrak{T}^+ : \dot{H}^1_C(\Sigma) \oplus L^2_C(\Sigma) \longrightarrow \dot{\mathcal{H}}^1(\scri^+) \]
that takes the initial data for \eqref{rescaledMaxwellequationspotentialphysicalLorenzgauge} on $\Sigma$ to the characteristic data on $\scri^+$ is called the \emph{future trace operator} for the free Maxwell's equations in the gauge \eqref{simultaneousgaugesMinkowski}.
\end{definition}

 \begin{remark} \label{rmk:surjectivity_of_trace_operator}
 To show that $\mathfrak{T}^+$ is surjective (and hence an isomorphism between $\dot{H}^1_C(\Sigma) \oplus L^2_C(\Sigma)$ and $\dot{\mathcal{H}}^1(\scri^+)$), it is enough to show that its range is dense in $\dot{\mathcal{H}}^1(\scri^+)$, i.e. that for every $\mathcalboondox{b} = (\check{A}_0^+, 0, \check{A}_2^+) \in \mathcal{C}^\infty(\scri^+) \times \mathcal{C}^\infty_c(\scri^+) \times \mathcal{C}^\infty_c(\scri^+)$ there exists a unique $\mathcalboondox{a} \in \dot{H}^1_C(\Sigma) \oplus L^2_C(\Sigma)$ such that $\mathfrak{T}^+ \mathcalboondox{a} = \mathcalboondox{b}$. Indeed, then the inverse trace operator can be extended to $\dot{\mathcal{H}}^1(\scri^+)$ as follows. For any $\mathcalboondox{b} \in \dot{\mathcal{H}}^1(\scri^+)$ we can find a sequence $\{ \mathcalboondox{b}_n \}_n \subset \mathcal{C}^\infty(\scri^+) \times \mathcal{C}^\infty_c(\scri^+) \times \mathcal{C}^\infty_c(\scri^+)$ such that $\mathcalboondox{b}_n \to \mathcalboondox{b}$ in $\dot{\mathcal{H}}^1(\scri^+)$. Then for each $n$ there exists a unique $\mathcalboondox{a}_n \in \dot{H}^1_C(\Sigma) \oplus L^2_C(\Sigma)$ such that $\mathcalboondox{b}_n = \mathfrak{T}^+ \mathcalboondox{a}_n$, and
\begin{equation} \label{traceoperatorlimitextension} \| \mathfrak{T}^+ \mathcalboondox{a}_n - \mathcalboondox{b} \|_{\dot{\mathcal{H}}^1(\scri^+)} \longrightarrow 0.	
\end{equation}
The above estimates easily imply that the sequence $\{ \mathcalboondox{a}_n \}_n$ is Cauchy, since
\[ \| \mathcalboondox{a}_n - \mathcalboondox{a}_m \|_{\dot{H}^1 \oplus L^2} \la \| \mathfrak{T}^+ \mathcalboondox{a}_n - \mathfrak{T}^+ \mathcalboondox{a}_m \|_{\dot{\mathcal{H}}^1(\scri^+)} \leq \| \mathcalboondox{b}_n - \mathcalboondox{b}_m \|_{\dot{\mathcal{H}}^1(\scri^+)}. \]
Therefore there exists $\mathcalboondox{a} \in \dot{H}^1_C(\Sigma) \oplus L^2_C(\Sigma)$ such that $\mathcalboondox{a}_n \to \mathcalboondox{a}$ in $\dot{H}^1_C(\Sigma) \oplus L^2_C(\Sigma)$, and by \eqref{traceoperatorlimitextension} $\mathfrak{T}^+ \mathcalboondox{a} = \mathcalboondox{b}$.
Proving that for every $\mathcalboondox{b} \in \mathcal{C}^\infty(\scri^+) \times \mathcal{C}^\infty_c(\scri^+) \times \mathcal{C}^\infty_c(\scri^+)$ there exists a unique $\mathcalboondox{a} \in \dot{H}^1_C(\Sigma) \oplus L^2_C(\Sigma)$ such that $\mathfrak{T}^+ \mathcalboondox{a} = \mathcalboondox{b}$ amounts to solving the Goursat problem from $\scri^+$.
\end{remark}

\subsubsection{The Goursat Problem}
\label{sec:Minkowski_Goursat_problem}

The underlying analytic tool that we shall use to resolve the Goursat problem is B\"ar--Wafo's formulation \cite{BarWafo2015} of a theorem due to H\"ormander \cite{Hormander1990}.

\begin{theorem}[H\"ormander; B\"ar--Wafo, \cite{BarWafo2015} Theorem 23] \label{thm:Goursatproblem}
Let $\hat{\mathcalboondox{M}}$ be a globally hyperbolic Lorentzian manifold (of any dimension) and let $\mathcal{S} \subset \hat{\mathcalboondox{M}}$ be a characteristic partial Cauchy hypersurface. Assume that $J^+(\mathcal{S})$ is past compact. Then for any $f \in L^2_{\mathrm{loc}, \mathrm{sc}}(\hat{\mathcalboondox{M}})$ and any $u_0 \in H^1_c(\mathcal{S})$ there exists $u \in \mathcal{C}^0_{\mathrm{sc}}( \tau(\hat{\mathcalboondox{M}}); H^1(\mathcal{S}_\circ)) \cap \mathcal{C}^1_{\mathrm{sc}}(\tau(\hat{\mathcalboondox{M}}); L^2(\mathcal{S}_\circ))$ such that $P u = f$ on $J^+(\mathcal{S})$, and $u|_\mathcal{S} = u_0$. On $J^+(\mathcal{S})$, $u$ is unique.
\end{theorem}

\begin{remark} Here $\tau(\hat{\mathcalboondox{M}})$ denotes a choice of a time function on $\hat{\mathcalboondox{M}}$, where $\mathcal{S}_\circ$ (shorthand for $\{ \mathcal{S}_\tau \}_\tau$) are the leaves of the foliation corresponding to this time function, with, say, $\mathcal{S}_1 = \mathcal{S}$. In the above theorem $u$ and $f$ are permitted to be quite general real or complex sections of a vector bundle $S \to \hat{\mathcalboondox{M}}$ over $\hat{\mathcalboondox{M}}$. In particular, the theorem applies to equations on 1-forms and systems of coupled equations. The operator $P$ is a linear wave operator (a hyperbolic second order differential operator whose principal symbol is the metric on $\hat{\mathcalboondox{M}}$), and a \emph{partial} Cauchy surface is a closed achronal hypersurface $\mathcal{S} \subset \hat{\mathcalboondox{M}}$. In particular, $\mathcal{S}$ does not need to be compact, and includes both the cases when $\mathcal{S}$ is a lightcone and an intersection of two null hyperplanes. The subscript ${}_{\text{sc}}$ denotes spaces of sections which are \emph{spatially compact}. When $\mathcal{S}$ is smooth and spacelike, $H^1_c(\mathcal{S})$ is the space of $H^1$ sections on $\mathcal{S}$ which have compact support. When $\mathcal{S}$ is merely Lipschitz, as in the case of a lightcone, the space $H^1_c(\mathcal{S})$ is the space of $\mathscr{FE}^1_{\text{sc}} \defeq \mathcal{C}^0_{\text{sc}}(\tau(\hat{\mathcalboondox{M}}); H^1(\mathcal{S}_\circ)) \cap \mathcal{C}^1_{\text{sc}}(\tau(\hat{\mathcalboondox{M}}); L^2(\mathcal{S}_\circ))$ sections restricted to $\mathcal{S}$; the space $H^1_c(\mathcal{S})$ in this case is well-defined because the space $\mathscr{FE}^1_{\text{sc}}$ does not depend on the choice of time function $\tau(\hat{\mathcalboondox{M}})$ (Cor. 19, \cite{BarWafo2015}). The lower case subscripts, as in $H^1_c$, should not be confused with upper case subscripts, as in $\dot{H}^1_C$, the space of $\dot{H}^1$ sections satisfying the Coulomb gauge $\boldsymbol{\nabla} \cdot \mathbf{A} = 0$).
\end{remark}

\vspace{3pt}

In this section we prove the following.
\begin{theorem} \label{thm:Goursat_solution_Minkowski}
For every $(\check{A}_0^+, 0, \check{A}_2^+) \in \mathcal{C}^\infty(\scri^+) \times \mathcal{C}^\infty_c(\scri^+) \times \mathcal{C}^\infty_c(\scri^+)$ there exists a unique solution $\check{A}_a \in \mathcal{C}^0(\mathbb{R}_t; H^1(\Sigma_t)) \cap \mathcal{C}^1(\mathbb{R}_t; L^2(\Sigma_t))$ to \eqref{rescaledMaxwellflatnonsingular}, for which the corresponding physical potential satisfies
\[ (\mathbf{A}, \dot{\mathbf{A}})|_\Sigma \in \dot{H}^1_C(\Sigma) \oplus L^2_C(\Sigma). \]
Moreover, the corresponding physical potential $A_a$ satisfies the temporal-Coulomb gauge throughout $\mathcalboondox{M}$.
\end{theorem}

\begin{proof}

We wish to solve the system \eqref{rescaledMaxwellflatnonsingular} from characteristic data in (a dense subspace of) $\dot{\mathcal{H}}^1(\scri^+)$ (cf. \eqref{Maxwellflattraceoperator}). We first construct a solution in a small neighbourhood of $\scri^+ \cup i^+$, which can then be easily extended to the rest of the spacetime. Given
\[ \left( \check{A}^+_0 = \int_{-\infty}^u 2 \operatorname{Re} \check{\eth} \bar{\check{A}}^+_2 \, \d u, \, \check{A}^+_1 \equiv 0, \, \check{A}_2^+ \right) \in \mathcal{C}^\infty(\scri^+) \times \mathcal{C}^\infty_c(\scri^+) \times \mathcal{C}^\infty_c(\scri^+) \subset \dot{\mathcal{H}}^1(\scri^+), \]
we observe that in the future of the support of $\check{A}_2^+$, the component $\check{A}_0^+$ is constant on $\scri^+$. The fact that $\hat{A}_0^+$ cannot be supported away from $i^+$ means that we must proceed carefully. Introduce a short outgoing null hypersurface $\mathcal{H}$ which intersects $\scri^+$ in the future of the support of $\check{A}_2^+$, as depicted in \Cref{fig:outgoing_null_hypersurface_Minkowski}, by first choosing a spacelike hypersurface $\Sigma_{t'}$ for $t'$ sufficiently large, and then choosing $\mathcal{H}$ to be the future lightcone of a point on $\Sigma_{t'}$. 
\begin{figure}[H]
\begin{tikzpicture}
\centering
\node[inner sep=0pt] (goursat_problem_outgoing_null_hypersurface_Minkowski) at (8,0)
    {\includegraphics[width=.4\textwidth]{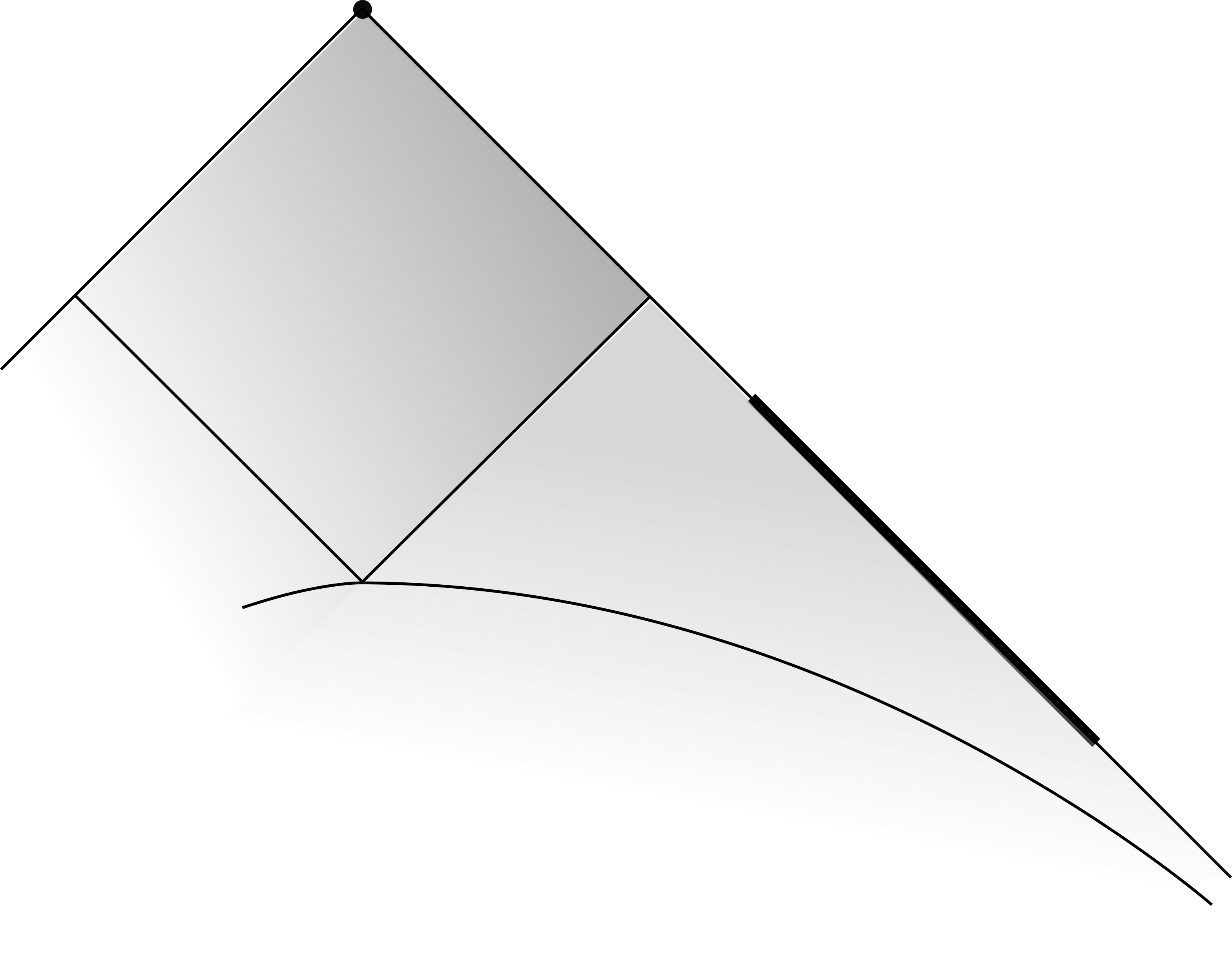}};

\node[label={[shift={(6.7,2.4)}]$i^+$}] {};
    
\node[label={[shift={(5.1, 1.2)}]$\mathscr{I}^+$}] {};
\node[label={[shift={(8.2, 1.2)}]$\mathscr{I}^+$}] {};

\node[label={[shift={(6.8, 0.75)}]$\mathcal{O}^+$}] {};

\node[label={[shift={(8.3, -0.5)}]$\mathcal{O}^-$}] {};
\node[label={[shift={(8.3, -1.5)}]$\Sigma_{t'}$}] {};

\node[label={[shift={(13,-0.5)}]$\operatorname{supp}\hat{A}_2^+$}] {};
\draw[->] (12.2,-0.05) .. controls (10.8,0.2) and (10.5, 0) .. (9.8,-0.47);

\node[label={[shift={(6.9,-0.25)}]$\mathcal{H}$}]{};

\draw[fill=white] (6.68,2.45) circle (1.5pt);

\end{tikzpicture} \vspace{-0.25cm}
\caption{The solution in the region $\mathcal{O}^+$ near $i^+$ is pure gauge.} 
\label{fig:outgoing_null_hypersurface_Minkowski}
\end{figure}
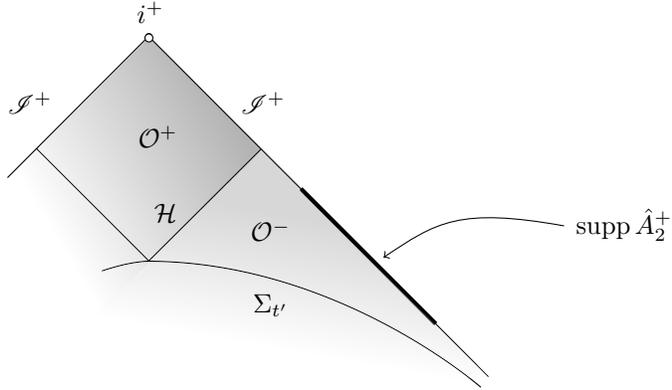
Denote the future of $\mathcal{H}$ (i.e. a neighbourhood of $i^+$) by $\mathcal{O}^+$, and the past of $\mathcal{H}$ in the future of $\Sigma_{t'}$ by $\mathcal{O}^-$. In the region $\mathcal{O}^+$ near $i^+$, the data for the potential is trivial in the following sense: $\check{A}_1^+ = 0 = \check{A}^+_2$, and $\check{A}_0^+ = C_0$, where $C_0$ is a constant. The rescaled \emph{field} $\check{F}_{ab}$ therefore has identically zero data on $\scri^+ \cap \overline{J^+(\mathcal{H})}$. Changing to the hatted conformal scale in $\mathcal{O}^+$ (so that $i^+$ is at a finite distance), the data for the field $\hat{F}_{ab}$ remains identically zero, so we may solve a wave equation for $\hat{F}_{ab}$, which, by the uniqueness part of \Cref{thm:Goursatproblem}, must be identically zero in $\mathcal{O}^+$, $\hat{F}_{ab} \equiv 0$. Return now to the checked conformal scale, where we have deduced that $\check{F}_{ab} \equiv 0$ in $\mathcal{O}^+$. By the Poincar\'e Lemma, there exists a function $\xi \in H^2_{\text{loc}}(\mathcal{O}^+)$ such that $A_a = \check{A}_a = \partial_a \xi$, i.e. $\check{A}_a$ is pure gauge. The gauge function $\xi$ is a priori not unique, but recalling that our solution should satisfy the gauge \eqref{simultaneousgaugesMinkowski}, we have $\Box \xi = 0$ in the physical spacetime, or equivalently in the checked conformal scale
\[ \widecheck{\Box} \check{\xi} = 0, \]
where $\check{\xi} = R^{-1} \xi$ (as the Ricci scalar vanishes in the checked scale, $\check{\mathrm{R}} = 0$). To get uniqueness of $\check{\xi}$, it is therefore enough to fix $\check{\xi}|_{\scri^+}$. But since the condition $\check{A}_0 \approx C_0$ in $\mathcal{O}^+$ is equivalent to $\check{\xi} \approx - C_0$ (this follows simply from the fact that $\check{A}_0 = \check{l}^a \partial_a \xi = \check{l}^a \partial_a (R \check{\xi}) \approx (\check{l}^a \partial_a R) \check{\xi} \approx - \check{\xi}$), this fixes $\check{\xi}$ and therefore $\xi$ in $\mathcal{O}^+$. As a note, the temporal gauge further implies $\partial_u \xi = 0$.

\begin{remark} It is worth noting that, although in $\mathcal{O}^+$ the solution is pure gauge, at the level of the potential the data on $\scri^+$ is \emph{divergent} in a conformal scale in which $i^+$ is at a finite distance (e.g. our hatted conformal scale). Indeed, $\check{A}_0^+$ is constant near $i^+$, and $\hat{A}_0^+ = \chi^{-2} \check{A}_0^+ = (1+u^2) \check{A}_0^+ \to \infty$ as $u \to \infty$. The $L^p(\scri^+)$ norm of $\hat{A}_0^+$ diverges for $p \geq 1$, the $p=1$ norm being conformally invariant.
\end{remark}

We now compute $\check{\thorn} \xi$ and $\check{\eth} \xi$ in $\mathcal{O}^+$, and restrict these to $\mathcal{H}$ to obtain $H^1_c$ data for $\check{A}_0$ and $\check{A}_2$ on $\mathcal{H}$. The data for $(\check{A}_0, \check{A}_2)$ is now in $H^1_c(\mathcal{H} \cup (\scri^+ \cap \overline{\mathcal{O}^-}))$. By contracting the system \eqref{rescaledMaxwellflatnonsingular} with $\check{l}^a$ and $\check{m}^a$ and using the temporal gauge condition \eqref{temporalgaugerescaled} to eliminate all instances of the component $\check{A}_1$, one derives a system of the form
\begin{equation}
    \label{Minkowski_system_from_scri}
    \begin{cases}
    \widecheck{\Box} \check{A}_0 + \check{L}_1^{(11)} \check{A}_0 + \check{L}_1^{(12)} \check{A}_2 = 0, \\
    \widecheck{\Box} \check{A}_2 + \check{L}_1^{(21)} \check{A}_0 + \check{L}_1^{(22)} \check{A}_2  = 0,
    \end{cases}
\end{equation}
where $\check{L}_1^{(ij)}$, $i, j \in \{1, 2\}$ are linear first order differential operators with coefficients depending on the spin coefficients and curvature components of $\check{\mathcalboondox{M}}$ (note that, by \eqref{Minkowski_rescaled_spin_coefficients}, the only non-vanishing spin coefficients are $\check{\alpha}$, $\check{\beta}$, and $\check{\gamma}$, and the only non-vanishing curvature component is $\check{\Phi}_{11} = \frac{1}{2}$). Choosing the time function to be the standard physical coordinate $t$ with corresponding leaves $\Sigma_t$, we may therefore apply \Cref{thm:Goursatproblem} from $\mathcal{H} \cup (\scri^+ \cap \overline{\mathcal{O}^-})$
 to solve the system \eqref{Minkowski_system_from_scri} in the region $\mathcal{O}^-$; we obtain in a neighbourhood of $\scri^+$ the components
 \[ \check{A}_0, \, \check{A}_2 \in \mathcal{C}^0(\mathbb{R}_t; H^1(\Sigma_t)) \cap \mathcal{C}^1(\mathbb{R}_t; L^2(\Sigma_t)), \]
 and reconstruct the remaining component by setting
 \[ \check{A}_1 \defeq - \frac{1}{2} R^2 \check{A}_0, \]
 which, of course, in $\mathcal{O}^+$ is equivalent to $\partial_u \xi = 0$.

To extend the solution to the rest of the spacetime, we pick a sufficiently large $t' < \infty$ as in \Cref{fig:Goursatproblem}, and on $\Sigma_{t'}$ reconstruct the physical field $F_{ab}|_{\Sigma_{t'}} \in L^2(\Sigma_{t'})$ from $\check{A}_a = A_a$. We then propagate $F_{ab}$ backwards in time to $\Sigma$, as shown in \Cref{fig:Goursatproblem}, using standard theory \cite{Leray1953}. As the solution $F_{ab}$ is spatially compact, it is straightforward to check, by performing an energy estimate as in \Cref{thm:aprioriestimatesflatMaxwell}, that $\mathcal{E}_\Sigma = \mathcal{E}_{\scri^+}$, and therefore the restriction of the solution to $\Sigma$ satisfies
\[ \mathcal{E}_\Sigma = \frac{1}{2} \int_\Sigma ( |\mathbf{E}|^2 + |\mathbf{B}|^2) \dvol_\Sigma < \infty. \]

\begin{figure}[H]
\begin{tikzpicture}
\centering
\node[inner sep=0pt] (Goursatproblem) at (8,0)
    {\includegraphics[width=.32\textwidth]{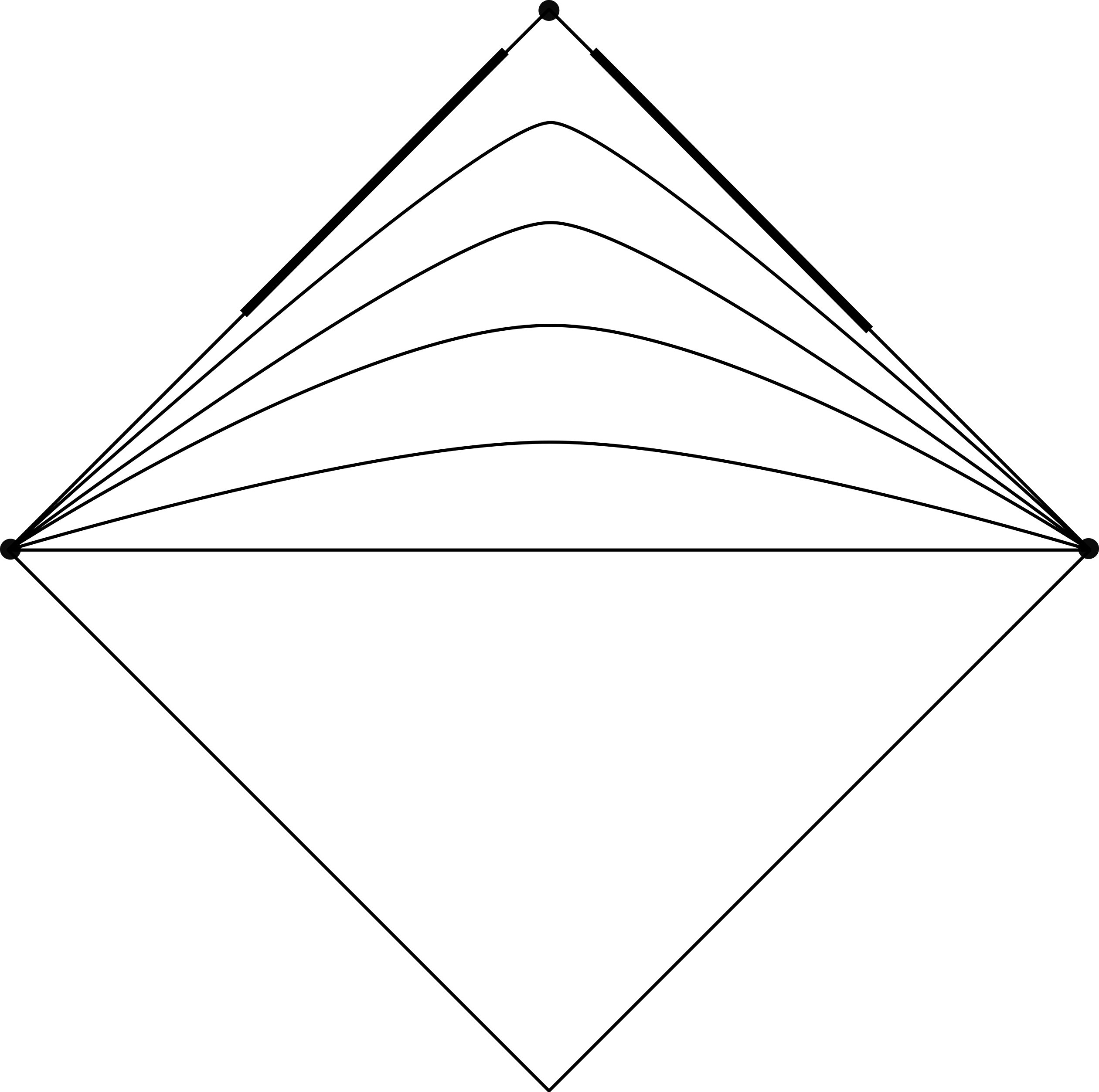}};

\node[label={[shift={(8.1,-3.3)}]$i^-$}] {};
\node[label={[shift={(8.1,2.5)}]$i^+$}] {};
    
\node[label={[shift={(6.5, 1.2)}]$\mathscr{I}^+$}] {};
\node[label={[shift={(9.6, 1.2)}]$\mathscr{I}^+$}] {};

\node[label={[shift={(5.1,-0.3)}]$i^0$}]{};
\node[label={[shift={(10.9,-0.3)}]$i^0$}]{};

\node[label={[shift={(11.55,0.45)}]$\Sigma_{t'}$}]{};
\draw[->] (11.3,0.85) .. controls (10.6,0.7) and (8.6, 0.7) .. (7.6, 1.7);

\node[label={[shift={(12,2)}]support of $\hat{A}^+_2$}] {};
\draw[->] (10.7,2.45) .. controls (9.3,2.7) and (9, 2.5) .. (8.58,2.09);
\draw[->] (10.7,2.45) .. controls (8.4,3.5) and (6, 4) .. (7.44, 2.13);

\node[label={[shift={(8.05,-0.7)}]$\Sigma$}]{};

\draw[fill=white] (5.5,-0.02) circle (1.5pt);
\draw[fill=white] (10.5,-0.02) circle (1.5pt);

\end{tikzpicture} \vspace{-0.25cm}
\caption{We solve the Goursat problem from $\scri^+$ backwards in time.} \label{fig:Goursatproblem}
\end{figure}
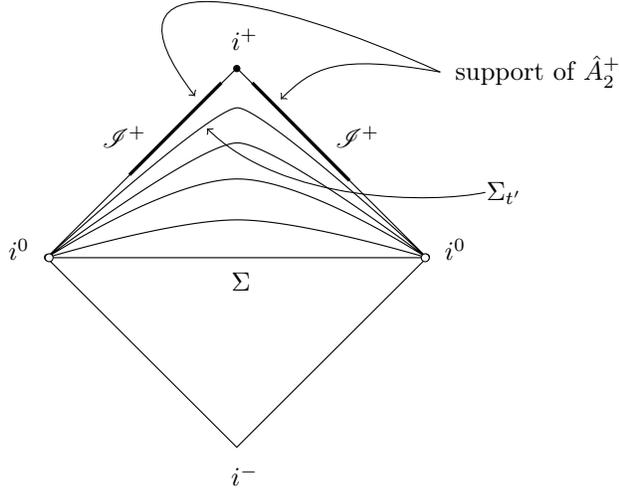

To see that the corresponding physical potential $(\mathbf{A}, \dot{\mathbf{A}})$ is in the desired function space $\dot{H}^1_C(\Sigma) \oplus L^2_C(\Sigma)$, it remains to check that the complete gauge \eqref{simultaneousgaugesMinkowski} is propagated off $\scri^+$. To do this, note that the temporal gauge holds in a neighbourhood of $\scri^+$ by construction of the component $\check{A}_0$. Further, we show that the physical Lorenz gauge $\nabla_a A^a = 0$ is propagated off $\scri^+$. The two will then imply the physical Coulomb gauge $\boldsymbol{\nabla} \cdot \mathbf{A} = 0$. Recall that our solution $\check{A}_a$ to \eqref{Minkowski_system_from_scri} solves \eqref{rescaledMaxwellflatnonsingular}, which is in turn equivalent to
\[ \Box A_a = 0 \]
in $\mathcalboondox{M}$. We commute $\nabla_a$ into this equation and define $\psi \defeq \nabla_a A^a$ to get
\[ \Box \psi = 0 \]
in $\mathcalboondox{M}$. As $\mathrm{R}  = 0$ on $\mathcalboondox{M}$, this is the conformally invariant scalar wave equation, so is equivalent to  (in the hatted conformal scale with $i^+$ finite)
\[
\widehat{\Box} \hat{\psi} + \frac{1}{6} \hat{\mathrm{R}} \hat{\psi} = 0
\]
on $\hat{\mathcalboondox{M}}$, where $\hat{\psi} = \Omega^{-1} \psi$. By the uniqueness part of \Cref{thm:Goursatproblem}, we will have $\hat{\psi} \equiv 0 $ in $\hat{\mathcalboondox{M}}$ if we can demonstrate that $\hat{\psi}^+ \defeq \hat{\psi}|_{\scri^+} = 0$. But now
\begin{align*} \hat{\psi} & = \Omega^{-1} \nabla_a A^a \\
& = \Omega ( \hat{\nabla}_a \hat{A}^a - 2 \Upsilon_a \hat{A}^a ) \\
& = -2 (\hat{\nabla}^a \Omega) \hat{A}_a + \mathcal{O}(\Omega) \\
& = - 2f \hat{A}_1 + \mathcal{O}(\Omega),
\end{align*}
and we have $\hat{A}_1 \approx 0$, which implies $\hat{\psi}^+ = 0$.

\end{proof}

\subsubsection{The Scattering Operator}

\begin{corollary}
    The forward trace operator $\mathfrak{T}^+ : \dot{H}^1_C(\Sigma) \oplus L^2_C(\Sigma) \longrightarrow \dot{\mathcal{H}}^1(\scri^+)$ is invertible and hence a linear isomorphism.
\end{corollary}
\begin{proof}
    This follows from \Cref{thm:Goursat_solution_Minkowski} and \Cref{rmk:surjectivity_of_trace_operator}.
\end{proof}

An analogous construction can be performed to the past of the initial surface $\Sigma$ to construct the \emph{past trace operator}
\[ \mathfrak{T}^- : \dot{H}^1_C(\Sigma) \oplus L^2_C(\Sigma) \longrightarrow \dot{\mathcal{H}}^1(\scri^-), \]
which is an isomorphism by the same token. We are therefore now in a position to define the scattering operator $\mathscr{S}$.

\begin{definition}[Scattering operator on Minkowski space] \label{def:flatMaxwellscatteringoperator} We call the linear isomorphism of Hilbert spaces
\[ \mathscr{S} \defeq \mathfrak{T}^+ \circ (\mathfrak{T}^- )^{-1} : \dot{\mathcal{H}}^1(\scri^-) \longrightarrow \dot{\mathcal{H}}^1(\scri^+) \]
taking finite energy characteristic data for the Maxwell potential on $\scri^-$ to finite energy characteristic data on $\scri^+$ the \emph{conformal scattering operator for Maxwell potentials in temporal-Coulomb gauge} on Minkowski space.	
\end{definition}

\subsubsection{Alternative Formulations} \label{sec:alternativeformulations}

The preceding construction of the scattering operator $\mathscr{S}$ is predicated on the usage of the multiplier Killing vector field
\[ K^a = \partial_u, \]
which, via natural energy estimates (cf. \Cref{sec:aprioriestimatesflatMaxwell}), defines the semi-norms on $\scri^\pm$
\begin{align*} & \mathcal{E}^K_{\scri^+} \simeq \int_{\scri^+} |\hat{F}_2|^2 \, \widehat{\dvol}_{\scri^+} \simeq \int_{\scri^+} |\partial_u(\chi \hat{A}_2)|^2 \chi^{-2} \, \widehat{\dvol}_{\scri^+} = \int_{\scri^+} | \partial_u\check{A}_2|^2 \, \d u \wedge \dvol_{\mathbb{S}^2} ,  \\
& \mathcal{E}^K_{\scri^-} \simeq \int_{\scri^-} | \hat{F}_0 |^2 \, \widehat{\dvol}_{\scri^-} \simeq \int_{\scri^-} |\partial_v (\chi \hat{A}_2)|^2 \chi^{-2} \, \widehat{\dvol}_{\scri^-} = \int_{\scri^-} | \partial_v \check{A}_2|^2 \, \d v \wedge \dvol_{\mathbb{S}^2}.
\end{align*}
However, one has many alternative choices for $K^a$ on Minkowski space. Indeed, inspecting the proof of \Cref{thm:aprioriestimatesflatMaxwell}, one sees that any uniformly timelike conformally\footnote{The multiplier vector field is allowed to be merely \emph{conformally} Killing on $\mathcalboondox{M}$ because the Maxwell stress-energy tensor is traceless.} Killing vector field on $\mathcalboondox{M}$ will do. One particular choice which is tied to the conformal structure of Minkowski space is the Morawetz vector field
\begin{equation} \label{Morawetzvectorfield} K_0^a = u^2 \partial_u + 2 r (u + r) \partial_r,
\end{equation}
discovered by Cathleen Morawetz in 1961 in her study of the decay of solutions to the wave equation in the exterior of an obstacle \cite{Morawetz1961}. The vector field $K^a_0$ is conformally Killing on $(\mathcalboondox{M}, \eta)$,
\[ \pounds_{K_0} \eta_{ab} = 4 (u + r) \eta_{ab}, \]
and in fact exactly Killing with respect to $R^2 \eta_{ab}$,
\[ \pounds_{K_0} ( R^2 \eta_{ab} ) = 0. \]
If one uses $K^a_0$ instead of $K^a$ in the energy estimates, one arrives at the following energies on $\Sigma$ and $\scri^\pm$,
\begin{align}
	\begin{split} \label{Morawetzenergyfield}
 \mathcal{E}^{K_0}_\Sigma &\simeq \int_\Sigma r^2 	\left( |F_0|^2 + |F_1|^2 + |F_2|^2 \right) \dvol_\Sigma \\
& \simeq \int_\Sigma r^2 \left( | \mathbf{E} |^2 + |\mathbf{B} |^2 \right) \dvol_\Sigma \\
& = \int_\Sigma \left( r^2 | \dot{\mathbf{A}} |^2 + r^2 |\boldsymbol{\nabla} \mathbf{A} |^2 - 2 |\mathbf{A} |^2 \right) \dvol_\Sigma,
	\end{split}
\end{align}
and 
\[ \mathcal{E}^{K_0}_{\scri^+} \simeq \int_{\scri^+} \left( u^2 | \hat{F}_2 |^2 + \chi^2 | \hat{F}_1 |^2 \right) \widehat{\dvol}_{\scri^+},
\]
which may be written, in the checked conformal scale on $\scri^+$, as
\begin{equation} \label{MaxwellMorawetzenergy} \mathcal{E}^{K_0}_{\scri^+} \simeq \int_{\scri^+} \left( u^2 |\partial_u \check{A}_2 |^2 + | \check{\eth} \bar{\check{A}}_2 |^2 \right) \d u \wedge \dvol_{\mathbb{S}^2}.
\end{equation}
An analogous expression exists on $\scri^-$. While the energies $\mathcal{E}^{K_0}_{\scri^\pm}$ on $\scri^\pm$ define weighted Sobolev semi-norms on $\hat{A}_2$ (and in this case also control the angular derivatives of $\hat{A}_2$), the energy $\mathcal{E}^{K_0}_{\Sigma}$ on $\Sigma$ no longer defines a (weighted) Sobolev semi-norm in terms of $(\dot{\mathbf{A}}, \mathbf{A})$ due to the presence of the negative-definite term $-2 |\mathbf{A}|^2$. This means that the space of data on $\Sigma$ has to be defined slightly differently in this context. As before, we have the trace operators
\begin{align*} \mathfrak{T}^\pm_{K_0} : \mathcal{D}^\infty_c(\Sigma) & \longrightarrow \mathcal{C}^\infty(\scri^\pm) \times \mathcal{C}^\infty(\scri^\pm) \times \mathcal{C}^\infty(\scri^\pm), \\
(\mathbf{A}, \dot{\mathbf{A}}) & \longmapsto (\check{A}_0^\pm, \check{A}_1^\pm, \check{A}_2^\pm)	
\end{align*}
from smooth initial data to smooth characteristic data, but instead of completing $\mathcal{D}^\infty_c(\Sigma)$ in the semi-norm $\dot{H}^1 \oplus L^2$, we shall show that the pairs $(\mathbf{A}, \dot{\mathbf{A}})$ are in 1-to-1 correspondence with finite energy Maxwell fields on $\Sigma$ in the natural energy space \eqref{Morawetzenergyfield}. Indeed, in Coulomb gauge on $\Sigma$ one has
\[ \mathbf{E}	= \dot{\mathbf{A}}, \qquad \mathbf{B} = \frac{1}{2} \boldsymbol{\nabla} \times \mathbf{A} .
\]
The time derivative component is therefore recovered trivially, whereas to recover $\mathbf{A}$ from $\mathbf{B}$ on $\Sigma$ we take the curl,
\begin{equation} \label{curlB} \boldsymbol{\nabla} \times \mathbf{B} = \frac{1}{2} \left( \boldsymbol{\nabla} ( \boldsymbol{\nabla} \cdot \mathbf{A} ) - \boldsymbol{\Delta} \mathbf{A} \right) = - \frac{1}{2} \boldsymbol{\Delta} \mathbf{A}. \end{equation}
For $\mathbf{B} \in L^2(\Sigma)$ there exists a unique solution $\mathbf{A} \in \dot{H}^1(\Sigma)$ to \eqref{curlB}, which we write as $\mathbf{A} = \boldsymbol{\Delta}^{-1} ( - 2 \boldsymbol{\nabla} \times \mathbf{B} )$. Indeed, note that $\mathcal{C}^\infty_c(\Sigma)$ is dense in $\dot{H}^1(\Sigma)$; multiplying \eqref{curlB} by a test vector field $\mathbf{X} \in \mathcal{C}^\infty_c(\Sigma)$ and integrating by parts, we obtain a continuous and strictly coercive bilinear form on $\dot{H}^1(\Sigma)$. It therefore follows from the Lax--Milgram Lemma that $\mathbf{\Delta} : \dot{H}^1(\Sigma) \to \dot{H}^{-1}(\Sigma)$ is an isomorphism, where $\dot{H}^{-1}(\Sigma)$ is the dual space of $\dot{H}^1(\Sigma)$ \cite{AdamsFournier2003}. It therefore remains to see that for $\mathbf{B} \in L^2(\Sigma)$ we have $\boldsymbol{\nabla} \times \mathbf{B} \in \dot{H}^{-1}(\Sigma)$. This is true: for a test vector field $\mathbf{X} \in \mathcal{C}^\infty_c(\Sigma)$,
\[ \langle \boldsymbol{\nabla} \times \mathbf{B}, \, \mathbf{X} \rangle_{\mathcal{D}'(\Sigma), \, \mathcal{C}^\infty_c (\Sigma)} = \int_{\Sigma} (\boldsymbol{\nabla} \times \mathbf{B}) \cdot \mathbf{X} \dvol_\Sigma = \int_{\Sigma} (\boldsymbol{\nabla} \times \mathbf{X}) \cdot \mathbf{B} \dvol_\Sigma \leq \| \mathbf{X} \|_{\dot{H}^1(\Sigma)} \| \mathbf{B} \|_{L^2(\Sigma)}.  \]

If $r \mathbf{E} \in L^2(\Sigma)$, it is obvious that $r \dot{\mathbf{A}} \in L^2(\Sigma)$, and we write $\dot{\mathbf{A}} \in r^{-1} L^2(\Sigma) \subset L^2(\Sigma)$. Also, $r \mathbf{B} \in L^2(\Sigma) \implies \mathbf{B} \in L^2(\Sigma) \implies \boldsymbol{\nabla} \times  \mathbf{B} \in \dot{H}^{-1}(\Sigma)$, and so $\mathbf{A} \in \dot{H}^1(\Sigma)$. We define
\[ r^{-1} \dot{H}^1_{C}(\Sigma)^\mathrm{curl} \defeq \left\{ \mathbf{A} \in \dot{H}^1(\Sigma) \, :  \, \boldsymbol{\nabla} \cdot \mathbf{A} = 0, ~ r( \boldsymbol{\nabla} \times \mathbf{A} ) \in L^2(\Sigma) \right\} \]
and
\[ r^{-1} L^2_C(\Sigma) \defeq \left\{ \dot{\mathbf{A}} \in L^2(\Sigma) \, : \, \boldsymbol{\nabla} \cdot \dot{\mathbf{A}} = 0, ~ r \dot{\mathbf{A}} \in L^2(\Sigma) \right\}. \]
Then the operator $\mathfrak{T}^+_{K_0}$ extends as an isomorphism
\[ \mathfrak{T}^+_{K_0} : r^{-1} \dot{H}^1_C(\Sigma)^{\mathrm{curl}} \oplus r^{-1} L^2_C(\Sigma) \longrightarrow u^{-1} \dot{\mathcal{H}}^1(\scri^+), \]
where $u^{-1} \dot{\mathcal{H}}^1(\scri^+)$ is the space defined analogously to \Cref{def:characteristicdatapotential_Minkowski}, but with respect to the semi-norm $( \mathcal{E}^{K_0}_{\scri^+} )^{\nicefrac{1}{2}}$, and similarly for $\mathfrak{T}^-_{K_0}$. We then define the scattering operator associated to $K_0$ by
\[ \mathscr{S}_{K_0} \defeq \mathfrak{T}^+_{K_0} \circ (\mathfrak{T}^-_{K_0})^{-1} : v^{-1} \dot{\mathcal{H}}^1(\scri^-) \longrightarrow u^{-1} \dot{\mathcal{H}}^1(\scri^+), \]
which is again an isomorphism of Hilbert spaces.

\begin{remark} Notice that the space $r^{-1} \dot{H}^1_C(\Sigma)^{\mathrm{curl}} \oplus r^{-1} L^2_C(\Sigma)$ is a subspace of $\dot{H}^1_C(\Sigma) \oplus L^2_C(\Sigma)$, and that $u^{-1} \dot{\mathcal{H}}^1(\scri^+)$ is a subspace of $\dot{\mathcal{H}}^1(\scri^+)$. In other words, the vector field $K^a$ defines a weaker---more general---scattering theory between $\scri^-$ and $\scri^+$ than the vector field $K^a_0$. The construction for $K^a_0$ shows that the faster-decaying characteristic data on $\scri^-$ scatters to the correspondingly faster-decaying characteristic data on $\scri^+$. Indeed, in the checked conformal scale on $\scri^+$, the scattering operator $\mathscr{S}$ maps data that is $\check{F}^-_0 = \mathcal{O}(|v|^{-1})$ on $\scri^-$, through data that is $F_{0,1,2} = \mathcal{O}(r^{-2})$ on $\Sigma$, to data that is $\check{F}^+_2 = \mathcal{O}(|u|^{-1})$ on $\scri^+$,
\[ \mathscr{S} \, : \, \check{F}^-_0 = \mathcal{O}(|v|^{-1}) \overset{(\mathfrak{T}^-)^{-1}}{\longleadsto{0.8}} F_{0,1,2} = \mathcal{O}(r^{-2}) \overset{\mathfrak{T}^+}{\longleadsto{0.8}} \check{F}^+_2 = \mathcal{O}(|u|^{-1}). \]
Equivalently, in terms of the potential
\[ \mathscr{S} \, : \, \check{A}^-_2 = \mathcal{O}(\log |v|) \overset{(\mathfrak{T}^-)^{-1}}{\longleadsto{0.8}} \mathbf{A} = \mathcal{O}(r^{-1}), \, \dot{\mathbf{A}} = \mathcal{O}(r^{-2}) \overset{\mathfrak{T}^+}{\longleadsto{0.8}} \check{A}^+_2 = \mathcal{O}(\log |u|). \]
On the other hand,
\begin{align*} & \mathscr{S}_{K_0} \, : \, \check{F}^-_0 = \mathcal{O}(|v|^{-2}) \overset{\big(\mathfrak{T}_{K_0}^-\big)^{-1}}{\longleadsto{0.8}} F_{0,1,2} = \mathcal{O}(r^{-3}) \overset{\mathfrak{T}^+_{K_0}}{\longleadsto{0.8}} \check{F}^+_2 = \mathcal{O}(|u|^{-2}), \\
& \mathscr{S}_{K_0} \, : \,  \check{A}^-_2 = \mathcal{O}(|v|^{-1}) \overset{\big(\mathfrak{T}_{K_0}^-\big)^{-1}}{\longleadsto{0.8}} \mathbf{A} = \mathcal{O}(r^{-2}), ~ \dot{\mathbf{A}} = \mathcal{O}(r^{-3}) \overset{\mathfrak{T}^+_{K_0}}{\longleadsto{0.8}} \check{A}^+_2 = \mathcal{O}(|u|^{-1}).	
\end{align*}
\end{remark}

\section{Curved Spacetimes}
\label{sec:curvedspacetimes}

\subsection[Asymptotically Simple and CSCD Spacetimes]{Asymptotically Simple and Corvino--Schoen--Chru\'sciel--Delay \\ Spacetimes}

In this second part of the paper we work on spacetimes constructed by Chru\'sciel--Delay \cite{ChruscielDelay2002,ChruscielDelay2003}, Corvino \cite{Corvino2000}, and Corvino--Schoen \cite{CorvinoSchoen2006}. These are asymptotically flat, asymptotically simple spacetimes with null and timelike infinities of specifiable regularity, which are in addition diffeomorphic to the Schwarzschild or Kerr solution in a neighbourhood of spacelike infinity. These spacetimes are generically non-stationary and contain matter\footnote{We will impose a mild assumption on the decay of the matter fields towards null infinity, see \Cref{def:strong_AE}.}, and therefore the scattering processes on such spacetimes may be quite complex. As a consequence of their structure near spatial infinity, their conformal compactifications are also necessarily singular at $i^0$. Away from $i^0$, the constructions of Chru\'sciel--Delay \cite{ChruscielDelay2002,ChruscielDelay2003} permit spacetimes with a $\mathcal{C}^k$ conformal compactification for any finite $k$ (but not $\mathcal{C}^\infty$), and in what follows we shall simply assume that a sufficiently large order of differentiability $k$ has been chosen. We will refer to such $\mathcal{C}^k$ differentiability as \emph{smoothness}. 

We first recall the definition of asymptotically simple spacetimes \cite{AshtekarBongaKesavan2014,Ashtekar2014,Penrose1965,spinorsandspacetime1,spinorsandspacetime2}.

\begin{definition}[Asymptotically simple spacetimes] \label{def:asymptoticallysimple} Let $(\mathcalboondox{M},g)$ be a smooth globally hyperbolic spacetime. We say that  $(\mathcalboondox{M}, g)$ is \emph{asymptotically simple} if there exists another globally hyperbolic spacetime $(\hat{\mathcalboondox{M}}, \hat{g})$ such that 
	\begin{enumerate}
		\item the spacetime $\hat{\mathcalboondox{M}}$ is a manifold with boundary $\partial \hat{\mathcalboondox{M}} = \scri$, and $\hat{\mathcalboondox{M}} \setminus \scri$ is diffeomorphic to $\mathcalboondox{M}$,
		\item there exists a smooth function $\Omega$ on $\hat{\mathcalboondox{M}}$ such that $\hat{g}_{ab} = \Omega^2 g_{ab}$ and $\Omega > 0$ in $\mathcalboondox{M}$, $\Omega = 0$ on $\scri$, and $\d \Omega \neq 0$ on $\scri$, and
		\item every inextendible null geodesic in $\mathcalboondox{M}$ acquires two distinct endpoints on $\scri$.
	\end{enumerate}
\end{definition}

\noindent The condition $\d \Omega \neq 0$ on $\scri$ ensures that $\Omega$ can be used as a coordinate on $\hat{\mathcalboondox{M}}$ (at least in the neighbourhood of $\scri$), e.g. to perform Taylor expansions to capture the decay of fields near $\scri$. If $\mathcalboondox{M}$ happens to be vacuum (in fact it is enough that the trace of the matter stress-energy tensor vanishes near $\scri$) and the cosmological constant is zero, then as a hypersurface of the unphysical spacetime $\hat{\mathcalboondox{M}}$, $\scri$ is null. 

\begin{remark}
The definition above is the original definition of Penrose (see for instance Penrose and Rindler, Vol. 2 \cite{spinorsandspacetime2}, p. 351), in which the main point of interest was the construction of null infinity, or $\scri$, i.e. the set of end-points of inextendible null geodesics. That is, Penrose did not consider the endpoints of inextendible timelike or spacelike geodesics as part of $\scri$. For Corvino--Schoen--Chru\'sciel--Delay spacetimes, we can assume enough conformal regularity so that null infinity refocuses to a point in the future and a point in the past. These two points, referred to as future and past timelike infinities, can naturally be included in the boundary of the spacetime provided we choose a conformal factor with enough decay. 
\end{remark}

\begin{definition}[Corvino--Schoen--Chru\'sciel--Delay spacetimes] \label{def:CSCD_spacetimes}
The spacetimes of Corvino--Schoen--Chru\'sciel--Delay are asymptotically simple, and in addition to the conditions of \Cref{def:asymptoticallysimple} satisfy the following:
\begin{enumerate}
	\item[4.] the physical spacetime $(\mathcalboondox{M},g_{ab})$ satisfies Einstein's equations,
	\[ \mathrm{R}_{ab} - \frac{1}{2} \mathrm{R} g_{ab} = - 8 \pi \gamma \mathbf{T}_{ab}, \]
	where $\Omega^{-2} \mathbf{T}_{ab}$ has a smooth limit on $\scri$,
	\item[5.] the boundary of $\hat{\mathcalboondox{M}}$ is the union of two null hypersurfaces $\scri^+$ and $\scri^-$, referred to as future and past null infinities, and two points $i^+$ and $i^-$, referred to as future and past timelike infinities, such that the hypersurface $\scri^+$ is the past lightcone of $i^+$, and $\scri^-$ is the future lightcone of $i^-$,
	\item[6.] the metric $\hat{g}_{ab}$ is smooth at $i^\pm$ and $\scri^\pm$, and
	\item[7.] the physical spacetime $\mathcalboondox{M}$ is diffeomorphic to the Schwarzschild or Kerr solution outside the domain of influence of a given compact subset of a Cauchy surface $\Sigma$.
\end{enumerate}
\end{definition}
Note that the point $i^0$---spacelike infinity, the endpoint of all inextendible spacelike geodesics---is not part of the boundary of $\hat{\mathcalboondox{M}}$ when the ADM mass is non-zero, as it is a singularity of the conformal structure. Condition $4$ in the above---called the \emph{asymptotically vacuum} condition---ensures that the matter fields in the physical spacetime $\mathcalboondox{M}$ decay sufficiently fast at infinity to allow a sensible analysis of the geometry of $\scri$. The above definition abstracts the compactification procedure performed in \Cref{sec:partial_compactification_Minkowski}.

\begin{remark} \label{rmk:scri_Taylor_expansions} Note that the above definition implies that for two smooth (say $\mathcal{C}^k$ for $k\geq 3$) scalar fields $\alpha$ and $\beta$, the equality $\alpha \approx \beta$ implies that there exists a $\mathcal{C}^{k-1}$ scalar field $\gamma$ such that $\alpha = \beta + \Omega \gamma $.
\end{remark}

\subsection{The Schwarzschildean Neighbourhood of Spacelike Infinity}

The spacetimes of Corvino--Schoen--Chru\'sciel--Delay are diffeomorphic to the Schwarzschild or Kerr spacetime in a neighbourhood of $i^0$. For simplicity\footnote{Although the Kerr case is more cumbersome, the crucial fact that $\partial_t$ is Killing near $i^0$ remains true (see the estimates in \Cref{sec:Maxwell_energy_estimates_proof}). Therefore our scattering construction should in principle be extendible to the case of CSCD spacetimes diffeomorphic to Kerr near $i^0$.}, we consider the case of Schwarzschild. The metric near $i^0$ is then given by
\begin{equation} \label{Schwarzschildmetric} g_{ab} \, \d x^a \, \d x^b = F(r) \, \d t^2 - F(r)^{-1} \d r^2 - r^2 g_{\mathbb{S}^2}, \end{equation}
where $F(r) = 1 - 2 m r^{-1}$, with inverse metric
\[ g^{ab} \partial_a \odot \partial_b = F(r)^{-1} \partial_t^2 - F(r) \, \partial_r^2 - r^{-2} g_{\mathbb{S}^2}^{-1}. \]
The lapse $N$ here is therefore given by $N = \sqrt{1-2mr^{-1}}$, and it can be checked that the extrinsic curvature of the $\{ t = 0\}$ slice is zero, $\kappa_{ab} \equiv 0 $ (indeed, the Schwarzschild spacetime is static). We define the Eddington--Finkelstein coordinates
\[ u \defeq t -r_*, \qquad r_* \defeq r + 2 m \log \left( \frac{r}{2m} - 1 \right), \]
and the inverted radial coordinate
\[ R \defeq \frac{1}{r}. \]
The metric \eqref{Schwarzschildmetric} in the coordinates $(u,r, \theta, \phi)$ becomes
\[ g_{ab} \, \d x^a \, \d x^b = F(r) \, \d u^2+ 2 \, \d u \, \d r - r^2 g_{\mathbb{S}^2}, \]
with the inverse metric
\[ g^{ab} \partial_a \odot \partial_b = 2 \, \partial_u \odot \partial_r - F(r) \, \partial_r^2 - r^{-2} g_{\mathbb{S}^2}^{-1}. \]
The conformally rescaled metric is given by $\hat{g}_{ab} = \Omega^2 g_{ab}$, where we will wish to choose $\Omega$ carefully (in particular, we will need surfaces of $\Omega = \text{const.}$ to be null near $\scri^+$). Our conformal scale is described in detail in \Cref{sec:structure_of_scri}, though an explicit expression for the conformal factor even in the Schwarzschildean sector is not readily available (and will not be needed).

\subsection{Newman--Penrose Tetrads}

On the physical spacetime $\mathcalboondox{M}$ we define an NP tetrad $(l^a,m^a,\bar{m}^a,n^a)$ by aligning $l^a$ and $n^a$ with outgoing and incoming null congruences respectively such that wherever the metric $g_{ab}$ agrees with the Schwarzschild metric, the tetrad $(l^a,m^a,\bar{m}^a, n^a)$ takes the concrete form
\begin{align} 
\label{Schwarzschildtetradl}	 & n^a = \partial_u - \frac{1}{2} F(r) \partial_r, && n_a = \frac{1}{2} F(r) \, \d u + \d r, \\
\label{Schwarzschildtetradn}	 & l^a = \partial_r, && l_a = \d u, \\
\label{Schwarzschildtetradm}	 & m^a = \frac{1}{\sqrt{2}r} \left( \partial_\theta + \frac{i}{\sin \theta} \partial_\phi \right), && m_a = - \frac{r}{\sqrt{2}} ( \d \theta + i \sin \theta \, \d \phi ), \\
\label{Schwarzschildtetradmbar} & \bar{m}^a = \frac{1}{\sqrt{2}r} \left( \partial_\theta - \frac{i}{\sin \theta} \partial_\phi \right), && \bar{m}_a = -\frac{r}{\sqrt{2}}( \d \theta - i \sin \theta \, \d \phi ),
\end{align}
 and extending $\sin \theta \, m^a$ as a $\mathcal{C}^k$-smooth (to avoid the singularity on the sphere) complex null vector everywhere orthogonal to $l^a$ and $n^a$. We assume that the vector fields $l^a$ and $n^a$ are $\mathcal{C}^k$-smooth and real. We obtain a rescaled NP tetrad on $\hat{\mathcalboondox{M}}$ by the rescaling \Cref{tetradrescalingl,tetradrescalingn,tetradrescalingm,tetradrescalingmbar}. We assume that $\hat{n}^a$ restricted to $\scri^+$ is a generator of $\scri^+$ and that $\hat{\nabla}^a \Omega$ is proportional to $\hat{n}^a$ on $\scri^+$. We also assume that the vector fields (and their corresponding $1$-forms) $(\hat{l}^a, \hat{m}^a, \bar{\hat{m}}^a, \hat{n}^a)$ are all $\mathcal{C}^k$-smooth throughout $\hat{\mathcalboondox{M}}$ (modulo the usual singularity on the spheres). We use $\hat{l}_a$ to define a $3$-volume form on $\scri^+$: 
\[ \widehat{\dvol}_{\scri^+} \defeq \hat{l}^\flat \wedge (i \hat{m}^\flat \wedge \bar{\hat{m}}^\flat). \]
Finally, we assume that the future-oriented unit normal $T^a$ to the hypersurfaces $\Sigma_t$ of constant time coordinate $t$ which approach $\scri^+$ as $t \to \infty$ is independent of the angular vector fields (which is an assumption on the null tetrad being adapted to the foliation). Then the normalisation $g_{ab} T^a T^b = 1$ implies that $T^a$ is given by
\begin{equation} \label{radialT} T^a = a n^a + \frac{1}{2a} l^a	
\end{equation}
for some positive function $a$ on $\mathcalboondox{M}$ which extends smoothly to $\scri^+$ and does not vanish there. Since $T^a$ should be invariant under rescalings of the NP tetrad, the function $a$ is a $\{1,1\}$-scalar. In terms of the rescaled tetrad $T^a$ is then given by
\[ T^a = a \hat{n}^a + \frac{\Omega^2}{2a} \hat{l}^a, \]
and becomes proportional to the generator $\hat{n}^a$ of $\scri^+$ on $\scri^+$. In the Schwarzschild sector we explicitly have that the unit normal to surfaces $\Sigma_t$ of constant $t$ is $T^a = F(r)^{-1/2} \partial_t$, which in terms of the physical NP tetrad is given by
\[ T^a = F(r)^{-1/2} n^a + \frac{1}{2} F(r)^{1/2} l^a. \]
Therefore here $a = F(r)^{-1/2} \approx 1$.

\subsection{Structure of \texorpdfstring{$\scri$}{scri}} \label{sec:structure_of_scri}

The topological structure of null infinity of all asymptotically flat asymptotically simple spacetimes is the same, and essentially identical to the topology of null infinity of Minkowski space \cite{Penrose1965,spinorsandspacetime2}. Indeed, one has the following theorem.

\begin{theorem}[Penrose, 1965] \label{thm:topology_of_scri} In any asymptotically simple spacetime $\mathcalboondox{M}$ for which $\scri$ is everywhere null, the topology of each of $\scri^\pm$ is given by
\[ \scri^+ \simeq \scri^- \simeq \mathbb{R} \times \mathbb{S}^2, \]
and the rays generating $\scri^\pm$ can be taken to be the $\mathbb{R}$ factors.
\end{theorem}

\noindent For these spacetimes future (or past) null infinity $\scri^+$ is therefore a null $3$-dimensional manifold ruled by the integral curves of $\hat{n}^a \propto \hat{\nabla}^a \Omega$. The pullback $\hat{q}_{ab}$ to $\scri^+$ of the metric $\hat{g}_{ab}$ gives a degenerate metric on $\scri^+$ which has signature $(0,-,-)$. Moreover, there is still considerable conformal freedom\footnote{If this conformal freedom is employed to choose a scale in which $\check{\nabla}_a \check{n}^a = 0$, the equivalence classes of pairs $(\check{q}_{ab}, \check{n}^a)$ related by the \emph{remaining} conformal freedom (that is, the conformal rescalings of the $2$-spheres) forms the so-called \emph{universal structure} of $\scri^+$ \cite{Ashtekar2014}.} on $\scri^+$ once $\Omega$ has been chosen to bring $\scri^+$ to a finite distance: if the physical metric $g_{ab}$ was related to the rescaled metric by $\hat{g}_{ab} = \Omega^2 g_{ab}$, then for any $\omega$ which is smooth and nowhere vanishing on $\scri^+$ the rescaling $\Omega \leadsto \omega \Omega$ is still permissible, giving $\check{g}_{ab} = \omega^2 \hat{g}_{ab}$. Here, by conformally scaling the $2$-spheres so that their metric is that of a geometric unit $2$-sphere, we may choose the conformal scale so that (minus) the induced metric on $\scri^+$ becomes
\begin{equation} \label{metricScri+}
\d l^2 = 0 \cdot \d u^2 + g_{\mathbb{S}^2}
\end{equation}
for a coordinate $u$ with range $u \in \mathbb{R}$ satisfying $-\hat{\nabla}^a \Omega \hat{\nabla}_a u \approx 1$ and $\hat{\eth} u \approx 0$. In this scale the generators of $\scri^+$ therefore map its cross-sections to one another isometrically, making $u$ a Bondi retarded time coordinate on $\scri^+$ (cf. \cite{spinorsandspacetime2}, (9.8.31)). In this scale $i^+$ (and $i^0$) are at infinity.

\subsubsection{Spin and curvature coefficients on $\scri^+$}

Further information about the structure of $\scri^+$ is provided by the fact that the physical spacetime $\mathcalboondox{M}$ is assumed to be asymptotically vacuum in the sense of point 4 in \Cref{def:CSCD_spacetimes}. The trace-free part of the Ricci tensor
\[ \Phi_{ab} \defeq -\frac{1}{2} \left( \mathrm{R}_{ab} - \frac{1}{4} \mathrm{R} g_{ab} \right) \]
transforms under a conformal rescaling as
\begin{align*} \Phi_{ab} &= \hat{\Phi}_{ab} + \hat{\nabla}_a \Upsilon_b - \frac{1}{4} \hat{g}_{ab} \hat{\nabla}^c \Upsilon_c + \Upsilon_a \Upsilon_b - \frac{1}{4} \hat{g}_{ab} \hat{g}^{cd} \Upsilon_c \Upsilon_d \\
&= \hat{\Phi}_{ab} + \Omega^{-1} \hat{\nabla}_a \hat{\nabla}_b \Omega - \frac{1}{4} \Omega^{-1} \hat{g}_{ab} \hat{\nabla}^c \hat{\nabla}_c \Omega.	
\end{align*}
One has, according to \Cref{def:asymptoticallysimple}, that $\Omega^{-2} \mathrm{R}_{ab}$ has a continuous limit on $\scri^+$, so multiplying the above by $\Omega$ and taking the limit $\Omega \to 0$ ensures that $\Omega \Phi_{ab} \approx 0$, and gives the \emph{asymptotic Einstein condition}
\begin{equation} \label{asymptoticEinstein} \hat{\nabla}_a \hat{\nabla}_b \Omega \approx \frac{1}{4} \hat{g}_{ab} \hat{\nabla}^c \hat{\nabla}_c \Omega.
\end{equation}
The normal $\hat{\nabla}^b \Omega$ to $\scri^+$ is proportional to $\hat{n}^a$, $\hat{\nabla}^b \Omega \approx f \hat{n}^b$ for some non-vanishing scalar function $f$, so the condition \eqref{asymptoticEinstein} reads
\begin{equation} \label{asEinstein2} f \hat{\nabla}_a \hat{n}_b + \hat{n}_b \hat{\nabla}_a f \approx \frac{1}{4} \hat{g}_{ab} ( f \hat{\nabla}_c \hat{n}^c + \hat{\Delta} f ). \end{equation}
Multiplying by $\hat{n}_c$ and antisymmetrizing shows that (see \cite{spinorsandspacetime2}, (7.1.58))
\begin{equation} \label{geodetictwistfreescri} \hat{n}_{[a} \hat{\nabla}_b \hat{n}_{c]} \approx 0 \iff (\hat{\nu} \approx 0 ,\quad \hat{\mu} \approx \bar{\hat{\mu}} ),
\end{equation}
where the conditions on the spin coefficients $\hat{\nu}$ and $\hat{\mu}$ may be rapidly obtained from the hypersurface orthogonal condition by contracting with $\hat{n}^a \hat{m}^b$ and $\hat{m}^{[a} \bar{\hat{m}}^{b]}$ respectively. The vanishing of the spin coefficient $\hat{\nu}$ on $\scri^+$ tells us that $\scri^+$ is generated by null geodesics, whereas the condition $\hat{\mu} \approx \bar{\hat{\mu}}$ says that the vectors $\hat{n}^a$ are \emph{twist-free} on $\scri^+$. Contracting \eqref{asEinstein2} with $\hat{m}^a \hat{m}^b$, we also get 
\begin{equation} \label{shearfreescri} \hat{\lambda} \approx 0,	
\end{equation}
which is the statement that the vectors $\hat{n}^a$ are \emph{shear-free} on $\scri^+$. We say the hypersurface $\scri^+$ is geodetic, twist-free and shear-free. Since the vectors $\hat{n}^a$ are geodetic on $\scri^+$, they are parallely propagated, $\hat{\Delta} \hat{n}^a = \hat{\nabla}_{\hat{n}} \hat{n}^a \approx s \hat{n}^a$ for some function $s$, which vanishes identically if the geodesics are affinely parametrized. Contracting with $\hat{l}^a$, one sees that the function $s$ is given by $s = \hat{l}^a \hat{\Delta} \hat{n}_a$. This is in fact the real part of another spin coefficient, $-(\hat{\gamma} + \bar{\hat{\gamma}}) = \hat{l}^a \hat{\Delta} \hat{n}_a$, so the condition for the geodesics generated by $\hat{n}^a$ on $\scri^+$ to be affinely parametrized is $\hat{\gamma} + \bar{\hat{\gamma}} \approx 0$. It is always possible to reparametrize a geodesic affinely, and here we will assume that the original parametrization has been made to that effect. The condition for the imaginary part of $\hat{\gamma}$ to vanish, $\hat{\gamma} - \bar{\hat{\gamma}} = 0$, can be translated as the statement that the spinor field $\hat{\iota}^A$ has parallelly propagated flag planes, where $\hat{n}^a = \hat{\iota}^A \hat{\iota}^{A'}$. If $\scri^+$ is affinely parametrized with parallelly propagated flag planes, then $\hat{\gamma} \approx 0$.

We next make a further specialization of our choice of $\Omega$ so that near $\scri^+$ surfaces of $\Omega = \text{const}.$ are null, i.e. $\hat{\nabla}^a \Omega = f \hat{n}^a$ near $\scri^+$, not just on $\scri^+$ (cf. \cite{spinorsandspacetime2}, (9.8.29)). Then $\hat{\Delta} \Omega = 0 = \hat{\delta} \Omega$ near $\scri^+$, and the remaining components of \eqref{asEinstein2} imply that
\begin{equation} \label{isometric_cross_sections_scri}
\hat{\mu} \approx 0 \approx \hat{\pi}
\end{equation}
and 
\[ \hat{\thorn} f \approx 0. \]
In fact, $\hat{\lambda} \approx 0 \approx \hat{\mu}$ may have been deduced directly from the form of the metric \eqref{metricScri+}. With these conditions \eqref{asEinstein2} further implies
\begin{equation} \label{f_constant_on_scri} \hat{\thorn}' f \approx 0 \approx \hat{\eth} f,
\end{equation}
and, as in \eqref{geodetictwistfreescri}, we now have
\[ \hat{\nu} = 0 \quad \text{and} \quad \hat{\mu} = \bar{\hat{\mu}} \]
\emph{near} $\scri^+$, not just on $\scri^+$. Now the condition $-\hat{\nabla}^a \Omega \hat{\nabla}_a u \approx 1$ may be rewritten as $\hat{\thorn}' u \approx - f^{-1}$; recalling that $\hat{\eth} u = 0$ and commuting $\hat{\thorn}'$ into this equation, one finally derives
\[ \hat{\tau} \approx 0. \]
This is the condition for the choice of parameter $u$ and the scaling near $\scri^+$ to be a so-called \emph{Bondi system}. We therefore have the following.

\begin{proposition} 
\label{prop:spin_coefficients_near_scri}
On any given Corvino--Schoen--Chru\'sciel--Delay spacetime $\mathcalboondox{M}$ there exists a conformal scale, a choice of NP tetrad $(\hat{l}^a, \hat{m}^a, \bar{\hat{m}}^a, \hat{n}^a)$, and a choice of Bondi time coordinate $u$ such that the metric on $\scri^+$ is given by \eqref{metricScri+},
\[ \hat{\lambda} \approx \hat{\pi} \approx \hat{\mu} \approx \hat{\tau} \approx \hat{\gamma} \approx 0, \]
and
\[ \hat{\nu} = 0 \quad \text{and} \quad \hat{\mu} = \bar{\hat{\mu}} \]
in a neighbourhood of $\Omega = 0$.
\end{proposition}

\noindent In addition to the asymptotic Einstein condition \eqref{asymptoticEinstein}, we further assume that our spacetime satisfies the so-called \emph{strong asymptotic Einstein condition}.

\begin{definition}[Strong asymptotic Einstein condition] \label{def:strong_AE}
A Corvino--Schoen--Chru\'sciel--Delay spacetime $\mathcalboondox{M}$ is said to satisfy the \emph{strong asymptotic Einstein condition} if it satisfies \eqref{asymptoticEinstein}, and
\begin{equation}
\label{strong_asymptotic_Einstein}
\hat{\Psi}_0 \approx \hat{\Psi}_1 \approx \hat{\Psi}_2 \approx \hat{\Psi}_3 \approx \hat{\Psi}_4 \approx 0,
\end{equation}
the $\hat{\Psi}_i$'s being the components of the (rescaled) Weyl tensor. 
\end{definition}

\begin{remark}
    The strong asymptotic Einstein condition holds if the physical spacetime satisfies $\mathrm{R}_{ab} \propto g_{ab}$ near $\scri^+$ (\cite{spinorsandspacetime2}, (9.6.32)). In particular, all vacuum CSCD spacetimes satisfy the condition, of which there are an infinite-dimensional family \cite{ChruscielDelay2002}.
\end{remark}

\begin{proposition} In the setting of \Cref{prop:spin_coefficients_near_scri}, the strong asymptotic Einstein condition \eqref{strong_asymptotic_Einstein} further implies
\[ \hat{\Phi}_{22} \approx 0 \approx \hat{\Phi}_{21}. \]
\end{proposition}

\begin{proof} This follows straightforwardly from the curvature equations
\[ \hat{\thorn}' \hat{\mu} - \hat{\eth} \hat{\nu} = - \hat{\mu}^2 - |\hat{\lambda}|^2 + \bar{\hat{\nu}} \hat{\pi} - \hat{\nu}\hat{\tau} -\hat{\Phi}_{22} \]
and
\[ \bar{\hat{\eth}} \hat{\mu} - \hat{\eth} \hat{\lambda} = \hat{\pi} ( - \hat{\mu} + \bar{\hat{\mu}}) + \hat{\nu}(\bar{\hat{\rho}} - \hat{\rho}) + \hat{\Psi}_3 - \hat{\Phi}_{21}. \]
\end{proof}
The conformal factor $\Omega$ that we have specified in this and the previous section is the analogue, for the spacetimes of Corvino--Schoen--Chru\'sciel--Delay, of $1/r$ on Minkowski spacetime.

\subsection{Construction of Gauge}

In order to recover the main aspects of the scattering construction on curved spacetimes, we must choose an appropriate gauge in which $\hat{A}_1 \approx 0$. In the case of Minkowski space, this was achieved by the temporal gauge, and then there turned out to exist a suitable second-order reduction of the Coulomb gauge which allowed us to recover $\hat{A}_0$ on $\scri^+$, and which made the equations non-singular up to $\scri^+$. This construction cannot be carried over, however, as in a generic curved spacetime of Corvino--Schoen--Chru\'sciel--Delay type if one imposes the Coulomb gauge $\boldsymbol{\nabla} \cdot \mathbf{A} = 0$ on the slices $(\Sigma_t, h_{ab})$ with normal $T^a = a n^a + \frac{1}{2a} l^a$, the component $\mathfrak{a} = T^a A_a$ no longer satisfies an unsourced elliptic equation. Instead, $\mathfrak{a}$ satisfies an equation of the form
\[ \boldsymbol{\Delta} \mathfrak{a} = \kappa \cdot f_0 + (\nabla \kappa ) \cdot f_1 \]
for sources $f_0$ and $f_1$. The presence of the extrinsic curvature $\kappa$ of $\Sigma_t$ therefore generically prevents $\mathfrak{a}$ from being zero, making the Coulomb and temporal gauges incompatible. 

Choosing an appropriate gauge is therefore a non-trivial problem. At the outset, one has two distinct classes of gauge conditions to consider: those defined in the physical spacetime $\mathcalboondox{M}$, and those defined in the rescaled spacetime $\hat{\mathcalboondox{M}}$. On the rescaled spacetime, of the common gauge fixing conditions (temporal, Coulomb and Lorenz), none give any useful information on $\scri^+$: the temporal gauge only relates two components of $\hat{A}_a$ in $\hat{\mathcalboondox{M}}$ but is otherwise severely incomplete (the rescaled field equations for $\hat{A}_a$ are not hyperbolic), the Coulomb gauge with respect to any foliation which intersects $\scri^+$ transversely is clearly not adapted to the problem, and the Lorenz gauge produces a PDE on $\scri^+$ which involves transverse derivatives, and is therefore not intrinsically solvable at finite energy regularity\footnote{While it is true that in principle transverse derivatives on $\scri^+$ of solutions to the wave equation are expressible as integrals along the null generators, this requires the data to have at least two derivatives in $L^2_{\text{loc}}(\scri^+)$, which we do not assume.}. One is therefore naturally led to consider imposing a gauge condition in the physical spacetime $\mathcalboondox{M}$. It turns out that the physical Lorenz gauge and the physical temporal gauge both reduce to $\hat{A}_1 \approx 0$ on $\scri^+$, whereas the physical Coulomb gauge reduces to the slightly weaker condition $\hat{\thorn}'(a \hat{A}_1) \approx 0$. Of course, the temporal gauge still suffers from the fact that it is an incomplete gauge fixing condition. The Coulomb gauge with respect to the asymptotically null foliation $\Sigma_t$ turns out to have a potentially useful (but messy) expansion in powers of $\Omega$ near $\scri^+$, but at second order---where we would expect to find the equation for $\hat{A}_0$---happens to contain a transversal derivative of $\hat{A}_1$ which is problematic to deal with. It turns out that a certain combination of the three is needed. We will impose the physical Lorenz gauge throughout $\mathcalboondox{M}$, and subsequently use the residual gauge freedom to fix $\mathfrak{a} = 0 = \boldsymbol{\nabla} \cdot \mathbf{A}$ on $\Sigma$, and impose the condition $\hat{A}^{[1]}_1 \defeq \Omega^{-1} \hat{A}_1 \approx 0$ on $\scri^+$. Unfortunately, the residual gauge transformation needed to set $\mathfrak{a} = 0 = \boldsymbol{\nabla} \cdot \mathbf{A} $ on $\Sigma$ may in general be incompatible with the one needed to set $\hat{A}_1^{[1]} \approx 0$ on $\scri^+$. Our gauge will therefore break the Lorenz gauge condition in the interior of $\mathcalboondox{M}$, away from a neighbourhood of $\Sigma$ and away from a neighbourhood of $\scri^+$. We describe the construction in detail below.

\subsubsection{Condition on \texorpdfstring{$\scri^+$}{$\scri^+$}} \label{sec:gauge_on_scri}

Suppose for the moment that we have a smooth solution $\hat{A}_a$ on $\hat{\mathcalboondox{M}}$ which extends smoothly to $\scri^+$. Note that, by the smoothness of $\hat{\mathcalboondox{M}}$, for any scalar field $q$ which vanishes on $\scri^+$, $q \approx 0$, there exists another scalar field $q^{[1]}$ which extends smoothly to $\scri^+$ and which satisfies $q = \Omega q^{[1]}$ (\Cref{rmk:scri_Taylor_expansions}). In particular, the spin coefficients $\hat{\lambda}$, $\hat{\pi}$, $\hat{\mu}$, $\hat{\tau}$, $\hat{\gamma}$ are $\mathcal{O}(\Omega)$ near $\scri^+$  (and $\hat{\nu} \equiv 0$ near $\scri^+$). Using that $\hat{\nabla}^a \Omega = f \hat{n}^a$ near $\scri^+$, we compute
\begin{align}
\label{physical_Lorenz_expansion}
\begin{split}
\Omega^{-2} \nabla_a A^a &= \hat{\nabla}_a \hat{A}^a - 2 \Upsilon_a \hat{A}^a \\
& = - 2 \Omega^{-1} f \hat{A}_1 +\hat{\thorn} \hat{A}_1 - 2 \hat{A}_1 \operatorname{Re}\hat{\rho} + \hat{\thorn}' \hat{A}_0 - 2 \operatorname{Re}(\hat{\eth} \bar{\hat{A}}_2) + \mathcal{O}(\Omega).
\end{split}
\end{align}
Imposing
\[ \nabla_a A^a \equiv 0 \]
throughout $\mathcalboondox{M}$, the leading order $\mathcal{O}(\Omega^{-1})$ in \eqref{physical_Lorenz_expansion} implies that in the limit $\Omega \to 0$
\begin{equation}
\label{Lorenz_gauge_first_order_reduction}
\hat{A}_1 \approx 0.
\end{equation}
Writing $\hat{A}_1 = \Omega \hat{A}_1^{[1]}$, we then may rewrite \eqref{physical_Lorenz_expansion} as
\[ -f \hat{A}_1^{[1]} + \hat{\thorn}' \hat{A}_0 - 2 \operatorname{Re}(\hat{\eth} \bar{\hat{A}}_2) + \mathcal{O}(\Omega) \equiv 0, \]
which becomes
\begin{equation}
\label{pre_gauged_second_order_gauge_reduction}
-f \hat{A}_1^{[1]} + \hat{\thorn}' \hat{A}_0 - 2 \operatorname{Re}(\hat{\eth} \bar{\hat{A}}_2) \approx 0
\end{equation}
in the limit $\Omega \to 0$. This is nearly the second order gauge reduction that we seek, with the exception of the term $-f \hat{A}_1^{[1]}$. Now the residual gauge freedom in the physical Lorenz gauge $\nabla_a A^a = 0$ is $A_a \rightsquigarrow A_a + \nabla_a \chi_{\text{res.}}$ for any $\chi_{\text{res.}}$ such that
\begin{equation}
\label{residual_gauge_freedom_wave_equation_physical}
\Box \chi_{\text{res.}} = 0 \quad \text{on } \mathcalboondox{M}.
\end{equation}
A direct rewriting of this equation in terms of rescaled quantities gives
\begin{equation} 
\label{residual_gauge_freedom_wave_equation}
\widehat{\Box} \hat{\chi}_{\text{res.}} + \frac{1}{6} \big( \hat{\mathrm{R}} - \mathrm{R} \Omega^{-2} \big) \hat{\chi}_{\text{res.}} = 0,
\end{equation}
where $\hat{\chi}_{\text{res.}} = \Omega^{-1} \chi_{\text{res.}}$. We have the following.

\begin{lemma} The equation \eqref{residual_gauge_freedom_wave_equation} for the residual gauge transformation $\chi_{\text{res.}}$ is non-singular up to $\scri^+$, and in fact in our conformal scale (\Cref{prop:spin_coefficients_near_scri}) reads
\begin{equation}
\label{residual_gauge_freedom_wave_equation_near_scri}
\widehat{\Box} \hat{\chi}_{\text{res.}} + 2 \hat{\mu}^{[1]} f \hat{\chi}_{\text{res.}} = 0
\end{equation}
near $\scri^+$, where $\hat{\mu}^{[1]} = \Omega^{-1} \hat{\mu}$.
\end{lemma}

\begin{proof} The fact that the quantity $\hat{\mathrm{R}} - \mathrm{R} \Omega^{-2}$ is non-singular up to $\scri^+$ may be read off directly from the asymptotic Einstein condition \eqref{asymptoticEinstein}. More concretely, a calculation using $\hat{\nabla}^a \Omega = f \hat{n}^a$ shows
\begin{align*}
\frac{1}{6}\big( \hat{\mathrm{R}} - \mathrm{R} \Omega^{-2} \big) &= \hat{\nabla}^a \Upsilon_a - \hat{g}^{ab} \Upsilon_a \Upsilon_b = 2 \hat{\mu} \hat{D} \log \Omega = 2 \hat{\mu}^{[1]} f .
\end{align*}
On $\scri^+$, $\mathrm{R} \Omega^{-2}$ in fact tends to zero by the asymptotic Einstein condition, so $2 \hat{\Lambda} = \frac{1}{12} \hat{\mathrm{R}} \approx \hat{\mu}^{[1]} f$. Since in our conformal scale the metric on $\scri^+$ is given by \eqref{metricScri+}, one also has $\hat{\Phi}_{11} + \hat{\Lambda} \approx \frac{1}{2}$ (this is the statement that $\hat{\Phi}_{11} + \hat{\Lambda}$ is one half of the Gaussian curvature of the unit $2$-sphere on $\scri^+$, see \cite{spinorsandspacetime2}, (9.8.33)). Altogether therefore $\hat{\mu}^{[1]} f \approx 1 - 2 \hat{\Phi}_{11} \approx 2\hat{\Lambda}$. This also shows that generically $\hat{\mu}$ only vanishes to first order on $\scri^+$, unless $\hat{\Phi}_{11} \approx \frac{1}{2} \iff \hat{\Lambda} \approx 0$.
\end{proof}

\begin{proposition} \label{prop:residual_gauge_scri}
In the physical Lorenz gauge we may perform a residual gauge transformation near $\scri^+$ to set 
\[ \hat{A}_1^{[1]} \approx 0. \]
\end{proposition}

\begin{proof} Here we assume\footnote{For a solution arising from smooth compactly supported initial data on $\Sigma$, this is always true; see e.g. \cite{Gajic2022}.} that $\hat{A}_a$ is supported away from $i^0$ and that $\hat{A}_1^{[1]} \to 0$ at $i^+$ in the conformal scale of \Cref{prop:spin_coefficients_near_scri} (that is, that the physical component $A_1$ vanishes to second order at $i^+$). By density, we can therefore assume that $\hat{A}_1^{[1]}|_{\scri^+}$ is compactly supported. Clearly there is nothing to be done outside of the support of $\hat{A}_1^{[1]}|_{\scri^+}$. Now in the neighbourhood of $\scri^+$ where $\hat{\Delta} \Omega = 0 = \hat{\delta} \Omega$, a residual gauge transformation sets
\[ \hat{A}_1 \rightsquigarrow \hat{A}_1 + \hat{\thorn}'( \Omega \hat{\chi}_{\text{res.}} ) = \Omega( \hat{A}_1^{[1]} + \hat{\thorn}' \hat{\chi}_{\text{res.}} ),  \]
\[ \hat{A}_0 \rightsquigarrow \hat{A}_0 + \hat{D} (\Omega \hat{\chi}_{\text{res.}} ) = \hat{A}_0 + f \hat{\chi}_{\text{res.}} + \Omega \hat{D} \hat{\chi}_{\text{res.}}, \]
and
\[ \hat{A}_2 \rightsquigarrow \hat{A}_2 + \Omega \hat{\delta} \hat{\chi}_{\text{res.}}. \]
This gives
\[ \hat{\thorn}' \hat{A}_0 \rightsquigarrow \hat{\thorn}' \hat{A}_0 + f \hat{\thorn}' \hat{\chi}_{\text{res.}} + \hat{\chi}_{\text{res.}} \hat{\thorn}' f + \Omega \hat{\thorn}' \hat{D} \hat{\chi}_{\text{res.}}, \]
so that, using \eqref{f_constant_on_scri}, one sees that \eqref{Lorenz_gauge_first_order_reduction} and \eqref{pre_gauged_second_order_gauge_reduction} are residual-gauge-invariant on $\scri^+$, and $\hat{A}_1^{[1]}$ is transformed according to $\hat{A}_1^{[1]} \rightsquigarrow \hat{A}_1^{[1]} + \hat{\thorn}' \hat{\chi}_{\text{res.}}$. In Lorenz gauge, we therefore put
\begin{equation} \label{residual_gauge_data_scri} \hat{\chi}_{\text{res.}}^+ \defeq - \int_{-\infty}^u \hat{A}_1^{[1]} \, \d u
\end{equation}
on $\scri^+$, which has the effect of setting $\hat{A}_1^{[1]} \approx 0$ in the new gauge. It remains to show that we can solve \eqref{residual_gauge_freedom_wave_equation_near_scri} for $\hat{\chi}_{\text{res.}}$ with this data. Introduce a short outgoing null hypersurface $\mathcal{H}$ which intersects $\scri^+$ in the future of the support of $\hat{A}_1^{[1]}|_{\scri^+}$, and prescribe constant-in-$v$ data for $\hat{\chi}_{\text{res.}}$ on $\mathcal{H}$ (the function on the intersection sphere chosen in such a way that it matches the values of $\hat{\chi}^+_{\text{res.}}$ on $\mathcal{H} \cap \scri^+$). Then $\hat{\chi}_{\text{res.}}^+$, as defined in \eqref{residual_gauge_data_scri}, is $H^1_c$ on the union of $\mathcal{H}$ and the part of $\scri^+$ in the past of $\mathcal{H}$, so we may apply \Cref{thm:Goursatproblem} to solve \Cref{residual_gauge_freedom_wave_equation_near_scri} for $\hat{\chi}_{\text{res.}}$ in a neighbourhood of $\scri^+$.
\end{proof}

\begin{remark} By the smoothness of the spacetime, the condition $\hat{A}^{[1]}_1 \approx 0$ implies that $\hat{A}_1^{[1]} = \mathcal{O}(\Omega)$, and therefore $\hat{A}_1 = \mathcal{O}(\Omega^2)$ near $\scri^+$. In this sense this residual gauge condition is reminiscent of the temporal gauge near $\scri^+$.
\end{remark}

\vspace{1em}
Imposing this residual gauge condition on $\scri^+$, we obtain from \eqref{pre_gauged_second_order_gauge_reduction} the second order gauge reduction on $\scri^+$
\begin{equation}
\label{second_order_gauge_reduction_scri}
\partial_u \hat{A}_0 \approx 2 \operatorname{Re}(\hat{\eth} \bar{\hat{A}}_2 ).
\end{equation}
Finally, we show that the rescaled field equations are non-singular up to $\scri^+$. In the physical Lorenz gauge the field equations \eqref{Maxwellsequationspotential} (cf. \eqref{Maxwellsequationspotentialhatted}) on $\hat{\mathcalboondox{M}}$ read
\begin{equation}
\label{Maxwells_equations_rescaled_physical_Lorenz}
\widehat{\Box} \hat{A}_a - \hat{\nabla}_a ( 2 \Upsilon_a \hat{A}^a ) + \hat{\mathrm{R}}_{ab} \hat{A}^b = 0,
\end{equation}
so it suffices to show that the quantity $\Upsilon_a \hat{A}^a$ is regular near $\Omega = 0$; but in our conformal scale $\Upsilon_a \hat{A}^a = f \hat{A}_1^{[1]}$ near $\scri^+$, which has a continuous (in fact vanishing) limit on $\scri^+$. Note also that \eqref{Maxwells_equations_rescaled_physical_Lorenz} are a system of linear wave equations in this gauge.

\subsubsection{Condition on \texorpdfstring{$\Sigma$}{$\Sigma$}}

On the initial surface $\Sigma$, we will need the conditions $\mathfrak{a} = 0 = \boldsymbol{\nabla} \cdot \mathbf{A} $ in order to define function spaces of initial data for the potential. 

\begin{proposition} \label{prop:residual_gauge_initial_surface}
In a neighbourhood of the initial surface $\Sigma$ we may perform a residual gauge transformation in the physical Lorenz gauge to set 
\[ \mathfrak{a}|_\Sigma = 0 = \boldsymbol{\nabla} \cdot \mathbf{A}|_\Sigma. \]
\end{proposition}

\begin{proof} Suppose we have a smooth solution $A_a$ in a neighbourhood of $\Sigma$. The residual gauge freedom is \eqref{residual_gauge_freedom_wave_equation}, so that on $\Sigma$ we may freely prescribe $\chi_{\text{res.}}$ and $\nabla_T \chi_{\text{res.}} = \frac{1}{N} \dot{\chi}_{\text{res.}}$. Performing a residual gauge transformation,
\[ \mathfrak{a} \rightsquigarrow \mathfrak{a} + \frac{1}{N} \dot{\chi}_{\text{res.}}, \]
so we simply set $\dot{\chi}_{\text{res.}}|_\Sigma = - N \mathfrak{a}|_\Sigma$. Also,
\[ \boldsymbol{\nabla} \cdot \mathbf{A} \rightsquigarrow \boldsymbol{\nabla} \cdot \mathbf{A} + \boldsymbol{\Delta} \chi_{\text{res.}},  \]
so for $\chi_{\text{res.}}$ we set
\[ \chi_{\text{res.}}|_\Sigma = \boldsymbol{\Delta}^{-1} (-\boldsymbol{\nabla}\cdot \mathbf{A} |_\Sigma), \]
with the boundary condition that $\chi_{\text{res.}} \to 0$ at $i^0$ (the existence of such a $\chi_{\text{res.}}$ is provided by the Lax--Milgram Lemma). We then propagate $\chi_{\text{res.}}$ a short time off $\Sigma$ according to \eqref{residual_gauge_freedom_wave_equation} to obtain the gauge near the initial surface.
\end{proof}

Altogether, we therefore impose $\nabla_a A^a \equiv 0$ throughout $\mathcalboondox{M}$, which directly leads to the condition $\hat{A}_1 \approx 0$. Using \Cref{prop:residual_gauge_initial_surface}, we obtain a $\chi^0_{\text{res.}}$ in a neighbourhood $\mathcal{O}^0$ of $\Sigma$ which sets $\mathfrak{a}|_\Sigma = 0 = \boldsymbol{\nabla} \cdot \mathbf{A}|_\Sigma$, and, using \Cref{prop:residual_gauge_scri}, we obtain a $\hat{\chi}_{\text{res.}}^1$ in a neighbourhood $\mathcal{O}^1$ of $\scri^+$ which sets $\hat{A}_1^{[1]} \approx 0$. In general there is no reason for $\chi^0_{\text{res.}}$ to be equal to $\Omega \hat{\chi}_{\text{res.}}^1$, so we interpolate smoothly between the two in the region between $\mathcal{O}^0$ and $\mathcal{O}^1$. This procedure will break $\Box \chi_{\text{res.}} = 0$ in the interpolation region, and hence we will no longer satisfy the Lorenz gauge there. Therefore in this region we will work with the Maxwell field $F_{ab}$. This will present no difficulties as we will simply need to solve a regular Cauchy problem a finite time into the future (or past) here.

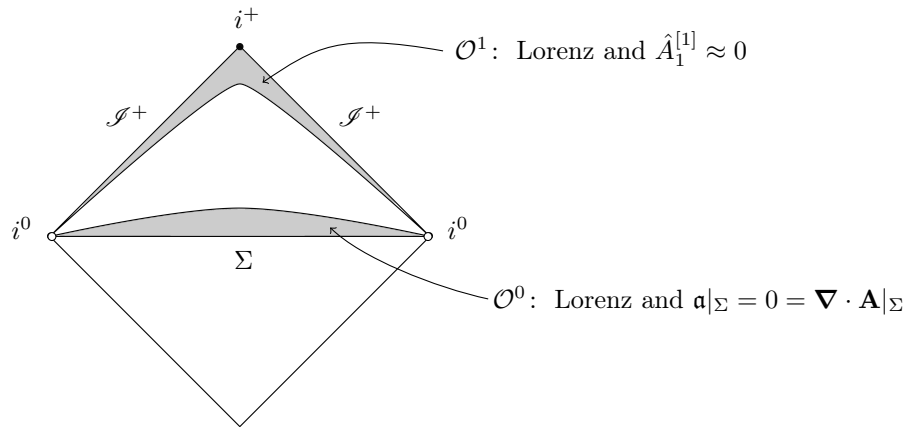
\begin{figure}[H]
\begin{tikzpicture}
\centering
\node[inner sep=0pt] (gauge_curved) at (7,0)
    {\includegraphics[width=.32\textwidth]{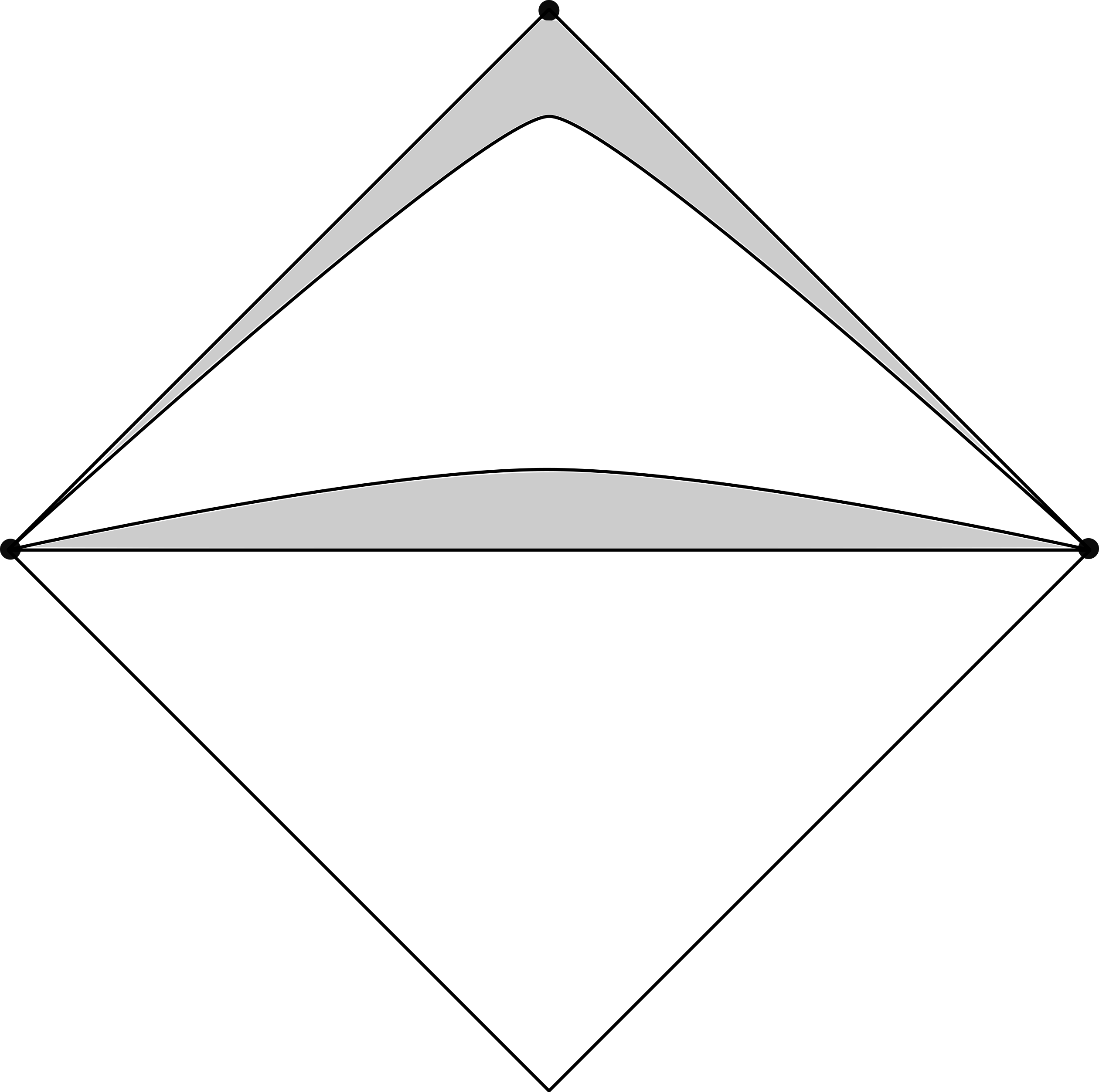}};

\node[label={[shift={(7.1,2.5)}]$i^+$}] {};
    
\node[label={[shift={(5.5, 1.2)}]$\mathscr{I}^+$}] {};
\node[label={[shift={(8.6, 1.2)}]$\mathscr{I}^+$}] {};

\node[label={[shift={(4.1,-0.3)}]$i^0$}]{};
\node[label={[shift={(9.9,-0.3)}]$i^0$}]{};

\node[label={[shift={(13.1,-1.3)}]$\mathcal{O}^0\colon$ Lorenz and $\mathfrak{a}|_\Sigma = 0 = \boldsymbol{\nabla}\cdot \mathbf{A}|_\Sigma$}]{};
\draw[->] (10.3,-0.85) .. controls (9.6,-0.7) and (8.6, -0.3) .. (8.2, 0.1);

\node[label={[shift={(11.8,2)}]$\mathcal{O}^1 \colon$ Lorenz and $\hat{A}_1^{[1]} \approx 0$}] {};
\draw[->] (9.7,2.45) .. controls (8.3,2.7) and (8, 2.5) .. (7.3,2.03);

\node[label={[shift={(7.05,-0.7)}]$\Sigma$}]{};

\draw[fill=white] (4.5,-0.02) circle (1.5pt);
\draw[fill=white] (9.5,-0.02) circle (1.5pt);

\end{tikzpicture} \vspace{-0.25cm}
\caption{Construction of the gauge on a generic Corvino--Schoen--Chru\'sciel--Delay spacetime.} \label{fig:gauge_curved}
\end{figure}

\subsection{Energy Estimates and Scattering Data}

\begin{theorem}  \label{thm:curvedestimates} For smooth compactly supported Maxwell data on $\Sigma$ one has the energy estimate
\begin{equation} \label{basicestimateMaxwellcurved} \mathcal{E}_{\scri^+} \simeq \int_{\scri^+} | \hat{F}_2 |^2 \, \widehat{\dvol}_{\scri^+} \simeq \int_{\Sigma} \left( |\mathbf{E}|^2 + |\mathbf{B}|^2 \right) \dvol_\Sigma \simeq \mathcal{E}_{\Sigma}.	
\end{equation}
\end{theorem}

\begin{proof}
The full details of the proof are given in \Cref{sec:Maxwell_energy_estimates_proof}. Recall that the energies are defined by \eqref{Maxwellconformalinvarianceenergies}. For clarity, we point out that the estimate is performed in three regions separately, a neighbourhood $U^0$ of $i^0$, a neighbourhood $U^+$ of $i^+$, and an intermediate region $U$. In $U^0$ we use the Schwarzschildean Killing vector field $\partial_t$ as the multiplier, which immediately gives the above estimate near $i^0$. In $U^+$ we use the multiplier $- \hat{\nabla}^a \Omega$; this decomposes into a term proportional to $\hat{n}^a$ and a lower-order term decaying like $\Omega$, which depends on the remaining vectors in the tetrad. Finally, in $U$ we use a mutiplier which interpolates between the one in $U^+$ and the one in $U^0$.
\end{proof}

\subsubsection{Space of initial data} \label{sec:space_of_initial_data_curved}

We construct the space of initial data on $\Sigma$ by working with the physical potential $A_a$. We compute, in general, the expressions for the electric and magnetic fields
\begin{align}
\label{initial_fields_from_potential}
 \begin{split}
     \mathbf{E}_a &= -h^b_a \nabla_T \mathbf{A}_b - \mathfrak{a} h^b_a \nabla_T  T_b + \boldsymbol{\nabla}_a \mathfrak{a} - \mathbf{A}_b \kappa_a^{\phantom{a}b}, \\
     \mathbf{B}_a &= \epsilon_a^{\phantom{a}bc} \boldsymbol{\nabla}_b \mathbf{A}_c,
 \end{split}
\end{align}
where $\epsilon_{abc}$ is the volume form on $\Sigma$. Noting that on $\Sigma$ we have $\mathfrak{a}|_\Sigma = 0$, this gives 
\begin{equation} 
\label{electric_field_in_terms_of_potential_initial_surface}
\mathbf{E}_a|_\Sigma = \Big( -h^b_a \nabla_T \mathbf{A}_b -\mathbf{A}_b \kappa_a^{\phantom{a}b} \Big)\Big|_\Sigma.
\end{equation}
In turn, we find
\[ \int_\Sigma |\mathbf{E}|^2 \dvol_\Sigma = \int_\Sigma | \nabla_T \mathbf{A} - \mathbf{A} \cdot \kappa |^2  \dvol_\Sigma, \]
where $(\mathbf{A} \cdot \kappa)_a = \mathbf{A}_b \kappa_a^{\phantom{a}b}$, and the squares are with respect to the positive-definite metric $h_{ab}$ on $\Sigma$. Next, for $\mathbf{A} \in \mathcal{C}^\infty_c(\Sigma)$ we have, using the Coulomb gauge on $\Sigma$,
\[ \int_\Sigma |\mathbf{B}|^2 \dvol_\Sigma = \int_{\Sigma} |\boldsymbol{\nabla} \mathbf{A}|^2 - \mathbf{R}_{ij} \mathbf{A}^i \mathbf{A}^j \dvol_\Sigma, \]
where $\mathbf{R}_{ij}$ is the Ricci curvature of $(\Sigma, h_{ab})$. If the Ricci and extrinsic curvatures of $\Sigma$ are bounded, then it is easy to see that for $(\mathbf{A}, \nabla_T \mathbf{A}) \in H^1(\Sigma) \oplus L^2(\Sigma)$ one has the estimate $\|\mathbf{E}\|^2_{L^2(\Sigma)} + \|\mathbf{B}\|^2_{L^2(\Sigma)} \la \|\mathbf{A} \|^2_{H^1(\Sigma)} + \|\nabla_T \mathbf{A} \|^2_{L^2(\Sigma)}$.

We claim that there is a one-to-one correspondence between $(\mathbf{E}, \mathbf{B}) \in L^2(\Sigma)^2$ and $(\mathbf{A}, \nabla_T \mathbf{A})$ living in a suitable Hilbert space. Suppose we have $\mathbf{B} \in L^2(\Sigma)$ with $\boldsymbol{\nabla} \cdot \mathbf{B} = 0$ in the sense of distributions. By the Poincar\'e lemma, we know that there exists $\mathbf{A} \in \mathcal{D}'(\Sigma)$ such that $\mathbf{B} = \boldsymbol{\nabla} \times \mathbf{A}$, where $(\boldsymbol{\nabla} \times \mathbf{A} )_i = \epsilon_i{}^{jk}\boldsymbol{\nabla}_j \mathbf{A}_k$. We impose that $\boldsymbol{\nabla} \cdot \mathbf{A} = 0$ in the sense of distributions. By the above energy identity, we have
\[ \| \mathbf{A} \|^2_{\dot{H}^1} \leq \| \mathbf{B} \|^2_{L^2} + \int_\Sigma | \mathbf{R}_{ij} \mathbf{A}^i \mathbf{A}^j |\dvol_\Sigma. \]
This may be written, using Hardy's inequality on $\Sigma$ (\cite{DambrosioDipierro2013}, eq. (1.1)), as
\begin{equation} \label{estimate_potential_H1} \|\mathbf{A}\|^2_{\dot{H}^1} \leq \| \mathbf{B} \|^2_{L^2} +C \delta \|\mathbf{A}\|^2_{\dot{H}^1}, \end{equation}
where
\[ \delta \defeq \left\| r^2 \mathbf{R} \right\|_{L^\infty(\Sigma)}. \]
For $\delta$ small enough this implies $\mathbf{A} \in \dot{H}^1(\Sigma)$. That is, we assume that the Ricci curvature of $\Sigma$ is sufficiently small on all of $\Sigma$,
\begin{equation} \label{curvature_assumption} C \delta < 1.
\end{equation}
Since in the Schwarzschild sector the Ricci curvature is given by
\[ \mathbf{R}^i{}_j = \frac{m}{r^3} \left( \begin{array}{ccc} -2 & & \\ & 1 & \\ & & 1 \end{array} \right), \]
this is automatically satisfied for $r \gg 1$, and amounts to a smallness assumption on $m$ if the Schwarzschild sector happens to contain a region of small $r$. Now differentiating the expression for $\mathbf{B}$ and using the Coulomb gauge, we find
\begin{equation}
\label{elliptic_equation_potential_sigma}
\boldsymbol{\Delta} \mathbf{A}_k + \mathbf{R}_{kj} \mathbf{A}^j = - (\boldsymbol{\nabla} \times \mathbf{B})_k .
\end{equation}
It is clear that $\boldsymbol{\nabla} \times \mathbf{B}  \in H^{-1}(\Sigma)$, but in fact also $\boldsymbol{\nabla} \times \mathbf{B} \in \dot{H}^{-1}(\Sigma)$, where $\dot{H}^{-1}(\Sigma)$ is the dual space of $\dot{H}^1(\Sigma)$. This follows from the fact that $\mathcal{C}^\infty_c(\Sigma)$ is dense in $\dot{H}^1(\Sigma)$ and integration by parts, as in analysis of equation \eqref{curlB}. Similarly, if $\mathbf{A} \in \dot{H}^1(\Sigma)$, then $\boldsymbol{\Delta} \mathbf{A} \in \dot{H}^{-1}(\Sigma)$. Further, since $\mathbf{R}^k{}_{j}$ on the Schwarzschild sector decays like $\sim m r^{-3}$, then by Hardy's inequality as before, we have $\mathbf{R}_{kj} \mathbf{A}^j \in \dot{H}^{-1}(\Sigma)$. The operator $(\mathbf{P} \mathbf{A})_k \defeq - \boldsymbol{\Delta} \mathbf{A}_k - \mathbf{R}_{kj} \mathbf{A}^j$ therefore maps $\dot{H}^1(\Sigma) \to \dot{H}^{-1}(\Sigma)$, and is continuous, elliptic, formally self-adjoint and coercive (as a consequence of \eqref{curvature_assumption}),
\[ D(\mathbf{A}, \mathbf{A}) \geq (1- C \delta ) \|\mathbf{A}\|^2_{\dot{H}^1(\Sigma)} , \]
where
\[ D(\mathbf{U}, \mathbf{V}) \defeq \int_\Sigma \mathbf{U}^k (\mathbf{P} \mathbf{V})_k \dvol_\Sigma. \]
To get uniqueness of $\mathbf{A}$, it remains to investigate the kernel of the curl operator. We claim that on $\dot{H}^1(\Sigma)$ this is equal to the kernel of $\mathbf{P}$, which in turn is trivial on $\dot{H}^1(\Sigma)$ under the assumption \eqref{curvature_assumption}. Indeed, $\operatorname{ker} \mathbf{P}$ consists of those potentials $\mathbf{A} \in \dot{H}^1(\Sigma)$ for which $\boldsymbol{\nabla} \times \mathbf{B} = 0$, i.e. by the Poincar\'e lemma $\mathbf{B} = \boldsymbol{\nabla} \phi$ for some $\phi \in \dot{H}^1(\Sigma)$ (precisely because $\mathbf{B} \in L^2(\Sigma)$). But one also has $\boldsymbol{\nabla} \cdot \mathbf{B} = 0$, which together imply $\boldsymbol{\Delta} \phi = 0$. Since $\phi \in \dot{H}^1(\Sigma)$, it can be approached by $\phi_n \in \mathcal{C}^\infty_c(\Sigma)$ in $\dot{H}^1(\Sigma)$, so, as $\boldsymbol{\Delta} \phi \in \dot{H}^{-1}(\Sigma)$,
\[ 0 = - \int_\Sigma \phi_n \boldsymbol{\Delta} \phi \dvol_\Sigma = \int_\Sigma \boldsymbol{\nabla} \phi_n \cdot \boldsymbol{\nabla} \phi \dvol_\Sigma \to \| \phi \|^2_{\dot{H}^1(\Sigma)}.  \]
Hence $\boldsymbol{\nabla} \phi = \mathbf{B} = 0$. Then, for $\mathbf{B} = 0$, $\mathbf{A} = 0$ follows from the coercivity of $\mathbf{P}$.
We therefore define
\begin{equation}
    \label{initial_surface_potential_space}
    \dot{H}^1_C(\Sigma)^\text{curl} \defeq \big\{ \mathbf{A}^{(0)} \in \dot{H}^1(\Sigma) \, : \, \boldsymbol{\nabla} \cdot \mathbf{A}^{(0)} = 0, ~ \epsilon_i^{\phantom{i}jk} \boldsymbol{\nabla}_j \mathbf{A}^{(0)}_k \in L^2(\Sigma) \big\} .
\end{equation}
The equation \eqref{elliptic_equation_potential_sigma} therefore has a unique solution in $\dot{H}^1_C(\Sigma)^\text{curl}$, which we write as $\mathbf{A}^{(0)} = \mathbf{P}^{-1}(-\boldsymbol{\nabla} \times \mathbf{B})$. Note that, by construction, $\mathcal{C}^\infty_c(\Sigma)$ is dense in $\dot{H}^1_C(\Sigma)^{\text{curl}}$ because $\mathcal{C}^\infty_c(\Sigma)$ is dense in $L^2(\Sigma)$, and the choice of norm on $\dot{H}^1_C(\Sigma)^\text{curl}$ is \emph{precisely} the norm on $\mathbf{B}$.

\begin{remark} For unrestricted $\delta$, due to the lack of positivity of the spacelike Ricci curvature and the fact that $\Sigma$ is unbounded, we expect the precise control of the kernel of $\mathbf{P}$ to be a very delicate question.
\end{remark}

Given $\mathbf{A}^{(0)} \in \dot{H}^1_C(\Sigma)^\text{curl}$ and $\mathbf{E} \in L^2(\Sigma)$, we then reconstruct the time derivative component of the initial data from \eqref{electric_field_in_terms_of_potential_initial_surface},
\[ \mathbf{A}^{(1)}_a \defeq -h^b_a \nabla_T \mathbf{A}_b = \mathbf{E}_a + (\mathbf{A}^{(0)} \cdot \kappa)_a.  \]
Since $\mathbf{E} \in L^2(\Sigma)$, $\mathbf{A}^{(0)} \in L^2_{\text{loc}}(\Sigma)$ and $\kappa \in \mathcal{C}^k(\Sigma)$, we have $\mathbf{A}^{(1)} \in L^2_{\text{loc}}(\Sigma)$. Moreover, since $\kappa$ vanishes in the Schwarzschild sector, in fact the $L^2$ norm of $\mathbf{A}^{(1)}$ is controlled by the $L^2$ norm of $\mathbf{E}$ plus a constant, i.e. $\mathbf{A}^{(1)} \in L^2(\Sigma)$. We therefore have a bijection
\begin{align*}
    \operatorname{Int} : L^2(\Sigma) \oplus L^2(\Sigma) &\longrightarrow \dot{H}^1_C(\Sigma)^\text{curl} \oplus L^2(\Sigma) \\
    (\mathbf{B}, \mathbf{E}) &\longmapsto (\mathbf{A}^{(0)}, \mathbf{A}^{(1)}) = (\mathbf{P}^{-1} ( - \boldsymbol{\nabla} \times \mathbf{B})), ~ \mathbf{E} + \mathbf{A}^{(0)} \cdot \kappa ).
\end{align*}
Our space of initial data for the components $(\mathbf{A}, \dot{\mathbf{A}})$ is therefore $\dot{H}^1_C(\Sigma)^\text{curl} \oplus L^2(\Sigma)$, the elements of which, by construction, are in 1-to-1 correspondence with pairs of fields $(\mathbf{E}, \mathbf{B}) \in L^2(\Sigma)^2$.

For completeness, we also observe how to prescribe data for the component $\mathfrak{a}$. Our gauge implies that
\begin{equation}
\label{little_a_initial_data}
(\mathfrak{a}, \nabla_T \mathfrak{a})|_\Sigma = (0, \, \mathbf{A}^{(0)} \cdot \boldsymbol{\nabla} \log N ).
\end{equation}
Certainly $(\mathfrak{a}, \nabla_T \mathfrak{a}) \in \mathcal{C}^\infty_c(\Sigma) \oplus L^2_{\text{loc}}(\Sigma)$. In the Schwarzschild sector in fact $\log N = \frac{1}{2} \log (1-2mr^{-1})$, so $\mathbf{A}^{(0)} \cdot \boldsymbol{\nabla} \log N \sim \frac{1}{r} \mathbf{A}_r^{(0)}$, so $\nabla_T \mathfrak{a}$ decays one order faster than $\mathbf{A}^{(0)}$, i.e. $\nabla_T \mathfrak{a} \in L^2(\Sigma)$ by Hardy.

\subsubsection{Space of scattering data} \label{sec:scattering_data_potential}

As in \Cref{sec:Maxwellpotentialsflatscattering}, the condition \eqref{Lorenz_gauge_first_order_reduction} also implies $\eth \hat{A}_1 \approx 0$, and by \Cref{prop:spin_coefficients_near_scri}, the expression \eqref{F2components} for $\hat{F}_2$ on $\scri^+$ reduces to
\[ \hat{F}_2 \approx - \partial_u \bar{\hat{A}}_2. \]
Hence
\begin{equation}
    \label{energy_on_scri}
    \int_{\scri^+} |\hat{F}_2|^2 \, \widehat{\dvol}_{\scri^+} = \int_{\scri^+} |\partial_u \hat{A}_2|^2 \, \d u \wedge \dvol_{\mathbb{S}^2}.
\end{equation} 
Suppose we have $\hat{F}_2 \in \mathcal{C}^\infty(\scri^+)$ supported away from $i^0$. We then put
\[ \hat{A}^+_2 \defeq - \int_{-\infty}^u \bar{\hat{F}}_2 \, \d u \in \mathcal{C}^\infty(\scri^+), \]
which remains supported away from $i^0$. Using the construction in \Cref{sec:gauge_on_scri}, we then define
\[ \hat{A}_1^+ \equiv 0 \]
and
\[ \hat{A}_0^+ \defeq \int_{-\infty}^u 2 \operatorname{Re} (\hat{\eth} \bar{\hat{A}}_2^+ ) \, \d u \in \mathcal{C}^\infty(\scri^+). \]
We then define the space $\dot{\mathcal{H}}^1(\scri^+) \simeq \dot{H}^1(\mathbb{R}_u; L^2(\mathbb{S}^2))$ of scattering data by completing $(\hat{A}^+_0, \hat{A}_1^+, \hat{A}_2^+) \in \mathcal{C}^\infty(\scri^+) \times \mathcal{C}^\infty_c(\scri^+) \times \mathcal{C}_c^\infty(\scri^+)$ in the norm \eqref{energy_on_scri}, as in \Cref{def:characteristicdatapotential_Minkowski}.

\subsection{Trace and Scattering Operators} \label{sec:scattering_operator_curved_spacetimes}

\subsubsection{The forward Cauchy problem}

Suppose we are given $(\mathbf{A}^{(0)}, \mathbf{A}^{(1)}) \in \mathcal{C}^\infty_c(\Sigma)^2 \cap (\dot{H}^1_C(\Sigma)^{\text{curl}} \oplus L^2(\Sigma))$ such that $\boldsymbol{\nabla} \cdot \mathbf{A}^{(0)} = 0$. We obtain $(\mathfrak{a}, \nabla_T \mathfrak{a}) \in \mathcal{C}^\infty_c(\Sigma)^2$ using \eqref{little_a_initial_data}, and then reconstruct the initial fields on $\Sigma$ using \eqref{initial_fields_from_potential}. We therefore obtain a smooth compactly supported $\hat{F}_{ab}|_\Sigma$, which we propagate in $\hat{\mathcalboondox{M}}$ as we did in \Cref{thm:MaxwellpotentialflatCauchyproblem} (see Lemma 2.4 in \cite{MasonNicolas2004}). We thus obtain a smooth $\hat{F}_{ab}$ on $\scri^+$ which, by finite speed of propagation, is supported away from $i^0$, and which satisfies the estimate \eqref{basicestimateMaxwellcurved}. On $\scri^+$, we reconstruct the potential as described in \Cref{sec:scattering_data_potential}. We therefore obtain a linear operator
\begin{align*} \mathfrak{T}^+ : \mathcal{C}^\infty_c(\Sigma)^2 \cap (\dot{H}^1_C(\Sigma)^{\text{curl}} \oplus L^2(\Sigma)) &\longrightarrow \dot{\mathcal{H}}^1(\scri^+) \\
(\mathbf{A}^{0}, \mathbf{A}^{(1)}) &\longmapsto (\hat{A}_0^+, \hat{A}_1^+, \hat{A}_2^+) = \left( \int_{-\infty}^u 2 \operatorname{Re}(\hat{\eth}\bar{\hat{A}}_2^+) \, \d u, ~ 0, ~ \hat{A}_2^+ \right),
\end{align*}
which extends by density to a bounded linear operator
\begin{equation}
\label{forward_trace_operator_curved_spacetimes}
\mathfrak{T}^+ : \dot{H}^1_C(\Sigma)^{\text{curl}} \oplus L^2(\Sigma) \longrightarrow \dot{\mathcal{H}}^1(\scri^+).
\end{equation}

\subsubsection{The Goursat problem}

To show that the operator \eqref{forward_trace_operator_curved_spacetimes} is invertible, it now remains to show that we can solve the Goursat problem from data in a dense subspace of $\dot{\mathcal{H}}^1(\scri^+)$, and that the solution gives rise to a unique $(\mathbf{A}^{(0)}, \mathbf{A}^{(1)}) \in \dot{H}^1_C(\Sigma)^{\text{curl}} \oplus L^2(\Sigma)$. We have the following.
\begin{theorem} \label{thm:Goursat_problem_curved_spacetimes} For every triplet
\[(\hat{A}_0^+, \hat{A}_1^+, \hat{A}_2^+) = \left(\int_{-\infty}^u 2 \operatorname{Re}(\hat{\eth} \bar{\hat{A}}_2^+) \, \d u , \,  0 , \, \hat{A}_2^+ \right) \in \mathcal{C}^\infty(\scri^+) \times \mathcal{C}^\infty_c(\scri^+) \times \mathcal{C}_c^\infty(\scri^+) \] 
with $\hat{A}_0^+$ supported away from $i^0$ there exists a unique solution $\hat{A}_a \in \mathcal{C}^0(t(\mathcalboondox{M}); \, H^1(\Sigma_t)) \cap \mathcal{C}^1(t(\mathcalboondox{M}); \, L^2(\Sigma_t))$ to \eqref{Maxwells_equations_rescaled_physical_Lorenz} which in particular satisfies $(\mathbf{A}^{(0)}, \mathbf{A}^{(1)})\in \dot{H}^1_C(\Sigma)^{\text{curl}} \oplus L^2(\Sigma)$.
\end{theorem}

\begin{proof}
We proceed as in \Cref{sec:Minkowski_Goursat_problem}. Suppose we are given data 
$(\hat{A}_0^+, \hat{A}_1^+, \hat{A}_2^+)$ as above. Working first in the neighbourhood of $\scri^+$ in which our gauge holds, and subsequently working with the field, we solve the Goursat problem in two steps. Introduce a short outgoing null hypersurface $\mathcal{H}$ which intersects $\scri^+$ in the future of the support of $\hat{A}_2^+$, as shown in \Cref{fig:outgoing_null_hypersurface_Minkowski}. In the future of $\mathcal{H}$, the data for $\hat{A}_1$ and $\hat{A}_2$ is identically zero, whereas the data for $\hat{A}_0$ is a constant, say $\hat{A}_0 \approx C_0$. Therefore the field $\hat{F}_{ab}$ there has identically zero data on $\scri^+$, $\hat{F}_0^+ = \hat{F}_1^+ = \hat{F}_2^+ = 0$. Briefly changing conformal scale in $\mathcal{O}^+$ to bring $i^+$ to a finite distance, we may solve a wave equation for $\hat{F}_{ab}$ in $\mathcal{O}^+$, the solution to which, by the uniqueness part of \Cref{thm:Goursatproblem}, must be $\hat{F}_{ab} \equiv 0$. Change conformal scale back to that of \Cref{prop:spin_coefficients_near_scri}. We therefore have that in $\mathcal{O}^+$ the potential is pure gauge, $\hat{F}_{ab} \equiv 0 \implies \hat{A}_a = \hat{\nabla}_a \chi$, where we use the Poincar\'e Lemma to obtain the existence of such a $\chi \in H^2_{\text{loc}}(\mathcal{O}^+)$. A priori $\chi$ is not unique, but we recall that our solution satisfies the gauge constructed in \eqref{sec:gauge_on_scri}. This imposes the conditions $\Box \chi = 0$, $\hat{\thorn}'(\Omega^{-1} \chi) \approx 0$, and $\Delta_{\mathbb{S}^2} \chi \approx 0$. We have already seen that the equation $\Box \chi = 0$ is equivalent to
\[ \widehat{\Box} \hat{\chi} + 2\hat{\mu}^{[1]} f \hat{\chi} = 0, \]
where $\hat{\chi} = \Omega^{-1} \chi$, so to obtain uniqueness we need only fix the data for $\hat{\chi}$ on $\scri^+ \cap \overline{\mathcal{O}^+}$. We compute
\[ \hat{\thorn}(\Omega^{-1} \chi ) = -f\Omega^{-1} \hat{\chi} + \Omega^{-1} \hat{\thorn} \chi, \]
which, noting that $\hat{\thorn} \chi \approx C_0$, implies that we should have
\[ \hat{\chi} \approx \frac{C_0}{f}, \]
and so this fixes $\chi$ in $\mathcal{O}^+$.

By restriction to $\mathcal{H}$ of $\hat{\thorn} \chi$, $\hat{\thorn}' \chi$ and $\hat{\eth}\chi$, we therefore obtain $H^1_c$ data for $\hat{A}_0$, $\hat{A}_1$ and $\hat{A}_2$ on $\mathcal{H}$. Since the full data is now $H^1_c(\mathcal{H}\cup (\scri^+\cap \overline{\mathcal{O}^-}))$, we may apply \Cref{thm:Goursatproblem} from $\mathcal{H} \cup (\scri^+ \cap \overline{\mathcal{O}^-})$ to solve \eqref{Maxwells_equations_rescaled_physical_Lorenz} in the region $\mathcal{O}^-$. We thus obtain a solution $\hat{A}_a$ to \eqref{Maxwells_equations_rescaled_physical_Lorenz} in the neighbourhood $\mathcal{O}^+ \cup \mathcal{O}^- \subset \mathcal{O}^1$ of $\scri^+$ which has the regularity
\[ \hat{A}_a \in \mathcal{C}^0(t(\mathcalboondox{M}); H^1(\Sigma_t)) \cap \mathcal{C}^1(t(\mathcalboondox{M}); L^2(\Sigma_t)). \]
On a slice $\Sigma_{t'}$, $t' < \infty$, contained in $\mathcal{O}^+ \cup \mathcal{O}^-$, we now reconstruct the physical field $F_{ab}|_{\Sigma_{t'}} \in L^2_c(\Sigma_{t'})$ from $\hat{A}_a = A_a$ and propagate $F_{ab}$ backwards in time to $\Sigma$. We thus obtain $(\mathbf{E}, \mathbf{B}) \in L^2_c(\Sigma)^2$ on $\Sigma$, and, using the operator $\operatorname{Int}$ constructed in \Cref{sec:space_of_initial_data_curved}, we obtain
\[  (\mathbf{A}^{(0)}, \mathbf{A}^{(1)}) =  \operatorname{Int}(\mathbf{E}, \mathbf{B}) \in \dot{H}^1_C(\Sigma)^{\text{curl}} \oplus L^2(\Sigma). \]
\end{proof}

\begin{corollary}
    The trace operator $\mathfrak{T}^+$ defined in \eqref{forward_trace_operator_curved_spacetimes} is invertible, and maps
    \[ (\mathfrak{T}^+)^{-1} : \dot{\mathcal{H}}^1(\scri^+) \longrightarrow \dot{H}^1_C(\Sigma)^{\text{curl}} \oplus L^2(\Sigma). \]
\end{corollary}
\begin{proof}
    This follows from \Cref{thm:Goursat_problem_curved_spacetimes} and the discussion in \Cref{rmk:surjectivity_of_trace_operator}.
\end{proof}

We may perform the same construction towards past null infinity, and it then follows immediately that the composition $\mathfrak{T}^+ \circ (\mathfrak{T}^-)^{-1}$ is an isomorphism. We conclude with a definition as in the case of Minkowski space.

\begin{definition}[Scattering operator on CSCD spacetimes]
We define the \emph{scattering operator for Maxwell potentials} on $\mathcalboondox{M}$ to be the isomorphism of Hilbert spaces
\begin{align*}
\mathscr{S} \defeq \mathfrak{T}^+ \circ (\mathfrak{T}^-)^{-1} : \dot{\mathcal{H}}^1(\scri^-) & \longrightarrow \dot{\mathcal{H}}^1(\scri^+) \\
(\hat{A}^-_0, \hat{A}^-_1, \hat{A}_2^- ) &\longmapsto (\hat{A}^+_0, \hat{A}^+_1, \hat{A}_2^+ ),
\end{align*}
where $\hat{A}_1^+ \equiv 0 \equiv \hat{A}_0^-$, and
\[ \hat{A}_0^+ = \int_{-\infty}^u 2 \operatorname{Re}(\hat{\eth} \bar{\hat{A}}_2^+) \, \d u \quad \text{and} \quad \hat{A}_1^- = \int_{-\infty}^v 2 \operatorname{Re}(\hat{\eth} \bar{\hat{A}}_2^-) \, \d v. \]
\end{definition}

\section*{Acknowledgements}

For the purpose of open access, the authors have applied a Creative Commons Attribution (CC BY) licence to any Author Accepted Manuscript version arising from this submission. The authors thank Lionel Mason for useful discussions during the preparation of this article, and the Mathematical Institute at Oxford University for hospitality in Spring 2022.

\appendix
\section{Proof of Energy Estimates for Maxwell Fields}
\label{sec:Maxwell_energy_estimates_proof}

In this section we prove the energy estimate \eqref{basicestimateMaxwellcurved}. For convenience, here we shall use spinor notation. A real Maxwell field $F_{ab}$ can be decomposed into its self-dual and anti self-dual parts that are complex conjugates of one another,
\[ F_{ab} = \phi_{AB} \bar{\varepsilon}_{A'B'} + \bar{\phi}_{A'B'} \varepsilon_{AB} \, ,\]
with $\phi_{AB} = \phi_{(AB)}$; $\varepsilon_{AB}$ and its complex conjugate $\bar{\varepsilon}_{A'B'}$ are the Levi-Civita symbols, the symplectic forms on the left- and right-handed spin bundles $\mathbb{S}^A$ and $\mathbb{S}^{A'}$, such that $g_{ab} = \varepsilon_{AB} \varepsilon_{A'B'}$. Under a conformal rescaling $\hat{g}_{ab} = \Omega^2 g_{ab}$, $\phi_{AB}$ transforms as $\hat{\phi}_{AB} = \Omega^{-1}\phi_{AB}$ and the Levi-Civita symbol as $\hat{\varepsilon}_{AB} = \Omega \varepsilon_{AB}$. The Maxwell field is then invariant under conformal rescalings and we have $\hat{F}_{ab} = F_{ab} = \hat{\phi}_{AB} \bar{\hat{\varepsilon}}_{A'B'} + \hat{\varepsilon}_{AB} \bar{\hat{\phi}}_{A'B'}$. In a normalized spin frame $\{ \hat{o}^A, \hat\iota^A \}$ the components of $\hat{\phi}_{AB}$ correspond exactly to the components of $\hat{F}_{ab}$, $\hat{\phi}_i = \hat{F}_i$, $i=0, \, 1, \, 2$, where 
\[ \hat{\phi}_0 = \hat{\phi}_{AB} \hat{o}^A \hat{o}^B, \qquad \hat{\phi}_1 = \hat{\phi}_{AB} \hat{o}^A \hat\iota^B, \quad \text{and} \quad \hat{\phi}_2 = \hat{\phi}_{AB} \hat\iota^A \hat\iota^B. \]
Maxwell's equations reduce to equations on the self-dual part $\phi_{AB}$ that are conformally invariant. On the compactified spacetime and for the rescaled self-dual Maxwell spinor, they have the following form
\begin{equation}\label{ASDMaxwell}
\hat{\nabla}^{AA'} \hat{\phi}_{AB} = 0 .
\end{equation}
The natural stress-energy tensor \eqref{stresstensor} has a very simple expression in terms of the spinors $\hat{\phi}_{AB}$ and $\overline{\hat{{\phi}}}_{A'B'}$, given by
\begin{equation} \label{MaxwellSET}
\hat{\mathbf{T}}_{ab} = \hat{\phi}_{AB} \overline{\hat{{\phi}}}_{A'B'} \, .
\end{equation}
The tensor $\hat{\mathbf{T}}_{ab}$ is symmetric, conserved on-shell, and conformally covariant with weight $-2$. Let $\Sigma = \Sigma_0$ be the $\{ t=0 \}$ slice in our asymptotically null foliation $\{ \Sigma_t \}$ of $\mathcalboondox{M}$. We denote by $\mathcalboondox{M}^+$ the future of $\Sigma$ in $\mathcalboondox{M}$ and by $\hat{\mathcalboondox{M}}^+$ its closure in $\hat{\mathcalboondox{M}}$. Let $\tau^a$ be an observer (a timelike future-oriented vector field) on $\mathcalboondox{M}^+$ that aligns on $\scri^+$ with its null generator. Defining the energy current
\begin{equation} \label{EnCurrentMaxwell}
\hat{J}_a \defeq \tau^b \hat{\mathbf{T}}_{ab} ,
\end{equation}
we have, for $\hat\phi_{AB}$ a solution to \eqref{ASDMaxwell}, the approximate conservation law
\begin{equation} \label{ACL}
\hat\nabla^a \hat{J}_a = \hat\nabla^{a} \big( \tau^{b} \hat\phi_{AB} \bar{\hat\phi}_{A'B'} \big) = \hat\nabla^{(a} \tau^{b)} \hat\phi_{AB} \bar{\hat\phi}_{A'B'} \, .
\end{equation}
The energy of the field on a given spacelike hypersurface $\mathcal{S}$ is simply the $L^2$ norm of $\hat\phi_{AB}$ on $\mathcal{S}$, with measure induced by $\hat{g}_{ab}$, and a weight associated to our choice of observer $\tau^a$. More precisely,
\begin{equation} \label{MaxEnergyS}
\mathcal{E}_\mathcal{S} [\hat\phi] = \int_\mathcal{S} \hat{J}_a \nu^a ( v \intprod \widehat{\dvol}),
\end{equation}
where $\widehat{\dvol}$ is the $4$-volume measure associated to $\hat{g}_{ab}$, $\nu^a$ is a normal vector field to $\mathcal{S}$ compatible with the orientation of $\mathcal{S}$ (i.e., future-pointing), and $v^a$ is a vector field transverse to $\mathcal{S}$ such that $\nu^a v^b \hat{g}_{ab} =1$. Since the stress-energy tensor has conformal weight $-2$, $\tau^a$ has weight $0$, $\nu^a$ and $v^a$ can be chosen to have weight $-1$ and $\widehat{\dvol}$ has weight $4$, it follows that the energy flux \eqref{MaxEnergyS} is conformally invariant, i.e. conformally covariant with weight $0$ (see also\eqref{Maxwellconformalinvarianceenergies}).

For the purpose of proving the estimate \eqref{basicestimateMaxwellcurved}, we decompose $\mathcalboondox{M}^+$ into three distinct regions:
\begin{enumerate}[(i)]
    \item a neighbourhood $U^0$ of $i^0$,
    \[ U^0 = \{ u \leq u_0 \} \mbox{ for a given } u_0 \ll -1 \, ; \]
    \item a neighbourhood $U^+$ of $i^+$,
    \[ U^+ = \{ \tau \geq \tau_0 \} \mbox{ for a given } \tau_0 \gg 1, \]
    where $\tau$ is the parameter of the foliation transverse to $\scri^+$ shown in Figure \ref{fig:transversefoliation};
    \item an intermediate region $U \defeq \overline{\mathcalboondox{M}^+ \setminus (U^0 \cup U^+)}$.
\end{enumerate}
We obtain energy estimates in each region separately. Since the energy flux \eqref{MaxEnergyS} is conformally invariant, we may work with different conformal factors in the different regions. However, we must choose our observer $\tau^a$ so that it is continuous and in fact smooth on $\hat{\mathcalboondox{M}}^+$.

In the region $U^0$, we work with the conformal factor $\Omega = 1/r$ which preserves the timelike Killing vector $K^a \partial_a = \partial_t = \partial_u$ on the Schwarzschildean neighbourhood of $i^0$. Here we choose $K^a$ for our observer,
\[ \tau^a \partial_a = \partial_u \, .\]
This gives immediate energy identities in the region $U^0$.

In the regions $U^+$ and $U$, we make use of special features of $\scri^+$ and choose a conformal factor $\Omega$ such that
\begin{enumerate}[(i)]
    \item $i^+$ is a finite regular point of the compactified spacetime $(\hat{\mathcalboondox{M}}^+ , \hat{g}_{ab} )$;
    \item $\hat{\mathrm{R}}_{ab}=0$ at $i^+$;
    \item $\hat{\mathrm{R}}$ and $\hat{n}^a \hat{\mathrm{R}}_{ab}$ vanish on $\scri^+$, where $\hat{n}^a$ is the null generator of $\scri^+$;
    \item $-\hat{\nabla}^a \Omega $ is timelike and future-oriented in $U^+ \setminus \scri^+$ and null and future-oriented on $U^+ \cap \scri^+$.
\end{enumerate}
The existence of such a conformal factor was established in \cite{MasonNicolas2004}, Lemma A.1. We work with a normalized spin frame $\{\hat{o}^A, \hat\iota^A \}$ and a Newman--Penrose tetrad $( \hl^a = \hat{o}^A \bar{\hat{o}}^{A'}, \, \hm^a = \hat{o}^A \bar{\hat\iota}^{A'}, \, \bhm^a = \bar{\hat{o}}^{A'} \hat\iota^A, \, \hn^a = \hat\iota^A \bar{\hat\iota}^{A'} )$ on $\mathcalboondox{M}^+$ such that $\hl^a$ and $\hn^a$ are real and future-oriented, smooth on $\mathcalboondox{M}^+ \setminus i^+ $, bounded and non-vanishing at $i^+$, and $\hat{n}^a$ is the null generator of $\scri^+$. We assume in addition that the vector field $\hl^a + \hn^a$ is hypersurface orthogonal. This is not a critical assumption; it may be easily removed if desired, but it turns out to simplify the following estimates.

\subsection{Energy estimates in $U^+$}

In this region we put $\tau^a \defeq -\hat{\nabla}^a \Omega$. This has the following decomposition along our null tetrad,
\[ -\hat{\nabla}^a \Omega = c_1 \hn^a + \Omega \big( c_0 \hl^a + c_2 \hm^a + \overline{c_2} \bhm^a \big)\, , \]
where $c_0$ $c_1$ and $c_2$ are smooth on $\mathcalboondox{M}^+$, $c_0$ and $c_1$ are real and positive on $\hat{\mathcalboondox{M}}^+ \setminus i ^+$ and $c_0$ vanishes at $i^+$. Here one might be tempted to work with the foliation by the level hypersurfaces of $\Omega$. The energy density on these slices is given by the quadratic form
\begin{align*}
\hat{\mathbf{T}}_{ab} \hat{\nabla}^a \Omega \hat{\nabla}^b \Omega = c_1^2 \vert \hat\phi_2 \vert^2 + \mathcal{O} ( \Omega ) \, .
\end{align*}
This is a natural foliation to choose as in the limit $\Omega \to 0$ it accumulates on $\scri^+$; indeed, this is the foliation that was used in \cite{MasonNicolas2004}. However, this also means that the energy on the slices degenerates as $\Omega \rightarrow 0$, and therefore in order to estimate the error term by the energy one needs to split the bulk integral of the error term and extract additional decay near $\scri^+$ from the $3+1$ splitting of the $4$-volume form. It is much simpler to choose a foliation transverse to $\scri^+$ whose normal vector field is given by
\[ \nu^a = \frac{1}{\sqrt{2}} ( \hl^a + \hn^a ) \, .\]
For this choice, the energy density becomes
\begin{equation} \label{EnergyDensityTransverseFoliationU+}
\hat{\mathbf{T}}_{ab} \tau^a \nu^b =  \frac{1}{\sqrt{2}} \big( c_1 \vert \hat\phi_2 \vert^2 + c_1 \vert \hat\phi_1 \vert^2 \big) + \mathcal{O}(\Omega ) \, .
\end{equation}
The advantage of such a foliation is that it does not degenerate near $\scri^+$, and the associated $3+1$ splitting of the $4$-volume form does not induce any additional decay. It is therefore enough to show directly that the bulk error term is controlled by \eqref{EnergyDensityTransverseFoliationU+}.

The Killing form of $\tau^a$ is
\[ \hat\nabla_{(a} \tau_{b)} = -\hat\nabla_{(a} \hat\nabla_{b)} \Omega = -\hat\nabla_a \hat\nabla_b \Omega, \]
and its behaviour can be understood using the conformal transformation law of the trace-free part of the Ricci tensor $\Phi_{ab}$, and the asymptotic Einstein condition \eqref{asymptoticEinstein}. We first note that it splits into two parts,
\begin{align*}
\hat\nabla_a \tau_b = \hat\nabla_{(a} \tau_{b)} &= -\hat\nabla_{[A'\vert [A} \hat\nabla_{B]\vert B']} \Omega - \hat\nabla_{(A'\vert (A} \hat\nabla_{B)\vert B')} \Omega \\
&= - \frac14 (\widehat{\Box} \Omega ) \hat{g}_{ab} - \hat\nabla_{A'(A} \hat\nabla_{B)B'} \Omega \, .
\end{align*}
The first part will not appear in the divergence of the energy current. Indeed, due to the symmetry of $\hat{\phi}_{AB}$, \eqref{ACL}, we have
\[ \hat{g}^{ab} \hat\phi_{AB} \overline{\hat\phi}_{A'B'} = \varepsilon^{AB} \hat\phi_{AB} \varepsilon^{A'B'} \overline{\hat\phi}_{A'B'} = 0 \, .\]
Therefore \eqref{ACL} becomes
\[ \hat\nabla^a \hat{J}_a = -(\hat\nabla^{A'(A} \hat\nabla^{B)B'} \Omega ) \hat\phi_{AB} \bar{\hat\phi}_{A'B'} \, .\]
The conformal transformation law for $\Phi_{ab}$ is given by
\[
\Phi_{ab} = \hat\Phi_{ab} + \Omega^{-1} \hat\nabla_{A'(A} \hat\nabla_{B)B'} \Omega,
\]
so---recalling the assumptions that $\Phi_{ab} = \mathcal{O}(\Omega^2)$ and that $\hat{\Phi}_{ab}$ is smooth at $\scri^+ \cup i^+$---we infer that the asymptotic Einstein condition \eqref{asymptoticEinstein} may be rewritten as
\begin{equation} \label{KFAsymptotics1U+}
-\hat\nabla_{A'(A} \hat\nabla_{B)B'} \Omega = \Omega \hat\Phi_{ab} - \Omega \Phi_{ab} = \Omega \hat\Phi_{ab} + \mathcal{O} (\Omega^3 ) \, .
\end{equation}
Decomposing $\hat\Phi_{ab}$ and $\hat\phi_{AB}$ along the spin frame $\{ \hat{o}^A, \hat\iota^A \}$, we can express the leading order part of the error term as follows (omitting the factor of $\Omega$ for the moment),
\begin{align*}
\hat\Phi^{ab} \hat{\mathbf{T}}_{ab} &= \hat\Phi^{ab} \hat\phi_{AB} \bar{\hat\phi}_{A'B'} \\
&= \hat\Phi_{00} \vert \hat\phi_0 \vert^2 + 4 \hat\Phi_{11} \vert \hat\phi_1 \vert^2 + \hat\Phi_{22} \vert \hat\phi_{2} \vert^2 - 4 \Re ( \hat\Phi_{01} \hat\phi_0 \overline{\hat\phi_1} ) - 4 \Re ( \hat\Phi_{12} \hat\phi_1 \overline{\hat\phi_2} ) +2 \Re ( \hat\Phi_{02} \hat\phi_0 \overline{\hat\phi_2} ) \, .
\end{align*}
We now recall that our conformal factor $\Omega$ was such that (ii) $\hat\Phi_{ab}$ and $\hat{\mathrm{R}}$ vanish at $i^+$, and (iii) $\hn^a \hat\Phi_{ab}$ and $\hat{\mathrm{R}}$ vanish on $\scri^+$. We have
\[ \hn^a \hat{\Phi}_{ab} = \hat{\Phi}_{11} \hn_b - \hat{\Phi}_{21} \hm_b - \bar{\hat{\Phi}}_{21} \bhm_b + \hat{\Phi}_{22} \hl_b \, ,\]
so $\hat{\Phi}_{11}$, $\hat{\Phi}_{21}$ and $\hat{\Phi}_{22}$ vanish on $\scri^+$ and are therefore such that
\[ \hat{\Phi}_{11}^{[1]} \defeq \Omega^{-1} \hat{\Phi}_{11}, \qquad \hat{\Phi}_{21}^{[1]} \defeq \Omega^{-1} \hat{\Phi}_{21} \quad \text{and} \quad \hat{\Phi}_{22}^{[1]} \defeq \Omega^{-1} \hat{\Phi}_{22} \]
are smooth at $\scri^+$. Hence, we can decompose the divergence of the energy current into terms that are of order $\mathcal{O}(\Omega)$, $\mathcal{O}(\Omega^2)$, and higher:
\begin{align*}
\hat\nabla^a \hat{J}_a &= \Omega \big( \hat\Phi_{00} \vert \hat\phi_0 \vert^2 - 4 \Re ( \hat\Phi_{01} \hat\phi_0 \overline{\hat\phi_1} ) +2 \Re ( \hat\Phi_{02} \hat\phi_0 \overline{\hat\phi_2} ) \big) \\
&\phantom{=} + \Omega^2 \big( 4 \hat{\Phi}^{[1]}_{11} \vert \hat\phi_1 \vert^2 + \hat{\Phi}^{[1]}_{22} \vert \hat\phi_2 \vert^2 - 4 \Re (\hat{\Phi}^{[1]}_{12} \hat\phi_1 \overline{\hat\phi_2} ) \big) \\
&\phantom{=} + \mathcal{O}(\Omega^3).
\end{align*} This error term is easily controlled by the energy density \eqref{EnergyDensityTransverseFoliationU+}. We can therefore obtain energy estimates in both directions on $U^+$ using Gr\"onwall's Lemma. This is done as follows. We start from the energy identity
\begin{equation} \label{EnId}
\mathcal{E}_{\scri^+_{\tau_{_0}}} - \mathcal{E}_{\Sigma_{\tau_{_0}}} = \int_{U^+} \hat\nabla^a \hat{J}_a \,\widehat{\dvol} ,
\end{equation}
where $\scri^+_{\tau_{_0}}$ is the part of $\scri^+$ in the future of $\Sigma_{\tau_{_0}}$.
\begin{enumerate}
    \item {\bf Forward-in-time estimate.} Assuming that the parameter $\tau$ of the foliation $\{ \Sigma_\tau \}_\tau$ of $\mathcalboondox{M}^+$ ranges from $0$ to $T$, we introduce the hypersurfaces $S_\tau$, $\tau_0 \leq \tau \leq T$, in $\overline{\mathcalboondox{M}^+}$ that are the union of $\Sigma_\tau$ and the part of $\scri^+$ in the past of $\Sigma_\tau$ and in the future of $\Sigma_{\tau_{_0}}$. Then \eqref{EnId} can be rewritten as follows,
    \[ \mathcal{E}_{S_{T}} = \mathcal{E}_{S_{\tau_{_0}}} + \int_{U^+} \hat\nabla^a \hat{J}_a \, \widehat{\dvol} \]
    from which we infer an estimate using the control of the error terms by the energy density obtained above:
    \[ \mathcal{E}_{S_{T}} 
    \leq  \mathcal{E}_{S_{\tau_{_0}}} + \int_{\tau_{_0}}^T \mathcal{E}_{S_{\tau}} \d \tau \]
    We also have the intermediate inequalities
    \[ \mathcal{E}_{S_{\tau}} - \mathcal{E}_{S_{\tau_{_0}}} \leq  \int_{\tau_{_0}}^\tau \mathcal{E}_{S_{\sigma}} \d \sigma \, .\]
    Hence by Grönwall's inequality, we obtain,
    \[ \mathcal{E}_{\scri^+_{\tau_{_0}}} \lesssim \mathcal{E}_{\Sigma_{\tau_{_0}}} \, .\]
    \item {\bf Converse estimate.} In this direction, we simply use the foliation $\{ \Sigma_\tau \}_\tau$. From \eqref{EnId} and the control of the error terms, we have
    \[ \Big\vert \mathcal{E}_{\scri^+_{\tau_{_0}}} - \mathcal{E}_{\Sigma_{\tau_{_0}}} \Big\vert \leq \int_{\tau_{_0}}^T \mathcal{E}_{\Sigma_{\tau}} \d \tau \, ,\]
    whence
    \[  \mathcal{E}_{\Sigma_{\tau_{_0}}} \leq \mathcal{E}_{\scri^+_{\tau_{_0}}} + \int_{\tau_{_0}}^T \mathcal{E}_{\Sigma_{\tau}} \d \tau \, .\]
    We also have the intermediate estimates for $\tau_0 \leq \tau \leq T$
    \[  \mathcal{E}_{\Sigma_{\tau}} \leq \mathcal{E}_{\scri^+_{\tau_{_0}}} + \int_{\tau}^T \mathcal{E}_{\Sigma_{\sigma}} \d \sigma \, .\]
    Grönwall's estimate therefore gives
    \[ \mathcal{E}_{\Sigma_{\tau_{_0}}} \lesssim \mathcal{E}_{\scri^+_{\tau_{_0}}} \, .\]

\end{enumerate}

\subsection{Energy estimates in $U$}

On $U$, we choose our observer $\tau^a$ to be
\[ \tau^a = d_1 \hn^a + \Omega \big( d_0 \hl^a + d_2 \hm^a + \overline{d_2} \bhm^a \big) \]
where $d_0$, $d_1$ and $d_2$ are smooth functions on $U$ that agree with $c_0$, $c_1$ and $c_2$ at the intersection with $U^+$, are such that $\tau^a \partial_a = \partial_u$ on $U \cap U^0$ and $d_0$ and $d_1$ are positive on $U$. As we did on $U^+$, we choose a foliation transverse to $\scri^+$ whose normal vector field is given by
\[ \nu^a = \frac{1}{\sqrt{2}} (\hl^a + \hn^a ) \, .\]
The associated energy density on the slices is then
\begin{equation} \label{EDTFU}
\hat{J}_a \nu^a = \hat{\mathbf{T}}_{ab} \tau^a \nu^b = \frac{1}{\sqrt{2}} \big( d_1 \vert \hat\phi_2 \vert^2 + d_1 \vert \hat\phi_1 \vert^2 \big) + \mathcal{O}(\Omega ) .
\end{equation}
The error term has the form
\begin{align*}
\hat\nabla^{(a} \tau^{b)} \hat{\mathbf{T}}_{ab} &= \hat\nabla^{(a} (  \Omega ( d_0 \hl^{b)} + d_2 \hm^{b)} + \overline{d_2} \bhm^{b)} ) + d_1 \hn^{b)} ) \hat{\mathbf{T}}_{ab} \\
&= \underset{\mathrm{I}}{\underbrace{\Omega \hat\nabla^a ( d_0 \hl^{b} + d_2 \hm^{b} + \overline{d_2} \bhm^{b} ) \hat{\mathbf{T}}_{ab}}} + \underset{\mathrm{II}}{\underbrace{(\hat\nabla^{a} \Omega ) ( d_0 \hl^{b} + d_2 \hm^{b} + \overline{d_2} \bhm^{b} ) \hat{\mathbf{T}}_{ab} }} \\
&\phantom{=} + \underset{\mathrm{III}}{\underbrace{(\hat\nabla^{(a} d_1 ) \hn^{b)} \hat{\mathbf{T}}_{ab} }} + \underset{\mathrm{IV}}{\underbrace{d_1 \hat\nabla^{(a} \hn^{b)} ) \hat{\mathbf{T}}_{ab}}} \, .
\end{align*}
In order to establish that the error terms are controlled by the energy density on the slices, it is enough to simply show that no error term without a factor $\Omega$ involves either $\hat\phi_0$ or $\overline{\hat\phi}_{0}$. The term I is a quadratic form with bounded coefficients and an overall factor of $\Omega$, so it is controlled by the energy density \eqref{EDTFU}. To estimate II, we decompose the gradient of $\Omega$ along our Newman--Penrose tetrad,
\[ \hat\nabla^a \Omega = e_1 \hn^a + \Omega \big( e_1 \hl^a + e_2 \hm^a + \overline{e_2} \bhm^a \big) \, .\]
Since the only term in $\hat\nabla^a \Omega$ that does not have a factor of $\Omega$ involves the tetrad vector $\hn^a$, the terms from II without a factor $\Omega$ also do not contain $\hat\phi_0$. Therefore, II is controlled by the energy density. The term III involves only $\vert \hat\phi_1 \vert^2$ and $\vert \hat\phi_2 \vert^2$ with bounded coefficients and is therefore also controlled by \eqref{EDTFU}. The fourth term requires the most care. We decompose $\hat\nabla^a \hn^b$ along our NP tetrad,
\begin{equation} \label{gradhn}
\hat\nabla^a \hn^b = \hl^a \hat{\Delta} \hn^b + \hn^a \hat{D} \hn^b - \hm^a \bar{\hat{\delta}} \hn^b - \bhm^a \hat{\delta} \hn^b,
\end{equation}
and note the following transport equations along the NP tetrad for $\hat{n}^a$ ((4.5.28), \cite{spinorsandspacetime1}):
\begin{eqnarray}
\hat{D} \hn^b &=& -(\hat{\epsilon} + \bar{\hat{\epsilon}} ) \hn^b + \hat{\pi} \hm^b + \bar{\hat{\pi}} \bhm^b , \label{Line1} \\
\hat{\delta} \hn^b &=& -(\hat{\beta} + \bar{\hat{\alpha}} ) \hn^b + \hat{\mu} \hm^b + \bar{\hat{\lambda}} \bhm^b , \label{Line2} \\
\bar{\hat{\delta}} \hn^b &=& - ( \hat{\alpha} + \bar{\hat{\beta}} ) \hn^b + \hat{\lambda} \hm^b + \bar{\hat{\mu}} \bhm^b , \label{Line3} \\
\hat{\Delta} \hat{n}^b &=& - (\hat{\gamma} + \bar{\hat{\gamma}}) \hat{n}^b + \hat{\nu} \hat{m}^b + \bar{\hat{\nu}} \bar{\hat{m}}^b . \label{Line4}
\end{eqnarray}
The right-hand side of \eqref{gradhn} therefore involves no $\hl^a \hl^b$ term, so $\vert \hat\phi_0 \vert^2$ does not appear in term IV. However, any term involving $\hl^a \hm^b$ or $\hl^a \bhm^b$ will produce a product involving either $\hat\phi_0$ or $\overline{\hat\phi_0}$, and these terms must contain a factor of $\Omega$ in order to be controlled by the energy density. We therefore need to take a closer look at the contributions of the first, third and fourth terms in \eqref{gradhn}. 
The potentially dangerous terms are those involving $\bar{\hat{\lambda}}$ in \eqref{Line2}, $\hat{\lambda}$ in \eqref{Line3}, and $\hat{\nu}$ and $\bar{\hat{\nu}}$ in \eqref{Line4}. But both $\hat{\lambda}$ and $\hat{\nu}$ vanish on $\scri^+$, $\hat{\lambda} \approx 0 \approx \hat{\nu}$, as a consequence of the asymptotic Einstein condition \eqref{asymptoticEinstein} (see \eqref{geodetictwistfreescri}, \eqref{shearfreescri}); indeed, the vanishing of the spin coefficients $\hat{\nu}$ and $\hat{\lambda}$ on $\scri^+$ simply restates the fact that $\scri^+$ is geodetic and shear-free. Therefore these spin coefficients decay like $\mathcal{O}(\Omega)$ towards $\scri^+$, and it follows that the error term is controlled by the energy density on the slices. Then energy estimates on $U$ can be obtained by means of Grönwall's estimates in much the same way as on $U^+$.

%\bibliographystyle{siam}
%\bibliography{bibliography}

\printbibliography

@article{HafnerBesset2021,
author={Besset, N. and H\"afner, D.},
title={Existence of exponentially growing finite energy solutions for the charged {K}lein--{G}ordon equation on the de {S}itter--{K}err--{N}ewman metric},
journal={J. Hyperbolic Differ. Equ.},
volume={18},
number={2},
pages={293--310},
year={2021}
}

@article{GeorgescuGerardHafner2013,
author={Georgescu, V. and G\'erard, C. and H\"afner, D.},
title={Boundary values of resolvents of self-adjoint operators in {K}rein spaces},
journal={hal-0074818v3},
year={2013}
}

@article{GeorgescuGerardHafner2015,
author={Georgescu, V. and G\'erard, C. and H\"afner, D.},
title={Resolvent and propagation estimates for {K}lein--{G}ordon equations with non-positive energy},
journal={Journal of Spectral Theory},
year={2015},
volume={5},
issue={1},
pages={113--192}
}

@article{GeorgescuGerardHafner2014,
author={Georgescu, V. and G\'erard, C. and H\"afner, D.},
title={Asymptotic completeness for superradiant {K}lein--{G}ordon equations and applications to the de {S}itter--{K}err metric},
journal={Journal of the European Mathematical Society},
volume={19},
number={8},
year={2014},
pages={2371--2444}
}

@article{AngelopoulosAretakisGajic2020,
author={Angelopoulos, Y. and Aretakis, S. and Gajic, D.},
title={A non-degenerate scattering theory for the wave equation on extremal {R}eissner--{N}ordstr\"om},
journal={Communications in Mathematical Physics},
volume={380},
pages={323--408},
year={2020}
}

@article{DimockKay1987,
author={Dimock, J. and Kay, B.S.},
title={Classical and quantum scattering theory for linear scalar fields on the {S}chwarzschild metric I},
journal={Annals of Physics},
volume={175},
number={2},
pages={366--426},
year={1987},
doi={https://doi.org/10.1016/0003-4916(87)90214-4}
}

@article{Alford2020,
author={Alford, F.},
title={The Scattering Map on {O}ppenheimer--{S}nyder Space-Time},
journal={Annales Henri Poincar\'e},
volume={21},
number={6},
pages={2031--2092},
year={2020}
}

@article{Besset2021,
author={Besset, N.},
title={Scattering Theory for the Charged {K}lein–{G}ordon Equation in the Exterior De {S}itter-–{R}eissner–{N}ordstr\"om Spacetime},
year={2021},
journal={Journal of Geometric Analysis},
volume={31},
number={11},
pages={10521--10585}
}

@article{Bachelot1997,
author={Bachelot, A.},
title={Scattering of scalar fields by spherical gravitational collapse},
year={1997},
journal={Journal de Mathématiques Pures et Appliquées},
volume={76},
number={2},
pages={155--210}
}

@article{DambrosioDipierro2013,
author={D'{A}mbrosio, L. and Dipierro, S.},
year={2014},
title={Hardy inequalities on {R}iemannian manifolds and applications},
journal={Ann. Inst. H. Poincar\'e Anal. Non Lin\'eaire},
volume={31},
issue={3},
pages={449--475},
url={https://doi.org/10.1016/j.anihpc.2013.04.004}
}

@book{AdamsFournier2003,
author={Adams, R. and Fournier, J. F.},
year={2003},
title={Sobolev Spaces},
publisher={Academic Press},
isbn={9780120441433},
url={https://www.elsevier.com/books/sobolev-spaces/adams/978-0-12-044143-3},
}

@article{MohamedKroon2021,
title={A comparison of {A}shtekar's and {F}riedrich's formalisms of spatial infinity},
author={Mohamed, M. M. A. and Valiente Kroon, J. A.},
year={2021},
journal={Classical and Quantum Gravity},
volume={38},
issue={16},
pages={165015},
url={https://iopscience.iop.org/article/10.1088/1361-6382/ac1208/meta}
}

@article{Friedrich1998a,
title={Gravitational fields near space-like and null infinity},
author={Friedrich, Helmut},
year={1998},
journal={Journal of Geometry and Physics},
volume={24},
issue={2},
pages={83--163},
url={https://doi.org/10.1016/S0393-0440(97)82168-7}
}

@article{Masaood2020,
title={A Scattering Theory for Linearised Gravity
on the Exterior of the {S}chwarzschild Black Hole {I}: The {T}eukolsky Equations},
author={Masaood, H.},
year={2022},
journal={Commun. Math. Phys.},
volume={393},
pages={477--581},
url={https://doi.org/10.1007/s00220-022-04372-3}
}

@article{HintzVasy2020,
title={Stability of Minkowski space and polyhomogeneity of the metric},
author={Hintz, P. and Vasy, A.},
year={2020},
journal={Annals of PDE},
number={2},
volume={6},
url={https://doi.org/10.1007/s40818-020-0077-0}
}

@article{AshtekarHansen1978,
title={A unified treatment of null and spatial infinity in general relativity. I. Universal structure, asymptotic symmetries, and conserved quantities at spatial infinity},
author={Ashtekar, A. and Hansen, R. O.},
year={1978},
journal={J. Math. Phys.},
volume={19},
pages={1542},
url={https://doi.org/10.1063/1.523863}
}

@article{Pham2020b,
title={Conformal scattering theory for the linearized gravity fields on {S}chwarzschild spacetime},
author={Pham, Truong Xuan},
year={2021},
journal={Annals of Global Analysis and Geometry},
volume={60},
pages={589--608},
url={
https://doi.org/10.1007/s10455-021-09789-y}
}

@article{Pham2020a,
title={Conformal scattering theories for tensorial wave equations on {S}chwarzschild spacetime},
author={Pham, Truong Xuan},
year={2020},
journal={arXiv:2006.02888},
url={
https://doi.org/10.48550/arXiv.2006.02888}
}

@article{Mokdad2022,
title={Conformal scattering and the {G}oursat problem for {D}irac fields in the interior of charged spherically symmetric black holes},
author={Mokdad, Mokdad},
year={2022},
journal={Reviews in Mathematical Physics},
volume={34},
number={1},
url={https://doi.org/10.1142/S0129055X21500379}
}

@article{Pham2022,
title={Conformal Scattering Theory for the {D}irac Equation on {K}err Spacetime},
author={Pham, Truong Xuan},
year={2022},
journal={Annales Henri Poincar\'e},
volume={23},
pages={3053--3091},
url={https://link.springer.com/article/10.1007/s00023-022-01155-3}
}

@article{Gajic2022,
title={Late-time asymptotics for geometric wave equations with inverse-square potentials},
author={Gajic, Dejan},
year={2023},
journal={Journal of Functional Analysis},
volume={285},
issue={7},
pages={110058},
url={https://doi.org/10.1016/j.jfa.2023.110058}
}

@article{BarWafo2015,
title={Initial Value Problems for Wave Equations on Manifolds},
author={B\"ar, Christian and Wafo, R. T.},
year={2015},
journal={Math. Phys. Anal. Geom.},
volume={18},
number={7},
url={https://doi.org/10.1007/s11040-015-9176-7}
}

@article{AshtekarBongaKesavan2014,
title={Asymptotics with a positive cosmological constant: {I}. Basic framework},
author={Ashtekar, A. and Bonga, B. and Kesavan, A.},
year={2014},
journal={Classical and Quantum Gravity},
volume={32},
number={2},
pages={025004},
url={https://doi.org/10.1088/0264-9381/32/2/025004},
}

@article{NicolasPham2018,
author={Nicolas, J.-P. and Pham, T. X.},
year={2019},
title={Peeling on {K}err Spacetime: Linear and Semi-linear Scalar Fields},
journal={Ann. Henri Poincar\'e},
volume={20},
pages={3419--3470},
url={https://doi.org/10.1007/s00023-019-00832-0},
}

@book{EspositoWitten1977,
author={Esposito, F.P. and Witten, L.},
year={1977},
title={Asymptotic Structure of Space-Time},
publisher={Springer},
address={Boston, MA},
isbn={978-1-4684-2345-7},
url={https://doi.org/10.1007/978-1-4684-2343-3_1},
}

@article{Joudioux2019,
author={Joudioux, J.},
year={2020},
title={{H}\"ormander's method for the characteristic {C}auchy problem and conformal scattering for a nonlinear wave equation},
journal={Letters in Mathematical Physics},
volume={110},
pages={1391--1423},
url={https://link.springer.com/article/10.1007/s11005-020-01266-0},
}

@article{Friedlander1967,
author = {Friedlander, F. G.},
journal = {Proceedings of the Royal Society of London. Series A, Mathematical and Physical Sciences},
pages = {264--278},
publisher = {The Royal Society},
title = {On the Radiation Field of Pulse Solutions of the Wave Equation {III}},
volume = {299},
year = {1967},
url={https://royalsocietypublishing.org/doi/10.1098/rspa.1967.0134},
}

@article{Friedlander1964,
author = {Friedlander, F. G.},
journal = {Proceedings of the Royal Society of London. Series A, Mathematical and Physical Sciences},
pages = {386--394},
publisher = {The Royal Society},
title = {On the Radiation Field of Pulse Solutions of the Wave Equation {II}},
volume = {279},
year = {1964},
url={https://royalsocietypublishing.org/doi/10.1098/rspa.1964.0111},
}

@article{Friedlander1962,
author = {Friedlander, F. G.},
journal = {Proceedings of the Royal Society of London. Series A, Mathematical and Physical Sciences},
number = {1336},
pages = {53--65},
publisher = {The Royal Society},
title = {On the Radiation Field of Pulse Solutions of the Wave Equation},
volume = {269},
year = {1962},
url={https://royalsocietypublishing.org/doi/abs/10.1098/rspa.1962.0162},
}

@book{Strocchi2013,
author={Strocchi, F.},
year={2013},
title={An Introduction to Non-Perturbative Foundations of Quantum Field Theory},
series={International Series of Monographs on Physics},
publisher={Oxford University Press},
isbn={978-0-19-878923-9},
}

@article{Taujanskas2018,
author={Taujanskas, Grigalius},
year={2019},
title={Conformal scattering of the {M}axwell-scalar field system on de {S}itter space},
journal={J. Hyperbolic Differ. Equ.},
volume = {16},
issue = {4},
pages = {743-791},
url={https://doi.org/10.1142/S021989161950019X},
}

@book{Leray1953,
author={Leray, Jean},
year={1953},
title={{H}yperbolic {D}ifferential {E}quations},
publisher={unpublished},
url={https://hdl.handle.net/20.500.12111/8019}
}

@article{Friedrich1986,
author={Friedrich,Helmut},
year={1986},
title={Existence and structure of past asymptotically simple solutions of {E}instein's field equations with positive cosmological constant},
journal={Journal of Geometry and Physics},
volume={3},
number={1},
pages={101-117},
isbn={0393-0440},
url={https://www.sciencedirect.com/science/article/pii/0393044086900045},
}

@article{KehleShlapentokhRothman2018,
author={Kehle,Christoph and Shlapentokh-Rothman,Yakov},
year={2018},
title={A scattering theory for linear waves on the interior of {R}eissner--{N}ordstr\"om black holes},
journal={arXiv:1804.05438},
url={https://arxiv.org/abs/1804.05438},
}

@article{DafermosRodnianskiShlapentokhRothman2014,
author={Dafermos,Mihalis and Rodnianski,Igor and Shlapentokh-Rothman,Yakov},
year={2018},
title={A scattering theory for the wave equation on {K}err black hole exteriors},
journal={Annales Scientifiques de L'\'Ecole Normale Sup\'erieure},
volume={51},
number={2},
pages={371--486},
url={https://doi.org/10.24033/asens.2358},
}

@article{Nicolas2013,
author={Nicolas,Jean-Philippe},
year={2016},
title={Conformal scattering on the {S}chwarzschild metric},
journal={Annales de l'Institut Fourier},
volume={66},
number={3},
pages={1175--1216},
url={http://www.numdam.org/item/AIF_2016__66_3_1175_0/},
}

@article{Mokdad2017,
author={Mokdad,M.},
year={2019},
title={Conformal Scattering of {M}axwell Fields on {R}eissner--{N}ordstr\"om--de {S}itter black hole spacetimes},
journal={Annales de l'Institut Fourier},
volume={69},
number={5},
pages={2291--2329},
url={http://aif.centre-mersenne.org/item/AIF_2019__69_5_2291_0},
}

@article{MasonNicolas2007,
author={Mason,Lionel J. and Nicolas,Jean-Philippe},
year={2009},
title={Regularity at spacelike and null infinity},
journal={Journal of the Institute of Mathematics of Jussieu},
volume={8},
number={1},
pages={179--208},
url={https://doi.org/10.1017/S1474748008000297},
}

@article{Joudioux2010,
author={Joudioux,J.},
year={2012},
title={Conformal scattering for a nonlinear wave equation on a curved background},
journal={Journal of Hyperbolic Differential Equations},
volume={9},
number={1},
pages={1--65},
url={https://doi.org/10.1142/S0219891612500014},
}

@article{Bachelot1999,
author={Bachelot,Alain},
year={1999},
title={The {H}awking effect},
journal={{A}nn. {I}nst. {H}enri {P}oincar\'e, Physique th\'eorique},
volume={70},
number={1},
pages={41-99},
url={http://www.numdam.org/article/AIHPA_1999__70_1_41_0.pdf},
}

@article{DimockKay1986,
author={Dimock,J. and Kay,B. S.},
year={1986},
title={Scattering for massive scalar fields on {C}oulomb potentials and {S}chwarzschild metrics},
journal={Classical and Quantum Gravity},
volume={3},
number={1},
pages={71},
isbn={0264-9381},
url={http://stacks.iop.org/0264-9381/3/i=1/a=010},
}

@article{Dimock1985,
author={Dimock,J.},
year={1985},
month={04/01},
title={Scattering for the wave equation on the {S}chwarzschild metric},
journal={General Relativity and Gravitation},
volume={17},
number={4},
pages={353-369},
isbn={1572-9532},
url={https://doi.org/10.1007/BF00759679},
}

@article{Bachelot1994,
author={Bachelot,Alain},
year={1994},
title={Asymptotic completeness for the {K}lein--{G}ordon equation on the {S}chwarzschild metric},
journal={{A}nn. {I}nst. {H}enri {P}oincar\'e, Physique th\'eorique},
volume={61},
number={4},
pages={411-441},
url={http://www.numdam.org/article/AIHPA_1994__61_4_411_0.pdf},
}

@article{Bachelot1991,
author={Bachelot,Alain},
year={1991},
title={Gravitational scattering of electromagnetic field  by a {S}chwarzschild black hole},
journal={{A}nn. {I}nst. {H}enri {P}oincar\'e, Physique th\'eorique},
volume={54},
number={3},
pages={261-320},
url={http://www.numdam.org/article/AIHPA_1991__54_3_261_0.pdf},
}

@book{LaxPhillips1967,
author={Lax,Peter D. and Phillips,Ralph S.},
year={1967},
title={Scattering theory},
publisher={Academic Press},
address={New York and London},
pages={276},
language={English},
}

@article{LaxPhillips1964,
author={Lax,Peter D. and Phillips,Ralph S.},
year={1964},
month={01},
title={Scattering theory},
journal={Bull. Amer. Math. Soc.},
volume={70},
number={1},
pages={130-142},
isbn={0002-9904},
language={English},
url={https://projecteuclid.org:443/euclid.bams/1183525789},
}

@article{Friedlander1980,
author={Friedlander,F. G.},
year={1980},
title={Radiation fields and hyperbolic scattering theory},
journal={Mathematical Proceedings of the Cambridge Philosophical Society},
volume={88},
number={3},
pages={483-515},
isbn={0305-0041},
url={https://doi.org/10.1017/S0305004100057819},
}

@article{Penrose1965,
author={Penrose,Roger},
year={1965},
title={Zero rest-mass fields including gravitation: asymptotic behaviour},
journal={Proceedings of the Royal Society of London. Series A, Mathematical and Physical Sciences},
volume={284},
number={1397},
pages={159-203},
url={http://rspa.royalsocietypublishing.org/content/284/1397/159.abstract},
}

@article{Penrose1963,
author={Penrose,Roger},
year={1963},
month={01/15},
title={Asymptotic Properties of Fields and Space-Times},
journal={Physical Review Letters},
volume={10},
number={2},
pages={66-68},
url={https://link.aps.org/doi/10.1103/PhysRevLett.10.66},
}

@book{spinorsandspacetime2,
author={Penrose,Roger and Rindler,Wolfgang},
year={1986},
title={Spinors and space-time},
publisher={Cambridge University Press},
volume={2},
isbn={0521347866},
}

@book{spinorsandspacetime1,
author={Penrose,Roger and Rindler,Wolfgang},
year={1984},
title={Spinors and space-time},
publisher={Cambridge University Press},
volume={1},
isbn={0521337070},
}

@article{BaezSegalZhou1990,
author={Baez,John C. and Segal,Irving E. and Zhou,Zheng-Fang},
year={1990},
title={The global {G}oursat problem and scattering for nonlinear wave equations},
journal={Journal of Functional Analysis},
volume={93},
number={2},
pages={239-269},
isbn={0022-1236},
url={https://www.sciencedirect.com/science/article/pii/0022123690901288},
}

@article{Hormander1990,
author={H\"ormander, Lars},
year={1990},
title={A remark on the characteristic {C}auchy problem},
journal={Journal of Functional Analysis},
volume={93},
number={2},
pages={270-277},
isbn={0022-1236},
url={https://www.sciencedirect.com/science/article/pii/0022123690901299},
}

@article{Morawetz1961,
author={Morawetz,Cathleen S.},
year={1961},
title={The decay of solutions of the exterior initial-boundary value problem for the wave equation},
journal={Communications on Pure and Applied Mathematics},
volume={14},
number={3},
pages={561-568},
url={https://onlinelibrary.wiley.com/doi/abs/10.1002/cpa.3160140327},
}

@inbook{Ashtekar2014,
    author = {Ashtekar, Abhay},
    title = {Geometry and Physics of Null Infinity},
    publisher={International Press of Boston},
    year = {2015},
    chapter = {Surveys in differential geometry --- One hundred years of general relativity, L. Bieri and S. T. Yau eds.},
    url = {https://arxiv.org/abs/1409.1800}
}

@article{ChruscielDelay2003,
author={Chru\'sciel,Piotr T. and Delay,Erwann},
year={2003},
title={On mapping properties of the general relativistic constraints operator in weighted function spaces, with applications},
journal={Mem. Soc. Math. France},
volume={94},
pages={1-103},
url={https://doi.org/10.24033/msmf.407},
}

@article{ChruscielDelay2002,
author={Chru\'sciel,Piotr T. and Delay,Erwann},
year={2002},
title={Existence of non-trivial, vacuum, asymptotically simple spacetimes},
journal={Classical and Quantum Gravity},
volume={19},
number={12},
pages={3389},
url={http://stacks.iop.org/0264-9381/19/i=12/a=501},
}

@article{Corvino2000,
author={Corvino,Justin},
year={2000},
title={Scalar Curvature Deformation and a Gluing Construction for the {E}instein Constraint Equations},
journal={Communications in Mathematical Physics},
volume={214},
number={1},
pages={137-189},
url={http://dx.doi.org/10.1007/PL00005533},
}

@article{CorvinoSchoen2006,
author={Corvino,Justin and Schoen,Richard M.},
year={2006},
title={On the Asymptotics for the Vacuum {E}instein Constraint Equations},
journal={J. Differential Geom.},
volume={73},
number={2},
pages={185-217},
url={http://projecteuclid.org/euclid.jdg/1146169910},
}

@article{MasonNicolas2004,
author={Mason,Lionel J. and Nicolas,Jean-Philippe},
year={2004},
title={Conformal Scattering and the {G}oursat Problem},
journal={Journal of Hyperbolic Differential Equations},
volume={01},
number={02},
pages={197-233},
url={http://dx.doi.org/10.1142/S0219891604000123},
}
\end{document}